\documentclass[11pt]{article} 
\usepackage{amsmath,amssymb}
\usepackage{amsthm}
\usepackage{amsfonts}
\usepackage{natbib}

\usepackage{graphicx}
\usepackage{relsize}
\usepackage{catchfilebetweentags} 
\usepackage{layout} 
\usepackage{bm} 

\usepackage{caption} 
\usepackage{float}

\usepackage{multirow} 
\usepackage{multicol} 

\usepackage{bigints} 

\usepackage[margin=1in]{geometry}
\allowdisplaybreaks  

\usepackage[onehalfspacing]{setspace}

\usepackage{etoolbox}
\makeatletter
\patchcmd{\CatchFBT@Fin@l}{\endlinechar\m@ne}{}
  {}{\typeout{Unsuccessful etoolbox patch!}}
\makeatother

\usepackage{mathtools}

\usepackage{verbatim} 

\usepackage[shortlabels]{enumitem} 

\newtheorem{theorem}{Theorem}

\newtheorem{remark}{Remark}

\newtheorem{lemma}{Lemma}

\newenvironment{nopgbreak}
  {\par\nobreak\vfil\penalty0\vfilneg
   \vtop\bgroup}
  {\par\xdef\tpd{\the\prevdepth}\egroup
   \prevdepth=\tpd}

\def \argmin{\mathop{\hbox{\rm argmin}}}

\def \diag{\mathop{\hbox{\rm diag}}}

\newcommand{\m}[1]{\mathcal{#1}}

\newcommand{\bo}[1]{\boldsymbol{#1}}
\newcommand{\boup}[1]{\boldsymbol{\mathrm{#1}}}

\newcommand{\ha}[1]{\widehat{#1}}
\newcommand{\ti}[1]{\widetilde{#1}}

\newcommand{\eps}{\varepsilon}

\newcommand{\mme}{\mathbb{E}}
\newcommand{\mmi}{\mathbb{I}}
\newcommand{\mmp}{\mathbb{P}}
\newcommand{\mmr}{\mathbb{R}}
\newcommand{\mmv}{\mathbb{V}}
\newcommand{\mmy}{\mathbb{Y}}
\newcommand{\mmz}{\mathbb{Z}}

\newcommand{\ba}{\boup{a}}
\newcommand{\bc}{\boup{c}}
\newcommand{\be}{\boup{e}}
\newcommand{\bmm}{\boup{m}}

\newcommand{\bu}{\boup{u}}
\newcommand{\bx}{\boup{x}}

\newcommand{\balpha}{\boup{\alpha}}
\newcommand{\beps}{\boup{\varepsilon}}
\newcommand{\bbeta}{\boup{\eta}}
\newcommand{\bphi}{\boup{\varphi}}
\newcommand{\bgam}{\boup{\gamma}}
\newcommand{\boeta}{\boup{\eta}}

\newcommand{\btheta}{\boup{\theta}}
\newcommand{\bzeta}{\boup{\zeta}}

\newcommand{\bB}{\boup{B}}
\newcommand{\bE}{\boup{E}}
\newcommand{\bG}{\boup{G}}
\newcommand{\bHt}{\ti{\boup{H}}}
\newcommand{\bH}{\boup{H}}
\newcommand{\bI}{\boup{I}}
\newcommand{\bJ}{\boup{J}}
\newcommand{\bQ}{\boup{Q}}

\newcommand{\bWt}{\widetilde{\boup{W}}}
\newcommand{\bZ}{\boup{Z}}
\newcommand{\bY}{\boup{Y}}

\newcommand{\bzero}{\boup{0}}

\newcommand{\bmB}{\boup{\m{B}}}
\newcommand{\bmV}{\boup{\m{V}}}
\newcommand{\bmW}{\boup{\m{W}}}

\newcommand{\bGam}{\boup{\Gamma}}

\newcommand{\bOmg}{\boup{\Omega}}

\newcommand{\loadaspt}[1]{ \ExecuteMetaData[assumptions.tex]{assu#1} }

\newcommand{\dto}{\overset{d}{\to}} 
\newcommand{\pto}{\overset{p}{\to}} 

\graphicspath{{results_paper/}} 

\makeatletter
\def\input@path{{results_paper/}}
\makeatother

\begin{document}

\title{
\renewcommand{\thefootnote}{\alph{footnote}}
Regression Discontinuity Design with Many Thresholds
}

\author{
\renewcommand{\thefootnote}{\alph{footnote}}
Marinho Bertanha\footnotemark[1]
}

\date{
This version: September 16, 2019
\\ 
First version: November 7, 2014
}


{
\renewcommand{\thefootnote}{\alph{footnote}}
%

\footnotetext[1]{
Gilbert F. Schaefer Assistant Professor, Department of Economics, University of Notre Dame;
  Address: 3060 Jenkins Nanovic Halls, Notre Dame, IN 46556.
  Email: mbertanha@nd.edu. 
  Website: www.nd.edu/$\sim$mbertanh.
  }
}

\maketitle


\normalsize
\begin{abstract}
Numerous empirical studies employ regression discontinuity designs with
multiple cutoffs and heterogeneous treatments.
A common practice is to normalize all
the cutoffs to zero and estimate one effect.
This procedure identifies the average treatment effect (ATE) on
the observed distribution of individuals local to existing cutoffs.
However, researchers often want to make inferences on more meaningful ATEs,
computed over general counterfactual distributions of individuals, rather than simply the observed distribution
of individuals local to existing cutoffs.
This paper proposes a consistent and asymptotically normal estimator for such ATEs
when heterogeneity follows a non-parametric function of cutoff characteristics in the sharp case.
The proposed estimator converges at the minimax optimal rate of root-$n$ for a specific choice of tuning parameters.
Identification in the fuzzy case, with multiple cutoffs, is impossible unless
heterogeneity follows a finite-dimensional function of cutoff characteristics.
Under parametric heterogeneity,
this paper proposes an ATE estimator for the fuzzy case that optimally combines
observations to maximize its precision.
\end{abstract}

\normalsize




\newpage

\section{Introduction}

\indent

Applications of regression discontinuity design (RDD) have become increasingly popular in economics since the late
1990s (\cite{black1999better}, \cite{angrist1999}, and \cite{van2002}).
One of
RDD's main advantages is identification of a local causal effect under minimal functional form assumptions. More recently,
with  increasing availability of richer data sets, there have been many applications with multiple cutoffs and treatments
(for example, \cite{black2007}, \cite{egger2010}, \cite{delamata2012}, \cite{pop2013going}).
Existing one-cutoff RDD methods applied to each individual cutoff produce
many local effects that are estimated using only a few observations near each cutoff.
Researchers often prefer one takeaway summary effect that
is more precisely estimated by pooling all the data. 
The meaning
of a summary effect crucially depends 
on heterogeneity assumptions and weights imposed on the different
local effects.

Applied studies with multiple cutoffs often 
normalize all cutoffs to zero and use the one-cutoff estimator.
This normalization procedure
estimates an average of local treatment effects weighted by the relative density of individuals near
each of the cutoffs (\cite{cattaneo2016keele}, Proposition 3).
Such an average effect would be a meaningful summary measure only in two cases:
(i) local treatment effects are all identical and the weighting scheme does not matter; or
(ii) local treatment effects are heterogeneous but the researcher is only interested in the average effect
on the individuals near the existing cutoffs.
However, researchers are often interested in combining observed data with assumptions weaker than (i)
to make inferences on counterfactual scenarios more general than (ii).\footnote{
In a RDD setting with multiple cutoffs and treatments, it is unreasonable to expect
that different local treatment effects are always identical. For example, \cite{pop2013going} find that
the impact of going to a better high school on academic achievement is heterogeneous
across students with different ability levels. Another example is \cite{delamata2012}, who
finds that the eligibility for Medicaid benefits
decreases the probability of having private
health insurance more strongly for lower income individuals.
Although I allow for heterogeneous effects across cutoffs,
counterfactual analysis requires a pooling and a policy invariance assumption (Section \ref{sec_setup}).
}

This paper proposes a novel estimation procedure for average treatment effects (ATE).
These ATEs are more valuable summary measures than the average effect estimated by the normalization procedure described above for two reasons.
First, the researcher explicitly chooses the counterfactual distribution of the ATE, and this distribution
may include individuals at or between existing cutoffs.
Second, the researcher does not need to assume any specific functional
form for the heterogeneity of treatment effects across different cutoffs.
As an example of an application,
suppose we are interested in estimating the effect of Medicaid benefits on health care utilization.
Medicaid eligibility is triggered by income cutoffs that vary across states.
Existing one-cutoff RDD methods identify the average effect on individuals with income equal to the income
 cutoffs. 
 However, most interesting policy questions require the average effect over the entire range of income values in the data.

The framework for RDD with many thresholds is introduced here using a simple example based on the work of \cite{pop2013going}, PU from now on.
Using a wealth of variation of cutoffs from high school assignments in Romania,
PU provide rigorous evidence of the impacts of attending  a better school on students' academic performance.
The economic logic of this application is briefly summarized as follows.
A central planner assigns students to high schools based on their scores from a placement test.
High schools have limited capacities and are ranked by their qualities.
The central planner ranks students by their scores and assigns each of them to the best school available.
Each student $i$ submits her score $X_i$ (forcing variable) to the central planner who,
based on the entire distribution of scores,
determines a minimum test score $c_j$ (cutoff)
for admission to each high school $j$.
The quality of high school $j$ is denoted $d_j$ (treatment dose).

The RDD assignment is assumed sharp for now.
That is, students attend the best high school available to them based on their score and the cutoffs that apply to them.
As the test score crosses an admission threshold $c_j$, the quality
of the school the student attends changes from
$d_{j-1}$ to $d_j$.
Local average effects are denoted
by $\mme[Y_i(d_j)-Y_i(d_{j-1})|X_i=c_j] = \beta(c_j,d_{j-1},d_j)$, where
$Y_i(d)$ is the potential academic achievement student $i$ has if attending a high school of quality $d$,
and $\beta(c,d,d')$ is the treatment effect function.
Heterogeneity of local effects comes from values of cutoffs and
treatment doses that change across the different cutoffs.
PU give a particularly illustrative application,
because it exhibits sufficient variation in cutoff and treatment doses to generate ATEs with substantially greater 
economic relevance than the typical average based on normalizing all of the cutoffs to zero.

Numerous other examples of RDD with multiple cutoffs and treatments exist in different fields of economics.
For instance, \cite{egger2010} study the effect of the size of city government councils on municipal
expenditures, where council size is determined by population cutoffs. \cite{delamata2012} estimates
the effects of Medicaid benefits on health care utilization, where Medicaid eligibility is triggered by
income cutoffs that vary across states. \cite{agarwal2016} and \cite{degiorgi2017} look at
multiple cutoffs on credit scores, used by banks to make credit decisions.
Education economics also provides a variety of applications.
\cite{angrist1999} and \cite{hoxby2000} use class size rules to estimate the impact of class size on student achievement.
\cite{hoxby2000} utilizes variation in cutoff values from specific school district class size rules.
Several researchers exploit different school starting dates to estimate the impact of educational attainment on various outcomes, for example,
\cite{dobkin2010school}, and \cite{mccrary2011}.
\cite{duflo2011} analyze school cohorts that are split into low and high-achieving classes based on test scores, where each school has its own cutoff score.
\cite{garibaldi2012} look at different income cutoffs that determine tuition subsidies
to study the impact of tuition payment on the probability of late graduation from university.
In short, despite many applications with variation in cutoffs and treatment doses, 
a lack of theory on how to combine observations from all cutoffs impedes our ability to estimate economically-relevant average effects.

Whether local effects can be combined into an average effect depends on how comparable the researcher believes these effects are.
The comparability of local treatment effects essentially depends
on the heterogeneity of treatment doses and on the heterogeneity of the treatment effect function
$\beta(c,d,d')$.  
This paper considers two types of assumptions regarding these two aspects of heterogeneity.
The first heterogeneity assumption says that treatment doses
are credibly quantifiable by some variable $d$.
For example, PU
find behavioral evidence that average student performance at each school
is a good summary measure for school quality.
Another example is the case of a single treatment being triggered by varying cutoffs, 
as when each state has its own income threshold for Medicaid coverage.
The second heterogeneity assumption specifies a parametric functional form for
$\beta(c,d,d')$ guided by economic theory or \textit{a priori} knowledge of the researcher.
For example, in a class size application like  Hoxby's \citeyearpar{hoxby2000},
a functional form based on Lazear's  \citeyearpar{lazear2001} model of achievement can be derived as a function of class size.
Another example is given by \cite{bajari2017} who present a principal-agent model to study
how insurers reimburse hospitals. 
The marginal reimbursement rate is discontinuous on health expenditures.

This paper proposes a consistent and asymptotically normal estimator for the ATE
of a counterfactual distribution of treatment assignments specified by the researcher.
A counterfactual policy scenario specifies the distribution of $(c,d,d')$, and the
ATE is the integral of $\beta(c,d,d')$ weighted by such a distribution.
The ability to predict effects of counterfactual policies depends crucially
on assuming that the distribution of potential outcomes $Y_i(d)$ does not depend on the initial schedule of cutoff-dose values. 
This policy invariance assumption,
along with 
the first heterogeneity assumption,
allows the researcher to choose counterfactual distributions
with support more general than the discrete set of cutoff-dose values observed in the data.

The estimator proposed in this paper approximates the ATE integral by
averaging estimates of $\beta(c,d,d')$ at existing cutoffs using a proper weighting scheme.
Under the first heterogeneity assumption with $\beta(c,d,d')$ non-parametric,
the proposed ATE estimator is  shown to be consistent and asymptotically normal. 
This result is novel, 
because estimation of the non-parametric function $\beta(c,d,d')$ is only possible at deterministic points of the domain,
and that creates an additional source of bias.
Asymptotic normality requires both the number of observations and cutoffs to
grow to infinity, and I provide sufficient conditions on their rate of growth. 
I demonstrate that the minimax rate of ATE estimation in this setting is root-$n$,
 and that the proposed estimator attains the minimax optimal rate for a specific choice of tuning parameters. 
This extends the previous literature on minimax optimality of non-parametric estimation of regression functions at a boundary point 
to estimation of averages of these regression functions. 

Many applications of RDD with multiple cutoffs are, in fact, fuzzy rather than sharp.
In the high school assignment example,
a student may choose to attend a high school other than the school she is originally eligible to attend.
Multiple treatments result in multiple compliance behaviors, and one-cutoff identification results
do not apply. 
Building on classic definitions of compliance behaviors (\cite{imbens1997}),
I define compliance groups in terms of \textit{changes} in treatment eligibility and receipt.
``Ever-compliers'' are those whose treatment received changes if and only if it changes to the treatment dose for which they become eligible for.
I  assume that individuals never change into a treatment dose different from the dose of eligibility, a ``no-defiance'' condition.
In the high school example,
if the test score of a student currently in school B increases so as to grant her access to school A,
no-defiance implies she either chooses to attend school A or stay at school B,
and that she is not triggered to attend some other school C.

This paper shows that even local identification in fuzzy RDD with finite multiple treatments is impossible
unless the class of treatment effect functions of ever-compliers is restricted to a finite-dimensional class.
Important empirical analyses of fuzzy RDD with multiple treatments include those of
 \cite{angrist1999}, \cite{chen2008vanderklaauw},  and \cite{hoekstra2009effect};
nevertheless, this is the first paper to define compliance and
study causal identification in a general framework for multi-cutoff fuzzy RDD.
This framework lays out conditions for the interpretation of
two-stage least squares (2SLS) estimates in applications of multi-cutoff fuzzy RDD,
a common practice in applied work.
The second heterogeneity assumption states that the treatment effect function
is of a parametric class. This assumption allows for consistent and asymptotically normal
estimation of ATEs on ever-compliers.
It also results in efficiency gains, because observations are optimally combined across cutoffs to minimize the mean squared error (MSE) of the ATE estimator.

The rapid growth in the number of applications of RDD in economics in the late 1990s
was accompanied by substantial theoretical contributions for inference
in the one-cutoff case. Identification and estimation in the sharp and fuzzy cases were formalized by \cite{hahn2001id}.
 \cite{fan1996} and \cite{porter2003} demonstrated
low-order bias and rate optimality of the local polynomial estimator.
Recent theoretical contributions have addressed the optimal bandwidth choice (\cite{imbens2012optimal}),
alternative asymptotic approximations with better finite sample properties (\cite{cattaneo2014calonico}),
quantile treatment effects (\cite{frandsen2012}), kink treatment effects (\cite{dong_jumpkink}),
and the difficulty of uniform inference (\cite{bertanha2019}).

The contribution of this paper is more closely related to the study of treatment effect extrapolation of \cite{angrist2004}, \cite{bertanha_imbens},
\cite{donglewbel2015}, \cite{angrist2013away}, and \cite{rokkanen2014jmp}.
These last two authors use observations on additional covariates.
They restrict the relationship between
the heterogeneity of treatment effects after conditioning on these covariates to obtain identification away from the cutoff.
This paper differs from these other contributions,
because the variation of multiple cutoffs and doses identify ATEs over distributions of individuals 
both between and at cutoffs,
without additional covariates.

The remainder of this paper is organized as follows.
Section \ref{sec_setup} presents the notation and lays out basic assumptions.
Section \ref{sec_ate_srd} describes the ATE estimator for the sharp case and proves asymptotic normality.
It is divided into two sub-sections. 
Section \ref{sec_case1} treats ATEs of discrete counterfactual distributions, which is a straightforward generalization of one-cutoff RDD.
Section \ref{sec_case2} is novel; it studies ATEs of continuous counterfactual distributions under the first heterogeneity assumption.
Section \ref{sec_case3} analyzes the fuzzy case.
Appendix \ref{sec_appen} contains all proofs. 
Supplemental Appendix \ref{sec_supappen} collects auxiliary results to the proofs in Appendix \ref{sec_appen}.\footnote{Appendix 
\ref{sec_supappen} is available online at \textit{www.nd.edu/$\sim$mbertanh}.}

\section{Setup}\label{sec_setup}
\indent

This section sets up the framework for RDD with multiple cutoffs.
There are $P$ sub-populations of individuals indexed by $p=1,\ldots, P$.
An example of a sub-population may be a town-year in the high school application, or a state in the Medicaid example.
Each individual $i$ in sub-population $p$ is fully characterized by a vector of random variables
$(X_{i,p},U_{i,p})$ drawn iid across $i$ from each sub-population. 
The forcing variable $X_{i,p}$ is a scalar score that governs eligibility for treatment,
and it lives in a compact interval $ \m{X}=[\underline{ \m{X}},\overline{ \m{X}}]$;
$U_{i,p}$ is a vector of unobserved heterogeneity.
Individual $(i,p)$ receives a treatment dose $D_{i,p}$ from a set of possible  treatments $\m{D}$.
The outcome variable $Y_{i,p}$ is determined by a function $\mmy$ of the individual characteristics and treatment,
\begin{equation}
Y_{i,p} = \mmy(X_{i,p},D_{i,p},U_{i,p}).
\end{equation}

I start with the simpler sharp RDD setting and defer the fuzzy RDD case to Section \ref{sec_case3}.
In the sharp case, the treatment received by the individual is a deterministic function of the forcing variable.
For an individual with forcing variable $X_{i,p}$ close to a cutoff $c$,
the treatment dose is $d$ if $X_{i,p}<c$, or $d'$ if $X_{i,p}\geq c$.
\cite{hahn2001id} demonstrate that continuity of the conditional mean of outcomes is sufficient to identify 
average causal effects for individuals local to the cutoff $c$.
\begin{lemma}\label{lemma_srd_id}
Assume that $\mme[ \mmy(X_{i,p},d,U_{i,p}) | X_{i,p}=x]$ is a continuous function of $x$ 
for the treatment doses $d$ and $d'$ in the neighborhood of the cutoff $c$.
Then, the average causal effect for individuals with $X_{i,p}=c$ is identified:
\begin{align}
& \mme\left[\left. \mmy(X_{i,p},d',U_{i,p}) - \mmy(X_{i,p},d,U_{i,p}) ~\right|~ X_{i,p}=c \right] 
\notag
\\*
& \hspace{2cm}= \lim_{e \downarrow 0} \Big \{~ \mme[Y_{i,p} ~|~ X_{i,p}=c+e]-\mme[Y_{i,p} ~|~ X_{i,p}=c-e] ~ \Big\}.
\end{align}
\end{lemma}

Lemma \ref{lemma_srd_id} generalizes to the case of multiple cutoffs and treatments under the assumption of 
continuity of $\mme[ \mmy(X_{i,p},d,U_{i,p}) | X_{i,p}=x]$ as a function of $x$ for every $d\in \m{D}$.
Many cutoffs arise because data sets may have many sub-populations with few cutoffs (e.g. Medicaid benefit with one cutoff per state, many states);
or few sub-populations with many cutoffs (e.g. Romanian high schools with one town and many schools).
The ability to exploit variation in cutoff-dose values relies on the following pooling assumption.
\loadaspt{00A} 

Assumption \ref{assu_pool} does not restrict average outcomes to be the same across different sub-populations.
It is less restrictive than common specifications for pooling data in applied work, for example,  time-trends and sub-population fixed effects.
The pooling assumption says that individuals with the same forcing variable
that undergo the same change in treatment have the same average response
across different sub-populations. 
The rest of the paper builds on Assumption \ref{assu_pool},
and it becomes irrelevant to distinguish sub-populations.
Thus, I drop the subscript $p$ and focus on the case of one population with multiple cutoffs.

The cutoffs are ordered such that $c_{1}< c_{2} < \ldots < c_{K}$.
Sharp RDD means that
an individual with forcing variable $X_i$ 
is deterministically assigned to a treatment dose $D_i=D(X_i)$ according to the following 
rule: 
\begin{align}
 D(x) =
 \left\{
 \begin{array}{cc}
 d_{0} & \text{ if } c_{0} \leq x < c_{1}
 \\
 d_{1} & \text{ if } c_{1} \leq x < c_{2}
 \\
 \vdots
 \\
 d_{K} & \text{ if } c_{K} \leq x \leq c_{K+1}
 \end{array}
 \right.
 \label{def_treat_sch}
\end{align}
where
$c_{0}=\underline{ \m{X}}$, and $c_{K+1}=\overline{ \m{X}}$.
Each cutoff is characterized by three variables:
the scalar threshold $c_{j}$; the treatment dose $d_{j-1}$ the individual receives if $c_{j-1}\leq X_i <c_{j}$;
 and the treatment dose $d_{j}$ the individual receives if $c_{j}\leq X_i < c_{j+1}$.
Let $\bc_j = (c_{j},d_{j-1},d_{j})$.
The schedule of cutoffs and treatment doses is given by the non-random set
$\m{C}_K=\left\{ \boup{c}_{j} \right\}_{j=1 }^K$.
The richness of set $\m{C}_K$ increases as the researcher collects more data.\footnote{
The validity of the RDD depends crucially on  exogeneity of cutoffs and no manipulation
of the forcing variable $X$ by individuals. See \cite{mccrary2008} for a test of forcing variable manipulation.
\cite{bajari2017} present a modified RDD estimator that is consistent under forcing variable manipulation
in a class of structural models.}

The data generating process is summarized as follows.
Values for the forcing variable $X_i$ and heterogeneity $U_i$ are drawn iid $i=1,...,n$ from a joint distribution.
Given $D(x)$, these $n$ individuals are assigned to different treatment doses $D_i=D(X_i)$.
The observed outcome is determined by $Y_i = \mmy(X_{i},D_i,U_{i})$.
The econometrician observes the schedule of cutoffs and treatment doses $D(x)$
and $(Y_i,X_i,D_i)$ for $i=1,...,n$. 
Following Rubin's model of potential outcomes, let $Y_i(d) = \mmy(X_{i},d,U_{i})$,
and assume continuity of $\mme[Y_i(d)|X_i=x]$ for every $d\in \m{D}$. 
A simple extension of Lemma \ref{lemma_srd_id} identifies average effects
at every cutoff $\bc \in \m{C}_K$,
\begin{align}
\beta(\boup{c})  = & \mme[Y_i(d')-Y_i(d)|X_i=c]
\notag\\
= & \lim_{e \downarrow 0} \Big \{~ \mme[Y_{i} ~|~ X_{i}=c+e]-\mme[Y_{i} ~|~ X_{i}=c-e] ~ \Big\}.
\label{def:beta}
\end{align}

Data with multiple cutoff-dose values allow the researcher to learn the causal effect of a variety of dose changes 
applied to individuals at various levels of the forcing variables.
This fact opens the possibility of using observed data to estimate the effect of new policy changes.
The individual response function $\mmy$ may well depend on the initial assignment of treatments $D_i$, and it could potentially change under counterfactual policies.
Unless such dependence is restricted, it becomes impossible to use existing data to infer the effect of new policies.
The remainder of this paper relies on the following policy-invariance assumption.
\loadaspt{00AB} 

The methods of this paper leverage RDD variation in cutoff-dose values to make inferences on average effects of policy changes.
A policy change is a counterfactual distribution of changes in treatment doses that
are randomly applied to individuals, conditional on the forcing variable.
An individual $i$ is assigned to a change in treatment dose from $D_i^*$ to $D_i^{**}$,
where the distribution of $(D_i^* , D_i^{**})$ is independent of $U_i$ after conditioning on $X_i$.
Under Assumption \ref{assu_poli}, the average causal effect of such an experiment is 
\begin{align}
\mu & = \mme\left[~ \mmy(X_{i},D_i^{**},U_{i}) - \mmy(X_{i},D_i^{*},U_{i}) ~ \right]
\notag
\\
\notag
& = \mme\left[ ~ \mme\left( \left. \mmy(X_{i},D_i^{**},U_{i}) - \mmy(X_{i},D_i^{*},U_{i}) ~\right| D_i^{**}, D_i^* , X_i \right) ~ \right]
\\
\notag
& = \mme\left[ ~ \mme\left( \left. \mmy(X_{i},D_i^{**},U_{i}) - \mmy(X_{i},D_i^{*},U_{i}) ~\right|  X_i \right) ~ \right]
\\
& = \mme\left[ ~ \beta( X_{i}, D_i^{*}, D_i^{**} ) ~ \right]
\end{align}
where the last equality uses the definition of $\beta$ in Equation \ref{def:beta}. 
The average effect $\mu$ equals an average of the $\beta$ function
over the counterfactual distribution of $( X_{i}, D_i^{*}, D_i^{**} )$.
The inference methods of this paper first identify $\beta$ from RDD with many cutoffs, then identify the average of $\beta$
under a counterfactual distribution pre-specified by the researcher.
In a similar setting, \cite{cattaneo2016keele} study identification 
under conditions equivalent to Assumptions \ref{assu_pool} and \ref{assu_poli} (respectively, their Assumptions 5a and 5b).

The definition of $\mu$ captures both the direct effect of changing $D$,
and the composition effect of a change in the distribution of $D$ conditional on $X$.
To investigate the direct effects of $D$, 
\cite{rothe2012} proposes methods 
for inference on partial policy effects that preserve the distribution of ranks of $(D,X)$ unchanged, thus controlling  for composition effects.
Although not the focus of this paper, Rothe's methods may be combined with the RDD identification strategy 
to study partial policy effects.

\section{Average Treatment Effects in the Sharp Case}
\label{sec_ate_srd}
\indent

This section investigates estimation and inference of averages of the non-parametric function $\beta$ under sharp RDD with many cutoffs.
First, I treat the case of qualitative treatment doses. 
This is a straightforward extension of single-cutoff RDDs which identify ATEs of discrete counterfactual distributions, with support contained in $\m{C}_K$.
Second, I treat the case of quantitative treatment doses, that is, the first heterogeneity assumption. 
Substantial variation in cutoff-dose values allows for novel methods that estimate ATEs with support more general than $\m{C}_K$.

\subsection{Discrete Counterfactuals}\label{sec_case1}

\indent

Consider applications of RDD where the treatment dose variable has a qualitative nature, and is \textit{not}
credibly summarized by a real-valued metric. 
For example, \cite{hastings2013} study the assignment of students into
different degree programs in universities in Chile.
There are multiple cutoffs on a test score, but different cutoffs switch students to
completely different programs, e.g. physics, engineering, economics, etc.
This limits the ability to combine local effects across cutoffs, which restricts ATEs 
to counterfactual distributions with discrete support contained in $\m{C}_K$.
In this section, it is not possible to identify effects of policies that places weight on cutoff-dose combinations $(c,d,d')$ that
are not in $\m{C}_K$.

The focus is on discrete counterfactual distributions 
with probability mass function $\omega^d (\bc)$
where 
$\omega_j^d = \omega^d(\bc_j)$ for every $j$.
For example, in the high school assignment application, a new policy may reallocate
students with test scores marginally across the existing cutoffs.
The weight $\omega_j^d$ represents
the probability mass of students
with test score equal to $c_{j}$ that undergo
a change in school quality from $d_{j-1}$ to $ d_{j}$ in the reallocation policy.

The parameter of interest is the average effect on these students,
which is a weighted  average of local effects at the existing cutoffs: 
\begin{gather*}
\mu^d = \sum\limits_{j=1}^K \omega_{j}^d ~ \beta(\boup{c}_{j}).
\end{gather*}

Identification follows from Equation \ref{def:beta}.\footnote{The
 common practice of normalizing all cutoffs to zero and estimating only one effect
produces an estimator consistent for $\mu^d$ with weights $\omega_j^d = f(c_j) / \sum_l f(c_l)$ where $f$ is the probability density function of $X$.}
Estimation is conducted in two steps.
The first step uses local polynomial regressions (LPR) near each cutoff $c_{j}$
to non-parametrically estimate 
\begin{equation}
B_{j} = \lim_{e \downarrow 0} \left\{~ \mme[Y_i|X_i=c_{j}+e]-\mme[Y_i|X_i=c_{j}-e] ~\right\}.
\label{eq_Bpj}
\end{equation}
The researcher chooses a bandwidth parameter $h_{1j}>0$ for each cutoff, a kernel density
function $k(.)$, and the order of the polynomial regression $\rho_1 \in \mathbb{Z}_{+}$.
A polynomial in $X$ is fitted on each side of the cutoff, and the estimator $\hat B_{j}$
is the difference between the intercepts of these two polynomial regressions:
\begin{align}
\hat B_{j} & =  \hat{a}_{j}^{+} - \hat{a}_{j}^{-} \label{eq_est_Bpj} &
\\
(\hat{a}_{j}^{+}, \hat{\boup{b}}_{j}^{+} )
&= \argmin\limits_{(a, \boup{b} )}
\sum\limits_{i=1}^n  \bigg\{ k \left( \frac{X_i - c_{j}}{h_{1j} } \right) v_{i}^{j+}
\nonumber\\*
& \hspace{3cm} \big[ Y_i - a - b_1 (X_i - c_{j}) - \ldots -b_{\rho_1} (X_i - c_{j})^{\rho_1} \big]^2 \bigg\}
\label{eq_est_lpr_r}
\\
(\hat{a}_{j}^{-}, \hat{\boup{b}}_{j}^{-} )
& = \argmin\limits_{(a, \boup{b} )}
\sum\limits_{i=1}^n  \bigg\{  k \left( \frac{X_i - c_{j}}{h_{1j} } \right) v_{i}^{j-}
\nonumber\\*
& \hspace{3cm} \big[ Y_i - a - b_1 (X_i - c_{j}) - \ldots -b_{\rho_1} (X_i - c_{j})^{\rho_1} \big]^2 \bigg\}
\label{eq_est_lpr_l}
\end{align}
where 
\begin{gather}
v_{i}^{j+}=  \mmi\{ c_{j} \leq X_i < c_{j}+ h_{1j}\}
\text{, }
v_{i}^{j-}=  \mmi\{ c_{j}- h_{1j} < X_i < c_{j}  \},
\label{def_vij}
\end{gather}
and $\boup{b}=(b_1,\ldots,b_{\rho_1})$.
The estimator $\ha{B}_j$ uses observations with $X_i$ in the estimation window
$[c_j - h_{1j}, c_j + h_{1j}]$.
The choice of bandwidths may allow the windows to overlap at consecutive cutoffs. 
However, it must be the case that $c_j + h_{1j}<c_{j+1}$ and $c_j \leq c_{j+1}-h_{j+1}$ for $j=1,\ldots,K-1$.
This ensures that $Y_i=Y_i(d_j)$ for $X_i \in [c_j, c_j + h_{1j}]$, 
and $Y_i=Y_i(d_{j-1})$ for $X_i \in [c_j - h_{1j}, c_j)$.\footnote{This is the first-step estimation procedure for one sub-population with $K$ cutoffs. 
In many settings, the data have many sub-populations $p =1, \ldots, P$ with one or more cutoffs $j=1, \ldots, K(p)$ in each sub-population.
In that case, the researcher first estimates $\ha{B}_{j,p}$ for every $j$ in each sub-population $p$.
Then, Assumption \ref{assu_pool} allows for pooling of $\ha{B}_{j,p}$ across $p$ in the second step.}

In the second step, the researcher averages out $\hat B_{j}$ to obtain the estimator
$\ha{\mu}^d$:
\begin{gather}
 \ha{\mu}^d = \sum_{j =1 }^K \omega_{j}^d \hat B_{j}.
 \label{est:mu:d}
\end{gather}

For the case of one cutoff, \cite{hahn2001id} and \cite{porter2003} derive 
the asymptotic normal distribution of the LPR estimator $\hat B_{j}$.
I build on their arguments to derive the asymptotic distribution of 
$\ha{\mu}^d$ under the assumptions listed below.

\loadaspt{00B} 

\loadaspt{00C} 

\loadaspt{00D} 

\begin{theorem}\label{theo_srd_est_avg}
Suppose Assumptions  \ref{assu_srd_est_kernel}-\ref{assu_srd_est_mx} hold. 
Let $\underline{h}_1 = \min_j h_{1j}$ and
$\overline{h}_1 = \max_j h_{1j}$.
As $n\to\infty$, assume that $\overline{h}_1 \to 0$,
$\overline{h}_1/\underline{h}_1 = O(1)$,
$n \overline{h}_1 \to \infty$,
and
$(n \overline{h}_1)^{1/2} \overline{h}_1^{\rho_1 + 1} = O(1)$.
Then,
\begin{align}
\frac{ \ha{\mu}^d  - \m{B}_n^d - \mu^d }
	{ \left(\m{V}_n^d \right)^{1/2} }
\overset{d}{\rightarrow } &
N(0,1)
\notag
\end{align}
where the bias $\m{B}_n^d$ and variance $\m{V}_n^d$ terms are characterized as follows:
\begin{align}
\m{B}_{n}^d = & \frac{1}{(\rho_1 + 1 )!} \sum_{j=1}^K
			h_{1j}^{\rho_1+1}  f(c_j) \left[ 
				\nabla^{\rho_1 + 1 } R(c_j,d_j) e_1'  G_n^{j +}
				- 
				\nabla^{\rho_1 + 1 } R(c_j,d_{j-1}) e_1'  G_n^{j -} 
			\right]
			\gamma^*
\label{def:bias:d}
\\
\m{V}_n^d = &n \mme \left\{
	\eps_i^2 \left[
		\sum_{j=1}^K 
			\frac{\omega_j^d}{n h_{1j}}
			k\left( \frac{X_i - c_{j}}{h_{1j}} \right)
			e_1'
			\left(
				v_i^{j +} \mme[ G_n^{j +} ]  
				-
				v_i^{j -} \mme[ G_n^{j -} ]  
			\right)
			\ti{H}_{i}^{j}				
	\right]^2
\right\},
\label{def:var:d}
\end{align}
with $\varepsilon_i = Y_i - \mme[Y_i|X_i]$;
$H(u) = [u^0, u^1, \ldots, u^{\rho_1} ]' $
is a $(\rho_1+1) \times 1$ vector-valued function,
$\ti{H}_i^j = H(h_{1j}^{-1}(X_i-c_j))$,
and
$G_n^{j \pm} = ( n h_{1j} )^{-1} \sum_{i=1}^n v_{i}^{j \pm} k(h_{1j}^{-1}(X_i-c_j))\ti{H}_i^j  \ti{H}_i^{j'}  $
is  a $(\rho_1+1) \times (\rho_1+1)$ matrix;
$v_{i}^{j \pm}$ are defined in Equation \ref{def_vij};
$\gamma^*=[\gamma_{\rho_1 + 1} ~~ \ldots ~~\gamma_{2 \rho_1 +1}]'$,
for $\gamma_d=\int_0^1 k(u) u^d du $;
and
$e_1$  is the  $(\rho_1+1 \times 1)$ vector with one in its first coordinate and zero otherwise.
Furthermore, $\left(\m{V}_n^d \right)^{-1/2} =  O \left( \left( n \overline{h}_1 \right)^{1/2}  \right)$, and
$\left(\m{V}_n^d \right)^{-1/2} \m{B}_n^d  =  O_P \left( \left( n \overline{h}_1\right)^{1/2} \overline{h}_1^{\rho_1 + 1}  \right)$.
\end{theorem}

The variance of $\ha{\mu}^{d}$ is consistently estimated by
\begin{align}
\ha{\m{V}}_n^d = &\sum_{i=1}^n \left\{
	\ha{\eps}_i^2 \left[
		\sum_{j=1}^K 
			\frac{\omega_j^d}{n h_{1j}}
			k\left( \frac{X_i - c_{j}}{h_{1j}} \right)
			e_1'
			\left(
				v_i^{j +}  G_n^{j +}   
				-
				v_i^{j -} G_n^{j -}  
			\right)
			\ti{H}_{i}^{j}				
	\right]^2
\right\}.
\label{est:var:d}
\end{align} 
The squared residuals $\ha{\eps}_i^2$ are computed by a nearest-neighbor matching estimator, as suggested by \cite{cattaneo2014calonico} (CCT from now on):
\begin{align}
\ha{\eps}_i^2 = & \frac{3}{4} \left(
	Y_i  - \frac{1}{3} \sum_{l =1}^3 Y_{\ell(i,l)}
\right)^2,
\label{est:resid}
\end{align}
and $\ell(i,l)$ is the index of the $l$-th closest $X$ to $X_i$ that lies within the same cutoffs $c_j$ and $c_{j+1}$ that $X_i$ does.
CCT's Theorem A3 demonstrates that $\ha{\m{V}}_n^d / \m{V}_n^d \pto 1$ in the case of one cutoff, and a straightforward generalization yields the same conclusion for a finite number of cutoffs. 
If the bandwidth choices are such that the standardized bias term $\left(\m{V}_n^d \right)^{-1/2} \m{B}_n^d$ differs from  zero asymptotically, then inference must be done using a bias-corrected estimator. 
A practical way of doing bias correction  is to increase the order of the polynomial from $\rho_1$ to $\rho_1+1$ and compute 
$\ha{\mu}^{d'}$ and $\ha{\m{V}}_n^{d'}$ using the same bandwidth choices as $\ha{\mu}^{d}$ and $\ha{\m{V}}_n^{d}$.
It follows that $(\ha{\m{V}}_n^{d'})^{-1/2}  (\ha{\mu}^{d'} - \mu^d )\dto N(0,1)$.

The multi-cutoff setup of Theorem \ref{theo_srd_est_avg} allows for choices of bandwidths that produce overlapping estimation windows in finite samples.
For example, if $c_1 + h_{11}>c_2-h_{12}$, the estimator $\ha{B}_2$ uses some of the same observations that the estimator $\ha{B}_{1}$ does.
In theory, a finite number of cutoffs with shrinking bandwidths leads to non-overlapping estimation windows in large samples. As a consequence, the asymptotic variance of $\sqrt{n \overline{h}_1} (\ha{\mu}^d  - \m{B}_n^d - \mu^d) $
may not approximate  its finite-sample variance well in case of overlap.
Instead, the variance term in \eqref{def:var:d} takes into account overlap because its formula is constructed based on the finite-sample variance.

In practice, implementation of $\ha\mu^{d}$ requires the researcher to choose bandwidths $h_{1j}>0$, 
the polynomial order $\rho_1 \in \mmz_+$, and a kernel density function $k(\cdot)$.
In the one-cutoff case, common choices in applied work include the edge kernel $k(u)=\mmi\{|u| \leq 1 \} (1-u)$,
local linear regression $\rho_1=1$, and a bandwidth choice that minimizes the mean squared error (MSE) of estimation.
Recent work by \cite{imbens2012optimal} (IK from now on) provides a practical data-driven rule for choosing the bandwidth in the case of one cutoff.
With multiple cutoffs, an interesting aspect of the optimal bandwidth problem is the variance reduction from overlapping estimation windows.\footnote{Following Equation \eqref{eq_est_Bpj}, 
$COV(\ha{B}_j,\ha{B}_{j+1})
=COV( \ha{a}_j^+ -\ha{a}_j^- ,\ha{a}_{j+1}^+ - \ha{a}_{j+1}^-)
=COV( \ha{a}_j^+ ,- \ha{a}_{j+1}^-)<0$ because $\ha{a}_j^+$ and $\ha{a}_{j+1}^-$
use some of the same observations in the case of overlap.} 
A formal investigation on optimal  bandwidths in the multi-cutoff case is deferred to future work.

A simple recommendation to implement Theorem \ref{theo_srd_est_avg} is to use 
the IK bandwidth based on local linear regressions with the edge kernel applied to the sub-sample pertaining to each cutoff. 
These bandwidths produce asymptotic bias, and valid inference must use a bias-corrected estimator   
and its variance.
Use local quadratic regressions ($\rho_1=2)$ with the edge kernel and the same bandwidths as before
to compute the consistent bias-corrected estimator $\ha{\mu}^{d'}$ and its variance $\ha{\m{V}}_n^{d'}$.
\cite{calonico2018optimal} propose shrinking MSE-optimal bandwidths as a rule of thumb to improve finite sample coverage of confidence intervals.
As means of a robustness check, the researcher may shrink the IK bandwidths by multiplying them by $n^{-1/20}$, and examine the resulting confidence intervals (Section 4.1, \cite{calonico2018optimal}).

\subsection{Continuous Counterfactuals}\label{sec_case2}

\indent

The first heterogeneity assumption allows the researcher to identify counterfactual ATEs with support more general than $\m{C}_K$. 
An empirical application satisfies the first heterogeneity assumption if the treatment dose is credibly quantifiable in a real-valued variable $d$.
For example, in the high school assignment of PU,
the treatment dose is a quality measure for each school.
Possible measures of school quality include the average test score of peers, the average number of teachers, or funding per student.
An infinite amount of data gives rise to a countably-infinite set of cutoff-dose values $\m{C}_{\infty}$.
In terms of the high school assignment example, 
a large number of towns and years produce substantial variation in cutoff-dose values.
Define  $\m{C}$ to be the convex hull of $\m{C}_{\infty}$.
If variation in cutoff-dose values is sufficiently rich, then ATEs with counterfactual distributions supported in  $\m{C}$ are identified (Lemma \ref{lemma_srd_id_kinf}).

I focus on scalar treatment doses $d$ and counterfactual distributions 
with continuous probability density function $\omega^c (\bc)$.
Minor changes to the setup can accommodate multivariate $d$ and discrete or mixed counterfactual distributions.
The ATE is defined as
\begin{gather}
\mu^c = \int_{\m{C}} \omega^c(\boup{c}) \beta(\boup{c}) ~d(\boup{c}).
\end{gather}

\begin{lemma}\label{lemma_srd_id_kinf}
Assume that an infinite amount of data has sufficient variation such that  (i) $\m{C}_{\infty}$ is dense in its convex hull $\m{C}$;
and that (ii) $\beta(\boup{c})$ is a continuous function over $\m{C}$.
Then, $\mu^c$ is identified.
\end{lemma}

The researcher may impose further heterogeneity restrictions to reduce the dimension of $\beta(\bc)$
and increase the set of possible counterfactual distributions.
For instance, linear returns to school quality say that $\beta(c,d,d')$ depends on $(c,d'-d)$ instead of $(c,d,d')$.
This implies that $\beta(\bc)=\phi(c)(d'-d)$ for a smooth function $\phi(c)$,
and changes the dimension of set $\m{C}_K$. 
See Figure \ref{figure_schedule} in Section \ref{sec_appli} for an empirical illustration.
Medicaid coverage is an example of binary treatment that is triggered by various income cutoffs across states.
In the case of binary treatment, the treatment effect function depends only on the cutoff value, that is, $\beta(c,d,d')= \phi(c)$.
Identification of averages of $\beta(\bc)$ requires identification
of averages of $\phi(c)$, which relies on infinitely many cutoff values that cover a compact interval on the real line.
For example, such variation identifies the average effect of giving Medicaid benefits
to an entire neighborhood of individuals within the range of income cutoffs seen in the data.\footnote{
For the Medicaid example, \cite{delamata2012} has many income cutoffs that differ by state, age, and year.
De La Mata's Table I suggests variation between  US\$ 21,394 and  US\$36,988.
Other examples of rich variation in cutoff values include: (i) \cite{agarwal2016} who have 714 credit-score cutoffs distributed between 620  and 800 (see their Figure II(E));
and (ii) \cite{hastings2013} who have at least 1,100 cutoffs on admission scores varying between 529.15 and 695.84 (refer to their online appendix's Table A.I.I).   
Although Angrist and Lavy (1999) have few cutoff values, the pattern of their Figure I suggests variation in dose changes across grades and schools.
In such cases, non-parametric identification of $\beta$ is possible for a range of dose changes at a few cutoff values.
 }

The parameter $\mu^c$ is estimated in two steps.
The first step
is identical to the procedure described in
Equations \ref{eq_est_Bpj}-\ref{eq_est_lpr_l}.
That is, LPRs produce estimates $\widehat{B}_{j}$, $j=1,\ldots,K$.
The second step computes a weighted average of the first-step estimates, using specially designed weights $\{ \Delta_j\}_{j=1}^{K}$ that I  call
``correction weights,''
\begin{gather}
\ha{\mu}^c  = \sum_{ j =1  }^K \Delta_{j} \hat B_{j}.
\label{est:mu:c}
\end{gather}

Unlike the intuition of the discrete case, the correction weight $\Delta_j$ is not necessarily equal or proportional to 
$\omega^c_j = \omega^c(\bc_j)$.
An analytical expression for $\Delta_j$ is given below in Equation \ref{eq_est_int_corr_wei}, and constructed as follows.
The correction weight $\Delta_j$ is the contribution of estimate $\ha{B}_j$ to the integral
$\int_{\m{C}} \omega^c(\boup{c}) \ha{\beta}(\boup{c}) ~d(\boup{c})$,
where $\ha{\beta}(\boup{c})$ is a non-parametric estimate of $\beta(\boup{c})$.
A weighted regression of  $\ha{B}_j$ on polynomial functions of $\bc_j$ centered at $\bc$ produces the estimate $\ha{\beta}(\boup{c})$.
The researcher specifies the order of the polynomials $\rho_2 \in \mmz_+$, and  a bandwidth $h_2 >0$ that defines an estimation neighborhood around $\boup{c} \in \m{C}$.
The estimate $\widehat \beta(\boup{c})$ is the intercept of the following weighted least squares regression:
\begin{align}
\widehat{\boup{\eta}}
 =  & \argmin\limits_{\boup{\eta}}
\left( \widehat{\boup{B}} - \boup{E}(\boup{c}) \boup{\eta} \right)'
 \bOmg( \bc; h_2 )
 \left( \widehat{\boup{B}} - \boup{E}(\boup{c}) \boup{\eta} \right)
\\
& \text{ where } 
\nonumber
\\
\widehat{ \boup{B} }  = &  \left[\widehat B_{1}, ~ \ldots, ~ \widehat B_{K} \right]'
\text{ is a } K \times 1 \text{ vector;}
\\
 \bOmg(\bc;{h_2}) =  & \diag \left\{ \Omega_{j}(\bc;h_2) \right\}_{j=1}^K
\textit{ is a } K \times K \textit{ matrix, with}
\label{eq_est_int_Omg}
\\*
& \Omega_{j}(\bc;h_2) = 
	k \left( \frac{c_{j}  - c}{h_2 } \right)
	k \left( \frac{d_{j-1} - d}{h_2 } \right)
	k \left( \frac{d_{j-1} - d' }{h_2 } \right); 
\nonumber
\\
\bE(\bc)=  & \left[E_{1}(\bc), ~ \ldots, ~ E_{K}(\bc) \right]'
\textit{ is a } K \times J \textit{ matrix, where}
\label{eq_est_int_E}
\\*
& E_{j}(\bc) \textit{ is a } J \times 1 \textit{ vector with all polynomials of the form } \nonumber
\\*
&  p_{\bgam}(\bc_j - \bc) =  (c_j - c)^{\gamma_1}(d_{j-1} - d)^{\gamma_2}(d_{j} - d' )^{\gamma_3} \nonumber
 \\
& \textit{for }  \bgam = (\gamma_1,\gamma_2,\gamma_3) \in \mmz^3_+,
	~\gamma_1+\gamma_2+\gamma_3 \leq \rho_2,
	~\min\{ \gamma_2, \gamma_3 \} = 0, \nonumber
\\
& J=  2 \frac{(\rho_2 +2)!}{2! \rho_2!} - (\rho_2+1) \textit{, where $!$ denotes factorial,} \nonumber
\\
& \textit{and the first element of }  E_j({\bc}) \textit{ is } 1.
\nonumber
\end{align}

The formula for $\Delta_j$ comes from integrating $\omega^c(\bc) \ha{\beta}({\bc})$:
\begin{align}
 \int_{\m{C}} \omega^c(\bc) \widehat \beta(\bc) ~ d\bc = &
	\bigintsss_{\m{C}} ~ \omega^c(\bc) e_1'
 		\left( \bE(\bc)' \bOmg(\bc;h_2) \bE(\bc) \right)^{-1}
 		\sum_j \Omega_{j}(\bc;h_2) E_{j}(\bc) \ha{B}_j ~~d(\bc)
\notag
\\ 
 = & \sum_j
	\bigintsss_{\m{C}} ~ \omega^c(\bc) e_1'
 		\left( \bE(\bc)' \bOmg(\bc;h_2) \bE(\bc) \right)^{-1}
 		\Omega_{j}(\bc;h_2) E_{j}(\bc) ~d(\bc) ~ \ha{B}_j
\notag
\\ 
= & \sum_j
	\underset{\equiv \Delta_j}{\underbrace{
		\bigintsss_{\m{C}} ~ \omega^c(\bc)
 			\frac{det \left( \bE(\bc)' \bOmg(\bc;h_2) \bE_{\bo{0} \leftarrow e_j}(\bc) \right)}
   			{det \left( \bE(\bc)' \bOmg(\bc;h_2) \bE(\bc) \right)}
   		~d(\bc)
   	}}
   	~ \ha{B}_j
   	\label{eq_est_int_corr_wei}
 \\
 = & \sum_j \Delta_j \ha{B}_j,
\end{align}
where  the third equality uses the Cramer rule, 
and
$\bE_{\bo{0} \leftarrow e_j}(\bc)$ is a  $K \times J$ matrix equal to $\bE(\bc)$ except for the first column,
which is replaced by the $K \times 1$ vector $e_j$. The vector $e_j$ has one in its $j$-th entry and zero otherwise.

The main contribution of this paper concerns inference on $\mu^c$ where
$\beta(\boup{c})$ is estimated non-parametrically and then averaged across cutoffs.
This is not the first paper to study estimation of averages
of non-parametric functions; for example, see \cite{newey1994partial}. 
The novelty here is that the non-parametric estimation step only occurs at $K$ fixed boundary points $\bc_j$.
A necessary condition for consistency of $\hat \mu^{c}$ is an ``infill type of asymptotics,'' that is,   
$K$ grows large with the sample size $n$, and $\m{C}_K$ 
becomes dense in its convex hull $\m{C}$.
Assumption \ref{assu_srd_kinf_est_cut} makes the dependence of $K$, $h_{1j}$, $h_2$, and $c_{j}$ on $n$ explicit with a subscript.
The main text omits the subscript $n$ whenever possible to simplify  notation.
\loadaspt{00E} 

For large $K$,
cutoff-dose values must be uniformly distributed on the domain $\m{C}$ such that
$\bE(\bc/h_2)' \bOmg(\bx;h_2) \bE(\bc/h_2)$ is invertible and of magnitude
$Kh_2^3$, that is, $K$ times the volume of every $h_2$-neighborhood of $\bc$, for every $\bc$ in $\m{C}$.
These conditions are satisfied in a variety
of examples of triangular arrays of points. 
In Section \ref{sec_int_app} of the supplemental appendix,
these conditions are verified for one example of a triangular array.
Asymptotic normality also relies on additional smoothness conditions on the moments of the data.
\loadaspt{00F} 

Theorem \ref{theo_srd_kinf_est_int} states the rate conditions under which the estimator $\ha{\mu}^c$
has an asymptotic normal distribution.
Estimation of the ATE consists of approximating the integral of the treatment effect function by a weighted sum of the values of such function at a finite
number of points in its domain.
The approximation error converges
to zero as the number of points grows large. Function evaluations $B_j$ are estimated by $\ha{B}_j$.
The correction weights guarantee that the integral approximation error converges to zero faster than the estimation error. 
\begin{theorem}\label{theo_srd_kinf_est_int}
Suppose Assumptions \ref{assu_srd_est_kernel}-\ref{assu_srd_kinf_est_int}
hold.
As $n\to\infty$, assume that $K \to \infty$,
$\overline{h}_1 \to 0$,
$ \overline{h}_1 / \underline{h}_1 = O(1)$,
and $h_2 \to 0$ such that
\textbf{(i)} $\left( K n \overline{h}_1 \right)^{1/2} \overline{h}_1^{\rho_1 + 1} = O(1)$;
\textbf{(ii)} ${ K^{1/2} \log n}/{ \left( n \overline{h}_1 \right)^{1/2} } = o(1)$, and
$K\overline{h}_1=O(1)$;
and
\textbf{(iii)} $\left( K n \overline{h}_1 \right)^{1/2}  h_2^{ \rho_2+1 } =O(1)$, and $1 / K h_2^{3} =O(1)$.
Then, 
\begin{gather}
\frac{  \ha{\mu}^c - \m{B}_{1n}^c - \m{B}_{2n}^c - \mu^c}
{ \left(\m{V}_n^c \right)^{1/2} }
\overset{d}{\rightarrow }
N(0,1).
\end{gather}
The first-step bias $\m{B}_{1n}^c$ and variance $\m{V}_n^c$ terms are defined as in Equations \ref{def:bias:d}-\ref{def:var:d} except that $\Delta_j$ replaces $\omega_j^d$;
the second-step bias $\m{B}_{2n}^c$ is characterized as follows:
\begin{align}
\m{B}_{2n}^c = & \bigintss_{\m{C}}  \omega^c(\bc)
	\sum\limits_{
			(\gamma_1, \gamma_2, \gamma_3) 
    	}
    	\sum\limits_{ j=1 }^{K}  \Bigg\{ 
    		\frac{ (c_j -c)^{\gamma_1} 
    			(d_{j-1} - d)^{\gamma_2}
    			(d_{j} - d')^{\gamma_3}
    		}{\gamma_1 ! \gamma_2 ! \gamma_3 !}   
    		\nabla_c^{ \gamma_1 } \nabla_d^{ \gamma_2 } \nabla_{d'}^{ \gamma_3} \beta(c,d,d')
\notag      \\*
   & \hspace{4.5cm} \frac{det \left( \bE(\bc)' \bOmg(\bc;h_2) \bE_{\bo{0} \leftarrow e_j}(\bc) \right)}
   {det \left( \bE(\bc)' \bOmg(\bc;h_2) \bE(\bc) \right)} \Bigg\}
~~d\bc,		
\end{align}
where the first sum runs over all triplets $(\gamma_1,\gamma_2,\gamma_3) \in \mmz_+^3$
such that $\gamma_1+\gamma_2+\gamma_3=\rho_2+1$,
and $\min \{\gamma_2, \gamma_3 \}=0$.
Furthermore, $\left(\m{V}_n^c \right)^{-1/2} =  O \left( \left( K n \overline{h}_1 \right)^{1/2}  \right)$, 
$\left(\m{V}_n^c \right)^{-1/2} \m{B}_{1n}^c  =  O_P \left( \left( K n \overline{h}_1\right)^{1/2} \overline{h}_1^{\rho_1 + 1}  \right)$, and 
$\left(\m{V}_n^c \right)^{-1/2} \m{B}_{2n}^c  =  O \left( \left( K n \overline{h}_1\right)^{1/2} h_2^{\rho_2 + 1}  \right)$.

\end{theorem}

A consistent estimator for ${\m{V}}^c_n$ is 
\begin{align}
\ha{\m{V}}_n^c = &\sum_{i=1}^n \left\{
	\ha{\eps}_i^2 \left[
		\sum_{j=1}^K 
			\frac{\Delta_j}{n h_{1j}}
			k\left( \frac{X_i - c_{j}}{h_{1j}} \right)
			e_1'
			\left(
				v_i^{j +}  G_n^{j +}   
				-
				v_i^{j -} G_n^{j -}  
			\right)
			\ti{H}_{i}^{j}				
	\right]^2
\right\},
\label{est:var:c}
\end{align} 
where $\ha{\eps}_i^2$ is computed using Equation \ref{est:resid}.
Lemma \ref{lemma:est:var} in the supplemental appendix's Section \ref{sec:supp:app:est:se}  demonstrates that 
$\ha{\m{V}}_n^c / \m{V}_n^c \pto 1$ under the condition that $(K\underline{h}_1)^{-1}=O(1)$. 
If the bandwidth choices are such that the standardized bias term 
$\left(\m{V}_n^c \right)^{-1/2} \left(\m{B}_{1n}^c + \m{B}_{2n}^c \right)$ differs from  zero asymptotically, then inference must be done using a bias-corrected estimator. 
A practical way of performing bias correction  is to increase the order of the polynomials from $(\rho_1,\rho_2)$ to $(\rho_1+1, \rho_2+1)$,
and to compute 
$\ha{\mu}^{c'}$ and $\ha{\m{V}}_n^{c'}$ using the same bandwidth choices as $\ha{\mu}^{c}$ and $\ha{\m{V}}_n^{c}$.
It follows that $(\ha{\m{V}}_n^{c'})^{-1/2}  (\ha{\mu}^{c'} - \mu^c )\dto N(0,1)$.

Convexity of $\m{C}$, along with the asymptotic behavior of the schedule of cutoff-doses (Assumption \ref{assu_srd_kinf_est_cut}),
is crucial for the numerical integration error to vanish sufficiently quickly, as required by Theorem \ref{theo_srd_kinf_est_int}.
Continuity of  $\omega^c(\bc)$ implies that the boundary of $\m{C}$ has zero probability under the counterfactual distribution.
Therefore, the convergence rate of $\ha{\mu}^c$ is not affected by the value of $\omega^c(\bc)$ over the boundary of $\m{C}$.
In finite samples, local polynomial estimates of $\ha{\beta}(\bc)$ may be noisy for values of $\bc$ at the boundary of the convex-hull of $\m{C}_K$. 
Researchers should take that into account when specifying the support of the counterfactual distribution $\omega^c(\bc)$.

A simple example illustrates the three rate conditions of Theorem \ref{theo_srd_kinf_est_int}.
Suppose $h_{1j}=n^{-\lambda_1}$ for all $j$, 
$h_2=n^{-\lambda_2}$,
and
$K=n^{\theta}$. 
The first-step estimation uses local-linear regression ($\rho_1=1$),
and the second step, local cubic regression ($\rho_2=3$).
The first rate condition 
says the first-step bandwidths  have to converge to zero fast enough to control the asymptotic
bias. That is, 
$\left( K n \overline{h}_1 \right)^{1/2} \overline{h}_1^{\rho_1 + 1} = O(1)$;
in terms of the example, this condition becomes
$\lambda_1 \geq (1+\theta)/(3+2 \rho_1)$.
The second rate condition 
restricts how fast the number of cutoffs grows with $n$. 
It cannot grow too fast to ensure having enough observations around the cutoffs for uniform consistency of first-step estimates.
The second condition has two parts: (a)
${ K^{1/2} \log n}/{ \left( n \overline{h}_1 \right)^{1/2} } =o(1)$
$\Leftrightarrow$
$\lambda_1 < 1 - \theta$;
and (b)
$K\overline{h}_1=O(1)$
$\Leftrightarrow$
$\lambda_1 \geq \theta$.
The third rate condition limits how slowly $K$ grows, relative to the sample size,
to ensure that the integral
approximation error vanishes faster than the estimation variance.
Part (a) of the third condition says 
$\left( K n \overline{h}_1 \right)^{1/2}  h_2^{ \rho_2+1 } =O(1)$
$\Leftrightarrow$
$\lambda_1 \geq 1+\theta - 2 \lambda_2 (\rho_2 + 1)$;
part (b) is
$1 / K h_2^{3} =O(1)$
$\Leftrightarrow$
$\lambda_2 \leq \theta/3$.

Figure \ref{figure_rates} illustrates these conditions and the feasible set for bandwidth choices (shaded area).\footnote{
Section \ref{sec_int_app} in the supplemental appendix 
gives an example of a schedule of cutoff-dose values that 
satisfies Assumption \ref{assu_srd_kinf_est_cut} for feasible choices of $(h_1, h_2)$
in this example.
}
Panel (a) shows the conditions in terms of $(\lambda_1,\theta)$
assuming $\lambda_2=\theta/3$ so that $h_2=K^{-1/3}$, which satisfies part (b) of the third condition.
Panel (b) depicts the same conditions in terms of $(\lambda_1,\lambda_2)$,
assuming $\theta=0.4$.
The feasible set is well-defined  as long as $K$ grows no faster than $\sqrt{n}$, 
that is, $\theta < 0.5$.
In addition,  $\rho_2 \geq 3$ because line 3(a) has to be below line 2(a).
The maximum rate of convergence of the estimator is $\sqrt{n}$, 
and it is reached along the dashed line 2(b).

\begin{figure}[H]
\caption{Rate Conditions of Theorem \ref{theo_srd_kinf_est_int}}
\label{figure_rates}
  \begin{minipage}{ 0.5\textwidth }
    \centering
    (a) First-Step Bandwidth and Number of Cutoffs
    
		\includegraphics[width=3in]{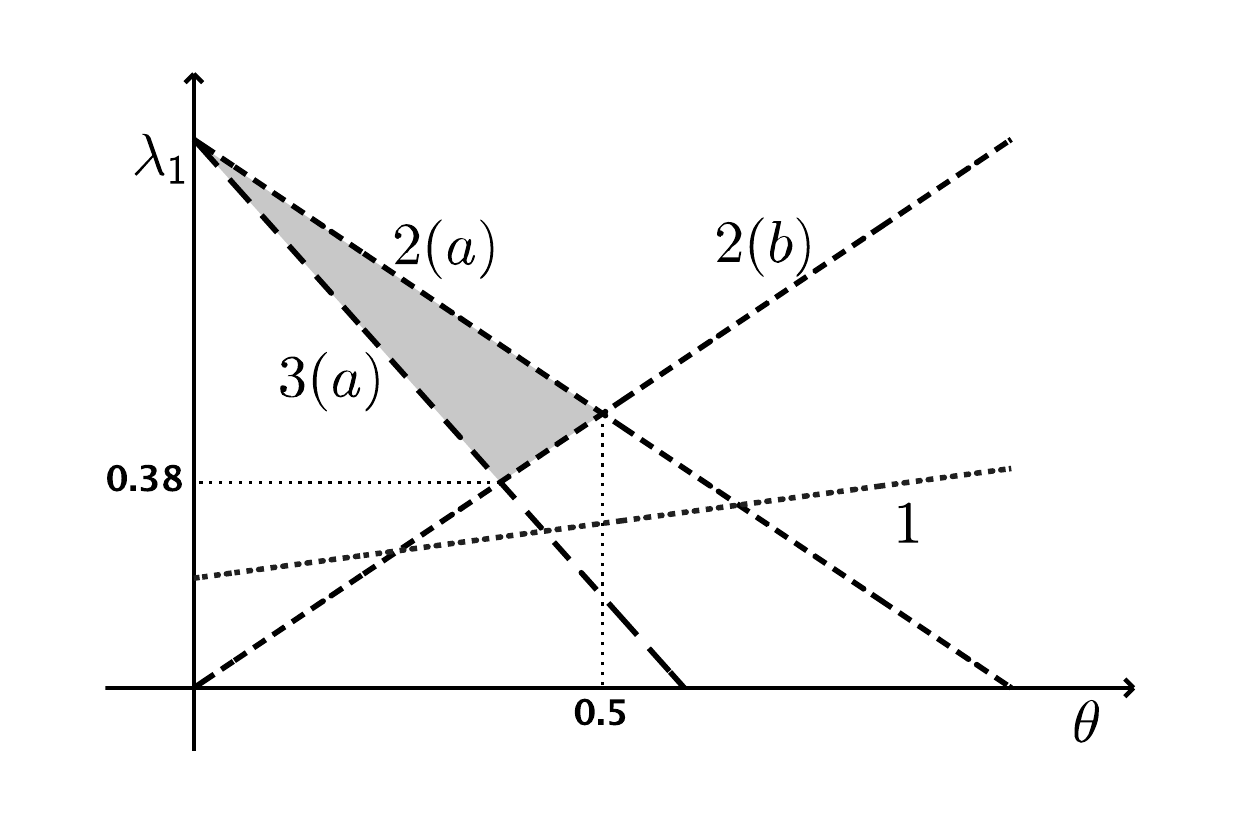}
    
  \end{minipage}
  \begin{minipage}{0.5\textwidth }
    \centering
    (b) First and Second-Step Bandwidths

	\includegraphics[width=3in]{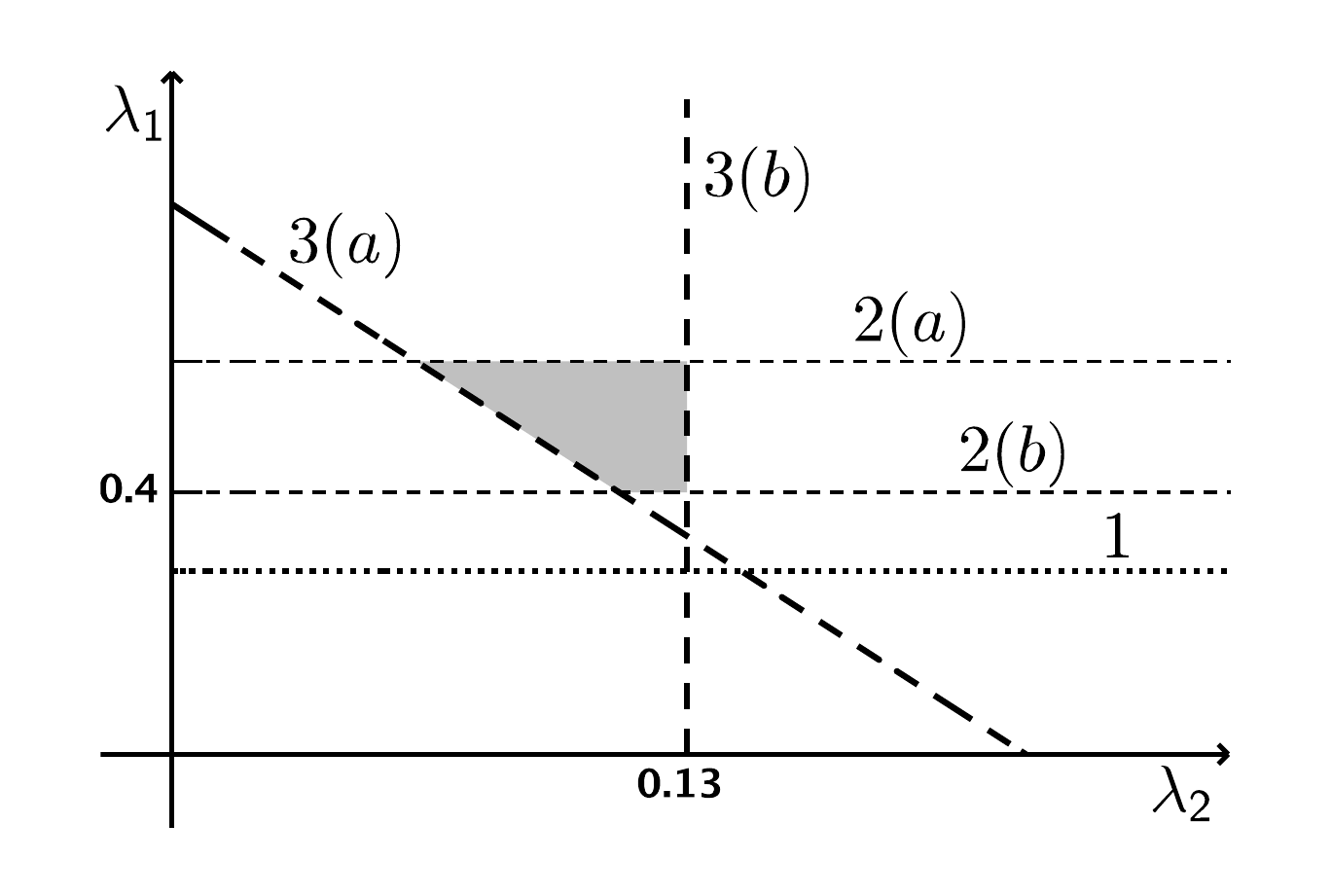}
	
  \end{minipage}%

\caption*{\footnotesize
Notes:
The diagram shows the rate conditions of Theorem \ref{theo_srd_kinf_est_int} 
applied to the case where
$h_{1j}=n^{-\lambda_1}$ for all $j$, 
$h_2=n^{-\lambda_2}$,
and
$K=n^{\theta}$. 
Condition 1, that is 
$\left( K n \overline{h}_1 \right)^{1/2} \overline{h}_1^{\rho_1 + 1} = O(1)$,
is equivalent to 
$\lambda_1 \geq (1+\theta)/(3+2 \rho_1)$;
condition 2(a):  
${ K^{1/2} \log n}/{ \left( n \overline{h}_1 \right)^{1/2} } =o(1)$
$\Leftrightarrow$
$\lambda_1 < 1 - \theta$;
condition 2(b): 
$K\overline{h}_1=O(1)$
$\Leftrightarrow$
$\lambda_1 \geq \theta$;
condition 3(a):
$\left( K n \overline{h}_1 \right)^{1/2}  h_2^{ \rho_2+1 } =O(1)$
$\Leftrightarrow$
$\lambda_1 \geq 1+\theta - 2 \lambda_2 (\rho_2 + 1)$;
and condition 3(b): 
$1 / K h_2^{3} =O(1)$
$\Leftrightarrow$
$\lambda_2 \leq \theta/3$.
Panel (a) illustrates the rate conditions on the first-step bandwidth and number of cutoffs 
$(\lambda_1, \theta)$ for  $\rho_1=1$, $\rho_2=3$,
and $\lambda_2=\theta/3$, so that $h_2=K^{-1/3}$ and condition 3(b) is satisfied.
Panel (b) displays the rate conditions on the bandwidths $(\lambda_1,\lambda_2)$ given 
$\theta=0.4$, $\rho_1=1$, and $\rho_2=3$.
}
\end{figure}

Implementation of Theorem \ref{theo_srd_kinf_est_int} requires the researcher to choose
$\rho_1 \in \mmz_+$, $h_{1j}>0 ~\forall j$, $\rho_2 \in \mmz_+$, $h_2>0$, and $k(\cdot)$.
A theory of optimal choice of these tuning parameters is beyond the goals of this paper.
Optimal choice of bandwidths is an  interesting topic for future research, because optimality in the multi-cutoff case
would account for: 
(i) the interaction between first and second-stage bandwidths;
(ii) the variance reduction from overlapping estimation windows at consecutive cutoffs; 
and
(iii) the recent advances of robust bias-corrected inference and coverage-error optimal bandwidths by \cite{calonico2018optimal}.

The IK bandwidth formula may produce first-step bandwidths with an incorrect rate of convergence.
For example, if $\rho_1=1$, these bandwidths converge to zero at $n^{-0.2}$, 
which is not fast enough if $\rho_2=3$ (Figure \ref{figure_rates}).
A simple way to correct for this is to adjust the bandwidths 
by multiplying them by $n^{0.2-\lambda_1}$ for $\lambda_1 \geq \theta$, so that their rate becomes $n^{-\lambda_1}$.
Conditions 2(a) and (b) imply that $\theta$ is never bigger than $0.5$ regardless of $\rho_1$, $\rho_2$, and $\lambda_2$.
Thus, the smallest value for $\lambda_1$ consistent with these restrictions is $0.5$. 
The same idea applies to the coverage-error optimal bandwidths by \cite{calonico2018optimal}, which converge to zero at rate $n^{-0.25}$, and need to be adjusted.

\label{parag:beta:binary} In certain cases, the $\beta$ function may depend on less than the three arguments $(c,d,d')$.
For example, in the Medicaid application, the treatment is binary and 
$\beta$ is only a function of $c$.
This is a particular case of the  theory in this section.
The only rate condition that changes is condition 3(b).
It becomes $1 / K h_2 =O(1)$,
or $\lambda_2 \leq \theta $ in terms of Figure \ref{figure_rates}.
A non-empty feasible set of bandwidth choices
requires $\rho_2 \geq 1$, as opposed to $\rho_2 \geq 3$ in the general case.

\label{parag:bandw:choice} A simple recommendation to implement Theorem \ref{theo_srd_kinf_est_int} is 
to use the edge kernel, first-step rate-adjusted IK bandwidths for each cutoff, 
and a second-step bandwidth $h_2$ that minimizes the MSE of estimation.
First, use observations pertaining to each cutoff $j$, compute the IK bandwidth $h_{1j}^{ik}$
for sharp RD and local-linear regression; 
adjust the rate of the bandwidths so that $h_{1j}=h_{1j}^{ik} \times n^{-0.3}$.
Second, create a grid of possible values for $h_2$.
For each value on the grid, compute
$\ha{\mu}^c(h_2)$ using the edge kernel, the choices of $h_{1j}$ given above,
$\rho_1=1$, and $\rho_2=3$ (or $\rho_2 = 1$ in the binary treatment case).
Similarly, compute
$\ha{\mu}^{c'}(h_2)$ using the edge kernel, the choices of $h_{1j}$ given above,
$\rho_1=2$, and $\rho_2=4$ (or $\rho_2 = 2$ in the binary treatment case).
Use Equation \ref{est:var:c} to estimate the variance of $\ha{\mu}^c(h_2)$ and call it $\ha{\m{V}}_n^c(h_2)$.
Evaluate the approximated MSE of $\ha{\mu}^{c}(h_2)$ by 
$(\ha{\mu}^{c}(h_2) - \ha{\mu}^{c'}(h_2))^2 + \ha{\m{V}}_n^c(h_2)$.
Choose the bandwidth value on the grid that minimizes the MSE and call it $h_2^*$.
The bias-corrected estimate is $\ha{\mu}^{c'}(h_2^*)$, and its variance estimate is
$\ha{\m{V}}_n^{c'}(h_2^*)$.

It may not be immediately clear that root-$n$ is the fastest estimation rate achievable in a setting where 
both $K$  and $n$ grow large.
The double asymptotic setting is conceptually different from the usual asymptotic setting where only $n \to \infty$ and 
non-parametric averages are estimable at root-$n$.
Estimation rates depend not only on bandwidth choices, but also on how fast $K$ grows, relative to $n$.
Similar examples in econometrics include panels with a large number of observations and time periods, 
and asymptotics with many instruments.
The following theorem demonstrates that the minimax optimal rate of estimation of $\mu^c$ is indeed root-$n$,
 as long as first-step bandwidths converge to zero at $1/K$ rate.

\begin{theorem}\label{theo_srd_kinf_minimax}
Let $\m{P}$ be the class of models generating potential outcomes $\{ Y_i(d) \}_{d \in \m{D}}$ and forcing variables $X_i$.
For a schedule of cutoffs and doses, $\{ \bc_j \}_{j=1}^K$, observed data $(Y_i,X_i,D_i)$ are generated iid from $P \in \m{P}$
as described in Section \ref{sec_case1}. 
Assume that 
(i) each model $P \in \m{P}$ satisfies Assumptions \ref{assu_srd_est_fx}-\ref{assu_srd_kinf_est_int};
(ii) $f(x)$ and $\sigma^2(x,d)$ are bounded away from zero uniformly in $\m{P}$;
(iii) the following functions are bounded uniformly in $\m{P}$:
$\nabla_x^{\rho}\sigma^2(x,d) ~~ \forall \rho \leq 1$ ,
$\nabla_x^{\rho} f(x)~~ \forall \rho \leq 1$,
$\nabla_x^{\rho} R(x,d) ~~ \forall \rho \leq \bar{\rho}$,
$\nabla_d^{\rho} R(x,d) ~~ \forall \rho \leq \bar{\rho}$,
where $\bar{\rho} = \max\{\rho_1+2, \rho_2+2 \}$;
and
(iv) there exists $M \in (0,\infty)$ such that 
$\mmp[|Y_i(d) - R(X_i,d)| < M]=1 ~~ \forall d \in \m{D}$ uniformly in $\m{P}$.
Then, for any $\epsilon>0$, there exists $\eta >0$ such that
\begin{gather}
\inf_{\ti{\mu}} \sup_{P \in \m{P}} \mmp_{P}  \left[ \sqrt{n} | \ti{\mu} - \mu^c(P) | > \epsilon/2 \right]  \geq \frac{1}{4 \eta} ~\text{ for large } n.
\label{theo_srd_kinf_minimax_a}
\end{gather}
The inf is taken over all estimators $\ti{\mu}$ built using the observed data $(Y_i,X_i,D_i)$, $i=1,\ldots,n$;
$\mu^c(P)=\int \omega^c(\bc) \beta(\bc;P) ~d \bc$ with $\beta(\bc;P) = \mme_P [Y_i(d')-Y_i(d) |X_i=c ]$;
and
$\mmp_{P}$ and $\mme_{P}$ denote the probability and expectation under model $P \in \m{P}$.  

Assume the conditions of Theorem \ref{theo_srd_kinf_est_int},
and that first-step bandwidths satisfy $\overline{h}_{1}  = O(K^{-1})$.
Consider the estimator $\ha{\mu}^c$ defined in Equation \ref{est:mu:c}.
For any small $\delta>0$, there exists large $\epsilon \in (0,\infty)$ such that
\begin{gather}
\sup_{P \in \m{P}} \mmp_{P}  \left[ \sqrt{n} | \ha{\mu}^c - \mu^c(P) | > \epsilon \right] < \delta  ~ \text{ for large } n.
\label{theo_srd_kinf_minimax_b}
\end{gather}
\end{theorem}

Equation $\ref{theo_srd_kinf_minimax_a}$ shows that no estimator converges faster than $\sqrt{n}$ uniformly over $\m{P}$.
Equation $\ref{theo_srd_kinf_minimax_b}$ says the estimator proposed in Theorem \ref{theo_srd_kinf_est_int} converges at root-$n$ uniformly over $\m{P}$ as long as first-step bandwidths converge to zero at $1/K$ rate.
Therefore, root-$n$ is the minimax optimal rate of convergence in the non-parametric estimation of ATE in RDD with many thresholds.
Authors have previously analyzed minimax optimality of non-parametric estimators of a regression function at a boundary point,
for example, \cite{cheng1997} and \cite{sun2005}. 
Theorem \ref{theo_srd_kinf_minimax} is novel because it combines boundary points 
to estimate averages of non-parametric regression functions.

\section{Fuzzy Case with Multiple Cutoffs}\label{sec_case3}
\indent

This section relaxes the sharp assignment mechanism of previous sections and studies the fuzzy RDD case.
The analysis focuses on multiple cutoffs,  but $K$ is finite as opposed to approaching infinity as in Section \ref{sec_case2}.
This makes the exercise more tractable, because the number of compliance cases grows super-exponentially with the number of cutoffs.
In contrast to the sharp case, non-parametric identification of local effects in the fuzzy case is impossible.
As a result, inference methods in this section rely on a second heterogeneity assumption, namely, the treatment effect function is assumed parametric.
Section \ref{sec_fuzzy_infer}, in the supplemental appendix, provides practical guidelines to compute an MSE-optimal ATE estimator, and demonstrates asymptotic normality.

In the sharp RDD case, all individuals with forcing variable equal to $x$ receive the same treatment $D(x)$ (Equation \ref{def_treat_sch}).
In the fuzzy RDD case, many of these individuals may receive treatments different
from $D(x)$.
In the high school
assignment example, students may choose to go to a school that is not the best school for which they are eligible.
For instance, a student may want to
attend the same high school as a certain friend or sibling.
Another example is given by \cite{garibaldi2012}.
In their study, a schedule of tuition subsidies applies to most students at Bocconi University, but the university reserves the right
to grant certain students different subsidies after reassessing their ability to pay.\footnote{The
source of fuzziness varies across applications.
One example is the case where
the assignment of individuals into different treatments
is made through a matching mechanism, and the econometrician does not observe
all the individual characteristics used in the matching algorithm. This is the reason
why the RDD of PU is fuzzy:
based on
the entire distribution of test scores and preferences, the central planner ranks students by their test scores
and assigns each one to her preferred school among schools with vacancies.
}

The fuzzy RDD case is modeled in terms of a potential treatment assignment framework.
A potential treatment assignment function $\m{U}:\m{X} \to \m{D}$ describes the treatment received for every value of the forcing variable $x \in \m{X}$.
For simplicity, these functions are assumed to belong to the following class:
\begin{gather}
\m{U}^* =
\Bigg\{ \m{U}: \m{X} \to \m{D} ~:~
 \m{U}(x) = \sum\limits_{j=0}^K u_{j} \mmi \left\{c_j \leq x < c_{j+1} \right\}
\nonumber
\\* 
\hspace{1cm} \text{ for some }
u_j \in \{d_0, \ldots, d_K \}, j=0,\ldots,K
\Bigg\}.
\label{eq_def_ustar}
\end{gather}

Sharp RDD is the particular case where the individual potential treatment assignment function $\m{U}_i$ is the same for every individual $i$, that is, $\m{U}_i(x)=D(x) ~ \forall i$
with $D(x)$ defined in Equation \ref{def_treat_sch}.
In the fuzzy case, $\m{U}_i$ is sampled iid from a distribution of functions with support in $\m{U}^*$.
Potential treatment functions $\m{U}_i (x)$ are unobserved, but the treatments received
are observed and given by
\begin{gather*}
D_i =
 \sum\limits_{j=0}^K \m{U}_i (c_j) \mmi \left\{c_j \leq X_i < c_{j+1} \right\}.
\end{gather*}

Using classic definitions of compliance behaviors (\cite{imbens1997}),
three types of compliance groups are defined in terms of \textit{changes} in treatment eligibility.
``Never-changers'' are those whose treatment received never changes when eligibility changes.
The treatment received by ``ever-compliers'' or ``ever-defiers'' changes at least once
when eligibility changes.
Ever-compliers are those whose treatment received changes if and only if it changes to the
treatment dose for which they become eligible.
Ever-defiers change to a treatment dose different from the one for which they become eligible.
In the case of one cutoff and two treatments, the definition of ever-complier (ever-defier) is equivalent to the classic definition of complier (defier) of \cite{imbens2008regression}.

The three compliance groups are measurable events that partition the population of individuals with
$\mathbf{G}_{nc}$ denoting never-changers, $\mathbf{G}_{ec}$  ever-compliers, and
$\mathbf{G}_{ed}$ ever-defiers.\footnote{
These definitions allow for non-monotonic treatment schedules; for example,
the average class-size varies non-monotonically across cutoffs on enrollment (\cite{angrist1999}).
Table \ref{table_compliance} in Section \ref{sec_fuzzy_table} of the supplemental appendix
illustrates these definitions of compliance groups using a simple example with
3 treatments and 2 cutoffs.
}
\begin{gather}
\mathbf{G}_{nc} = \Big\{ \m{U}_i  \in \m{U}^* : \left\{ j :~  \m{U}_i(c_{j-1}) \neq \m{U}_i(c_j) \right\} = \emptyset \Big\}
\label{def_neverchanger} \\
\mathbf{G}_{ec} =
\Big\{
 \m{U}_i  \in \m{U}^* :
  \left\{j:~  \m{U}_i(c_j) = D(c_j)  \right\}
  \supseteq
  \left\{j:~  \m{U}_i(c_{j-1}) \neq \m{U}_i(c_j) \right\} 
    \neq \emptyset 
\Big\}
\label{def_evercomplier} \\
\mathbf{G}_{ed} =
\Big\{
 \m{U}_i  \in \m{U}^* : 
 \left\{~
  \left\{j:~  \m{U}_i(c_j) \neq D(c_j)  \right\}
  \cap
  \left\{j:~  \m{U}_i(c_{j-1}) \neq \m{U}_i(c_j) \right\} 
  ~\right\}
    \neq \emptyset 
\Big\}
\label{def_everdefier}
\end{gather}
where  $\emptyset$ denotes empty set.

In the high school assignment case, an example of a never-changer is a student who strongly
prefers the high school with the lowest admission cutoff and attends that high school
even if she is admitted to better schools.
An example of an ever-complier is a student
who attends the best school into which she is admitted,
or a student who chooses the best school among the nearby schools.
Suppose a student has rational preferences and is never indifferent.
Assume her choice set is equal to those schools with admission cutoffs that are less than or equal to her test score.
Then, such a student is never an ever-defier.
In other words, as her test score increases, a new school is added to her choice-set of schools;
she either chooses to go to the new school for which she becomes eligible,
or she stays
at the school which she preferred prior to the increase in her choice-set.
Thus, it seems natural to rule out ``ever-defiers'' in this and other applications.

Never-changers do not produce changes in treatments, so there is no identification on them.
For ever-compliers, there are multiple possible changes in treatment at a given cutoff,
and ever-compliers may differ in terms of the treatments they comply with.
For example, the student who is willing to attend the best school possible complies with all
changes in treatment eligibility. 
On the other hand, the student who is willing to attend the best possible school within a certain distance from home only complies with some of the changes in treatment eligibility.
Therefore, besides no-defiance, identification also requires the heterogeneity of ever-compliers  to be restricted.

Assumption \ref{assu_frd_id_nodef} generalizes
the sufficient conditions for identification
on compliers in
the one-cutoff case (\cite{hahn2001id} and \cite{dong_alternative}).
In addition, it restricts the heterogeneity on ever-compliers.

\loadaspt{00H} 

A fuzzy assignment produces several different treatment changes at each cutoff, even after ruling out ever-defiers.
 The researcher only observes one aggregate change in $Y_i$ at each cutoff, 
 but there are several treatment effects on ever-compliers to be identified at that cutoff.
Theorem \ref{theo_frd_id} below shows that identification of these effects is not possible without further restricting the class of functions
$\beta_{ec}(\bc)$.
Economic theory or \textit{a priori} knowledge guides the
choice of a functional form that credibly summarizes the heterogeneity of treatment effects.
For example, the principal-agent model of \cite{bajari2017} yields a functional form to study reimbursement of hospitals by insurers.
The second heterogeneity assumption (Assumption \ref{assu_srd_id_param})
 restricts the treatment effect function on ever-compliers
to a finite-dimensional vector space of functions.
\loadaspt{00G} 

In this case, the ATE on ever-compliers is a linear combination of the true parameter vector $\btheta_0^{ec}$.
For a counterfactual distribution $F$ chosen by the researcher,
\begin{align}
\mu^{ec}(F) = & \int \beta(\boup{c};\btheta_0^{ec})  ~dF(\bc)
\\
= & \underset{\equiv \boup{Z}(F)}{\underbrace{ \int \left[\boup{\m{W}}(c,d')-\boup{\m{W}}(c,d) \right]' ~dF(\bc)}}  \btheta_0^{ec}
\\
= & \boup{Z}(F) \btheta_0^{ec}.
\label{eq_mu3}
\end{align}

Theorem \ref{theo_frd_id} shows that the observed change in average outcome
at a given cutoff is a weighted average
of treatment effects on ever-compliers who switch from various doses into the dose
of eligibility at that cutoff.
Assumption \ref{assu_srd_id_param} and variation in cutoff characteristics
are sufficient conditions for identification. Conversely,
identification on ever-compliers implies that
$\beta_{ec}(\bc)$ belongs to a finite-dimensional class of functions.
\begin{theorem}\label{theo_frd_id}
Under Assumption \ref{assu_frd_id_nodef}, for $j=1,\ldots,K$,
\begin{gather*}
 B_{j} = \sum\limits_{l =0, l \neq j}^{K} \omega_{j,l} \beta_{ec}(c_{j},d_{l},d_{j})
\end{gather*}
where $B_{j}$ is defined in Equation \ref{eq_Bpj},  and  
\begin{gather}
\omega_{j,l}=\lim_{e \downarrow 0} \left\{
\mmp[D_i = d_{l} | X_i=c_{j}-e] -
\mmp[D_i = d_{l} | X_i=c_{j}+e]
\right\},
\nonumber 
\end{gather}
for  $l=0,1,\ldots,K, l \neq j$.

Moreover, suppose $\beta_{ec}$ belongs to the class of functions $\m{H}$
defined in Assumption \ref{assu_srd_id_param} with $q \leq K$.
Define
\begin{gather}
\widetilde{W}_{j} =
\sum_{l =0, l \neq j}^{K} \omega_{j,l} \left[ \boup{\m{W}}(c_{j},d_{j}) - \boup{\m{W}}(c_{j},d_{l}) \right]
\label{eq_def_Wtilde}
\end{gather}
for the vector-valued function
$\boup{\m{W}}(c,d)$ of Assumption \ref{assu_srd_id_param};
build a $K \times q$ matrix
$\widetilde{\boup{W}}$ by stacking $\widetilde{W}_{j}$,
and $\boup{ B}$ by stacking $B_{j}$.
If $\widetilde{\boup{W}}' \widetilde{\boup{W}}$ is
invertible, then $\beta_{ec}(\bc)$ is identified and equal to
\[
\beta_{ec}(\bc) = \left[ \boup{\m{W}}(c,d') - \boup{\m{W}}(c,d) \right]' 
\left( \widetilde{\boup{W}}' \widetilde{\boup{W}} \right)^{-1}
\widetilde{\boup{W}}' \boup{ B}.
\]

Conversely, suppose $\beta_{ec}$ belongs to some class of functions $\ti{\m{H}}$, and treatment effects on ever-compliers are identified at the $p>K$ cutoff-dose values
$\{\tilde \bc: \tilde \bc = (c_j,d_l,d_j) \text{ with } \omega_{j,l}>0 \}$
of every possible fuzzy assignment generated from the given schedule of cutoffs $\{ \bc_j \}_{j=1}^K$. Then,
the class of functions $\ti{\m{H}}$ is ``finite dimensional'' in the sense
that
\begin{gather*}
\m{G} =
\Big\{
\Big( \beta(\tilde \bc_1),\ldots, \beta(\tilde \bc_p) \Big) : \text{ for } \beta \in \ti{\m{H}}
\Big\} \subseteq \mmr^p
\end{gather*}
has $dim \m{G} \leq K$
for every fuzzy assignment $\{\tilde \bc_j \}_{j=1}^p$
generated from $\{ \bc_j \}_{j=1}^K$.
\end{theorem}

Theorem \ref{theo_frd_id} reveals the requirement of stronger functional form assumptions on $\beta_{ec}(\bc)$ even for identification of local effects in the fuzzy case with a finite number of multiple cutoffs.
For example, identification is not possible when $\ti{\m{H}}$ is the class of all smooth functions studied in the non-parametric case of Section \ref{sec_case2}.
The result is striking because non-parametric identification of local effects is possible both in the sharp case with a finite number of cutoffs and in the fuzzy case with a single cutoff. 
It is likely possible to obtain non-parametric identification of $\beta_{ec}(\bc)$ under a large variation of cutoff-dose values. 
The function $\beta_{ec}(\bc)$ may be approximated by a sequence of parametric functions from Assumption \ref{assu_srd_id_param},
 where $q$ grows to infinity more slowly than $K$,
so to keep $dim \m{G} \leq K$ as $K\to \infty$. 
In this paper, the number of cutoffs is kept finite for simplicity, and the case with large $K$ is deferred to future work.

Theorem \ref{theo_frd_id} also clarifies the interpretation of
two-stage least squares (2SLS) estimates in applications of fuzzy RD with multiple cutoffs,
a common practice in applied work.
The practice consists of using $D(X_i)$ as an instrument for $D_i$ in the regression
of $Y_i$ on a constant, $D_i$, and $X_i$. See \cite{angrist2008mostly} for a discussion.
In the single-cutoff case,
both the non-parametric RD estimator and 2SLS applied to a neighborhood of the cutoff are consistent to the average treatment effect on compliers (\cite{hahn2001id}).
To my knowledge, such an equivalence has never been studied in the multiple-cutoff case.
Nevertheless, many important applications have multiple-fuzzy cutoffs and use 2SLS;
for example, \cite{angrist1999}, \cite{chen2008vanderklaauw}, and \cite{hoekstra2009effect}. 
The 2SLS estimator is consistent for a data-driven
weighted average of treatment effects
on ever-compliers
as long as
a sufficiently flexible specification is used; for example, cutoff fixed-effects or varying slopes.
The economic meaning of the 2SLS estimands depends crucially on the choice of
such a weighting scheme. 
Unless a parametric functional form is imposed on $\beta_{ec}(\bc)$,  or there is large variation in cutoff-doses,
only a data-driven weighted average of $\beta_{ec}(\bc)$ is identified.
In other words,
if $\beta_{ec}(\bc)$ is non-parametric and there are only a few cutoffs,
the researcher does not have control over the weighting scheme,
and 2SLS estimates don't have a clear interpretation.

Theorem \ref{theo_frd_id} leads to a two-step estimation procedure for $\btheta^{ec}_0$ and 
$\mu^{ec}$.
The mechanics are similar to the previous sections, 
so I omit the details from the main text for brevity. 
In the first step,
the researcher estimates the jump discontinuity of the vector $[Y_i ~~ \bmW(X_i,D_i)' ]'$
using LPRs at each cutoff to obtain $[ \ha{B}_j ~~ \ha{\ti{W}}_j']'$.
In the second step,
a regression of  $\ha{B}_j$ on $\ha{\bWt}_j'$  obtains $\ha{\btheta}^{ec}$.
The ATE estimator is $\ha{\mu}^{ec} = \bZ(F) \ha{\btheta}^{ec}$.
Estimation precision varies across cutoffs, 
and the parametric form of $\beta_{ec}$
allows us to optimally combine different cutoffs to minimize the MSE of  $\ha{\btheta}^{ec}$.
The researcher can simply re-weight the second-step regression by the inverse of the MSE matrix of the first-step estimators.
Section \ref{sec_fuzzy_infer}  in the supplemental appendix
delineates the estimation and inference procedures
of $\btheta^{ec}_0$ and $\mu^{ec}$
with practical steps.

\section{Simulations}\label{sec_simul}
\indent

In this section, Monte Carlo simulations illustrate the finite sample behavior of the ATE estimator proposed in Section \ref{sec_case2}.
The analysis considers estimation precision and coverage of confidence intervals  for different choices of tuning parameters and a non-linear specification for $\beta$.
As predicted by Theorem \ref{theo_srd_kinf_est_int},
 an incorrect choice of the second-step polynomial degree 
leads to severe bias and extremely poor coverage of confidence intervals.
Moreover, first-step bandwidths that imply overlapping estimation windows 
produce lower MSE than cases with no overlap, regardless of other tuning parameters.

The DGP draws $n$ iid observations of $(X_i,\eps_i)$ where $X_i$ is uniformly distributed over $[0,1]$,
$\eps_i$ is  normally distributed with zero mean and unit variance, and these variables are independent of each other.
There are $K$ cutoffs $c_j = j/(K+1)$, $j=1, \ldots, K$, on the unit interval $[0,1]$.
The number of cutoffs is $K=\lfloor n^{0.4} \rfloor$, where $\lfloor a \rfloor$ denotes the largest integer smaller than or equal to $a$.
An individual with forcing variable $X_i$ receives a treatment dose equal to $D(X_i)$ as in Equation \ref{def_treat_sch}.
The dose increases by one unit at each cutoff, starting at $d_0 = 1$ and ending at $d_K=K+1$. 
The outcome variable is $Y_i = \phi(X_i) D(X_i) + \eps_i$ where
$\phi(X_i) = 15 X_i^3 + 7.5 X_i^2 - 18.75 X_i +2.125$. 
This implies that $\beta(\bc)=\phi(c)(d'-d)$, which falls into the binary treatment case (see discussion on page \pageref{parag:beta:binary}).
Consider a counterfactual policy that uniformly increases treatment doses by one unit.
The ATE parameter $\mu$ is the integral of $\phi(c)$ over $c \in [0,1]$, which equals $-1$ in this case.
 
Estimation follows the procedure suggested in Section \ref{sec_case2}.
For given bandwidth choices $h_1$ and $h_2$, 
I compare the ATE estimator $\ha{\mu}$ that uses $\rho_1=\rho_2=1$,
to the bias-corrected ATE estimator $\ha{\mu}^{bc}$ that uses $\rho_1=\rho_2=2$. 
To emphasize the importance of the second step, I also compute a naive ATE estimator that simply averages the first-step estimates.
The naive and bias-corrected naive estimators, respectively $\ti{\mu}$ and $\ti{\mu}^{bc}$,
are constructed as $\ha{\mu}$ and $\ha{\mu}^{bc}$ except for the tuning parameters in the second step.
Both naive estimators use $\rho_2=0$ and  $h_2=\infty$.
To examine the effect of overlapping estimation windows in the first step, 
I compare estimators for two choices of $h_1$.
The first choice is the largest possible bandwidth $h_{1}=1/(K+1)$, which leads to maximum overlap.
The second choice is the largest possible bandwidth with no overlap, that is, $h_{1}=0.5/(K+1)$.
Finally, I study the effects of ten different choices for the second-step bandwidth, $h_2 \in \{3/(K+1), \ldots, 12/(K+1)\}$.
All choices of tuning parameters satisfy the rate conditions of Theorem \ref{theo_srd_kinf_est_int}
and produce a convergence rate of root-$n$ for $\ha{\mu}$ and $\ha{\mu}^{bc}$.
The Monte Carlo experiment simulates 10,000 draws of an iid sample with $n \in \{ 1789, 10120, 27886, 57244,100000\} $ and
$K \in \{20, 40, 60, 80, 100\}$, respectively. 
Section \ref{sec:supp:datadrivenh} in the supplemental appendix repeats the experiment with data-driven bandwidth choices,
following the bandwidth rules proposed on page \pageref{parag:bandw:choice} (Section \ref{sec_case2}).

\begin{table}
\begin{center}
\caption{Precision of Estimators - Choice of $h_1$}
\label{tab:choiceh1}
\scriptsize
\begin{tabular}{cc|cccc|cccc|cccc}
\hline \hline 
& & \multicolumn{4}{|c}{Bias} & \multicolumn{4}{|c}{Variance}  & \multicolumn{4}{|c}{MSE} \\ 
$n$ & $K$ &  $\ha{\mu}$ &  $\ha{\mu}^{bc}$ & $\ti{\mu}$ &  $\ti{\mu}^{bc}$ & 
                     $\ha{\mu}$  &  $\ha{\mu}^{bc}$ & $\ti{\mu}$ &  $\ti{\mu}^{bc}$ & 
                     $\ha{\mu}$  &  $\ha{\mu}^{bc}$ & $\ti{\mu}$ &  $\ti{\mu}^{bc}$ \\ 
\multicolumn{2}{c|}{\underline{\textit{Overlap}}}   
 & & & & & & & & & & & & \\   
 1789 & 20 
& 0.0617 & -0.0015  & -0.2504  & -0.2390  
& 0.0079 & 0.0164  & 0.0073  & 0.0110  
& 0.0117 & 0.0164  & 0.0700  & 0.0681  
\\ 
10120 & 40 
& 0.0206 & 0.0003  & -0.1245  & -0.1216  
& 0.0013 & 0.0022  & 0.0012  & 0.0017  
& 0.0017 & 0.0022  & 0.0167  & 0.0165  
\\ 
27886 & 60 
& 0.0095 & -0.0002  & -0.0834  & -0.0821  
& 0.0005 & 0.0007  & 0.0004  & 0.0006  
& 0.0005 & 0.0007  & 0.0074  & 0.0074  
\\ 
57244 & 80 
& 0.0056 & -0.0001  & -0.0625  & -0.0618  
& 0.0002 & 0.0003  & 0.0002  & 0.0003  
& 0.0003 & 0.0003  & 0.0041  & 0.0041  
\\ 
100000 & 100 
& 0.0038 & -0.0001  & -0.0499  & -0.0496  
& 0.0001 & 0.0002  & 0.0001  & 0.0002  
& 0.0001 & 0.0002  & 0.0026  & 0.0026  
\\ 
\multicolumn{2}{c|}{\underline{\textit{No Overlap}}}   
 & & & & & & & & & & & & \\   
 1789 & 20 
& 0.0698 & -0.0019  & -0.2423  & -0.2402  
& 0.0139 & 0.0386  & 0.0122  & 0.0288  
& 0.0188 & 0.0386  & 0.0710  & 0.0865  
\\ 
10120 & 40 
& 0.0227 & -0.0002  & -0.1225  & -0.1221  
& 0.0021 & 0.0051  & 0.0020  & 0.0043  
& 0.0026 & 0.0051  & 0.0170  & 0.0192  
\\ 
27886 & 60 
& 0.0106 & -0.0006  & -0.0824  & -0.0825  
& 0.0007 & 0.0017  & 0.0007  & 0.0016  
& 0.0009 & 0.0017  & 0.0075  & 0.0084  
\\ 
57244 & 80 
& 0.0062 & 0.0000  & -0.0620  & -0.0617  
& 0.0004 & 0.0008  & 0.0003  & 0.0008  
& 0.0004 & 0.0008  & 0.0042  & 0.0046  
\\ 
100000 & 100 
& 0.0039 & -0.0003  & -0.0498  & -0.0497  
& 0.0002 & 0.0005  & 0.0002  & 0.0004  
& 0.0002 & 0.0005  & 0.0027  & 0.0029  
\\ 
\hline \hline 
\end{tabular}

\caption*{\scriptsize
Notes: The table reports simulated bias, variance, and mean squared error (MSE) for four estimators
$(\ha{\mu}, \ha{\mu}^{bc} , \ti{\mu} ,\ti{\mu}^{bc})$, two choices of first-step bandwidth (overlap and no overlap),
and five sample sizes $n$ and respective numbers of cutoffs $K$.
The second-step bandwidth is set to $h_2=3/(K+1)$, which minimizes MSE of $\ha{\mu}$. Refer to Table \ref{tab:choiceh2}
for different choices of $h_2$.
The number of simulations is 10,000.
}
\end{center}
\end{table}

The bias and variance of all estimators converge to zero as the sample size increases,
regardless of the choice of $h_1$ (Table \ref{tab:choiceh1}).
The bias-correction of $\ha{\mu}^{bc}$ eliminates almost all the bias of $\ha{\mu}$, at the cost of a higher variance.
The naive estimator $\ti{\mu}$ oversmooths the second step beyond the conditions of Theorem \ref{theo_srd_kinf_est_int}.
As a result, the bias of $\ti{\mu}$ is substantially larger than that of $\ha{\mu}$.
Simply correcting for bias in the first step does not solve the problem, as the difference in bias between
$\ti{\mu}$ and $\ti{\mu}^{bc}$  is small.
First-step bandwidths that produce overlap (Table \ref{tab:choiceh1}, rows 1-5) yield approximately the same bias,
but substantially smaller variance,
compared to first-step bandwidths that produce no overlap (Table \ref{tab:choiceh1}, rows 6-10).

\begin{table}
\begin{center}
\caption{Precision of Estimators - Choice of $h_2$}
\label{tab:choiceh2}
\footnotesize
\begin{tabular}{c|cc|cc|cc|cc|cc|cc} 
\hline \hline 
& \multicolumn{6}{|c}{$(n,K)=(1789, 20)$}   
& \multicolumn{6}{|c}{$(n,K)=(10120, 40)$}   
\\ 
& \multicolumn{2}{|c}{Bias} & \multicolumn{2}{|c}{Variance}  & \multicolumn{2}{|c}{MSE} 
 & \multicolumn{2}{|c}{Bias} & \multicolumn{2}{|c}{Variance}  & \multicolumn{2}{|c}{MSE} 
\\ 
$h_2 \cdot (K+1)$ & $\ha{\mu}$ &  $\ha{\mu}^{bc}$ 
        & $\ha{\mu}$ &  $\ha{\mu}^{bc}$ 
        & $\ha{\mu}$ &  $\ha{\mu}^{bc}$ 
        & $\ha{\mu}$ &  $\ha{\mu}^{bc}$ 
        & $\ha{\mu}$ &  $\ha{\mu}^{bc}$ 
        & $\ha{\mu}$ &  $\ha{\mu}^{bc}$ 
\\ 
3 
& 0.0617 & -0.0015  
& 0.0079 & 0.0164  
& 0.0117 & 0.0164  
& 0.0206 & 0.0003  
& 0.0013 & 0.0022  
& 0.0017 & 0.0022  
\\ 
4 
& 0.1128 & -0.0012  
& 0.0079 & 0.0143  
& 0.0206 & 0.0143  
& 0.0376 & 0.0004  
& 0.0013 & 0.0020  
& 0.0027 & 0.0020  
\\ 
5 
& 0.1708 & -0.0014  
& 0.0079 & 0.0133  
& 0.0370 & 0.0133  
& 0.0584 & 0.0005  
& 0.0013 & 0.0019  
& 0.0047 & 0.0019  
\\ 
6 
& 0.2322 & -0.0014  
& 0.0079 & 0.0128  
& 0.0618 & 0.0128  
& 0.0826 & 0.0004  
& 0.0013 & 0.0019  
& 0.0081 & 0.0019  
\\ 
7 
& 0.2935 & -0.0015  
& 0.0079 & 0.0126  
& 0.0940 & 0.0126  
& 0.1097 & 0.0004  
& 0.0013 & 0.0019  
& 0.0133 & 0.0019  
\\ 
8 
& 0.3513 & -0.0015  
& 0.0079 & 0.0124  
& 0.1313 & 0.0124  
& 0.1392 & 0.0004  
& 0.0013 & 0.0019  
& 0.0207 & 0.0019  
\\ 
9 
& 0.4019 & -0.0015  
& 0.0080 & 0.0122  
& 0.1695 & 0.0122  
& 0.1707 & 0.0004  
& 0.0013 & 0.0018  
& 0.0304 & 0.0018  
\\ 
10 
& 0.4420 & -0.0015  
& 0.0080 & 0.0122  
& 0.2034 & 0.0122  
& 0.2036 & 0.0004  
& 0.0013 & 0.0018  
& 0.0427 & 0.0018  
\\ 
11 
& 0.4680 & -0.0015  
& 0.0081 & 0.0122  
& 0.2272 & 0.0122  
& 0.2375 & 0.0004  
& 0.0013 & 0.0018  
& 0.0577 & 0.0018  
\\ 
12 
& 0.4773 & -0.0015  
& 0.0082 & 0.0121  
& 0.2360 & 0.0121  
& 0.2720 & 0.0004  
& 0.0013 & 0.0018  
& 0.0753 & 0.0018  
\\ 
\hline \hline 
\end{tabular}

\caption*{\scriptsize
Notes: The table reports simulated bias, variance, and mean squared error (MSE) for two estimators
$(\ha{\mu}, \ha{\mu}^{bc})$, ten choices of second-step bandwidth ($h_2 \in \{ 3/(K+1), \ldots, 12/(K+1)\})$,
and the two smallest sample sizes $n$ and respective numbers of cutoffs $K$.
The first-step bandwidth is set to $h_1 =1/(K+1)$ (overlap).
Naive estimators $(\ti{\mu}, \ti{\mu}^{bc})$ are not in this table because they are not affected by the choice of $h_2$.
The number of simulations is 10,000.
}
\end{center}
\end{table}

Next, I study how the choice of $h_2$ affects precision of $(\ha{\mu}, \ha{\mu}^{bc})$
for a fixed choice of $h_1=1/(K+1)$ (Table \ref{tab:choiceh2}).
The smallest value for $h_2$ is $3/(K+1)$. 
This defines a second-step estimation window with at least three cutoffs to ensure invertibility of  matrices in the regressions.
The bias of $\ha{\mu}$ is substantially smaller when $h_2$ is set to its smallest value. 
All other measures are practically unaffected across different $h_2$.

The significant bias of the naive ATE estimators $\ti{\mu}$ and $\ti{\mu}^{bc}$ decreases the coverage of 95\% confidence intervals
as the sample increases (Table \ref{tab:ci}). 
The naive estimators oversmooth in the second step, and Theorem \ref{theo_srd_kinf_est_int} implies the bias grows faster than root-$n$.
For each of the four estimators, the confidence intervals equal the estimator plus or minus $1.96$ times
its standard error. The variance of estimators are obtained as described in Equation \ref{est:var:c}.
The bias-corrected ATE estimator $\ha{\mu}^{bc}$ produces confidence intervals with correct coverage for all samples sizes.
Although $\ha{\mu}$ yields intervals with average length smaller than $\ha{\mu}^{bc}$, the bias of $\ha{\mu}$
leads to a slightly lower coverage.

\begin{table}
\begin{center}
\caption{Coverage of 95\% Confidence Intervals}
\label{tab:ci}
\begin{tabular}{cc|cccc|cccc} 
\hline \hline 
& & \multicolumn{4}{|c}{\% Coverage} & \multicolumn{4}{|c}{ Avg. Length}   
\\ 
$n$ & $K$  
        & $\ha{\mu}$ &  $\ha{\mu}^{bc}$ 
        & $\ti{\mu}$ &  $\ti{\mu}^{bc}$ 
        & $\ha{\mu}$ &  $\ha{\mu}^{bc}$ 
        & $\ti{\mu}$ &  $\ti{\mu}^{bc}$ 
\\ 
1789 & 20 
& 0.8931 & 0.9546  & 0.1686 & 0.3834 
& 0.3526 & 0.5135  & 0.3368 & 0.4165 
\\ 
10120 & 40 
& 0.9087 & 0.9545  & 0.0568 & 0.1788 
& 0.1403 & 0.1850  & 0.1373 & 0.1656 
\\ 
27886 & 60 
& 0.9287 & 0.9489  & 0.0203 & 0.1005 
& 0.0834 & 0.1063  & 0.0822 & 0.0988 
\\ 
57244 & 80 
& 0.9292 & 0.9502  & 0.0113 & 0.0585 
& 0.0578 & 0.0724  & 0.0572 & 0.0686 
\\ 
100000 & 100 
& 0.9359 & 0.9492  & 0.0057 & 0.0353 
& 0.0436 & 0.0540  & 0.0432 & 0.0518 
\\ 
\hline \hline 
\end{tabular}

\caption*{\scriptsize
Notes: The table reports simulated percentage of correct coverage and average length of 95\% confidence intervals.
Confidence intervals are constructed 
using four estimators $(\ha{\mu}, \ha{\mu}^{bc} , \ti{\mu} ,\ti{\mu}^{bc})$.
They equal an estimator plus or minus its estimated standard deviation multiplied by $1.96$.
Coverage and average length are computed for five sample sizes $n$ and respective numbers of cutoffs $K$.
The first-step bandwidth is set to $h_1=1/(K+1)$ (overlap),
and the second-step bandwidth is set to $h_2=3/(K+1)$, which minimizes MSE of $\ha{\mu}$. 
The number of simulations is 10,000.
}
\end{center}
\end{table}

\section{Application}\label{sec_appli}
\indent

In this section, the methods proposed in this paper are illustrated
using the data from PU on high school assignments in Romania.\footnote{The data set is available online
in the supplemental materials of PU on the website of the \textit{American Economic Review}.}
Many policy questions demand an ATE of a continuous counterfactual distribution of treatments,
and this section provides an example of such a policy question.
The estimators designed for the sharp RDD case are consistent for ``Intent-to-Treat'' (ITT) average effects
when applied to the fuzzy data of PU.
In this application, the ITT effect measures the impact of
being assigned to a better school but not necessarily attending it.
The parametric methods of Section \ref{sec_case3}  yield noticeable efficiency gains in the estimation of the ATE on ever-compliers.
Treatment effects for ever-compliers reveal a heterogeneity pattern unlike the heterogeneity of ITT effects.

The administrative data from Romania cover 3 cohorts of 9th grade students for the years 2001, 2002, and 2003,
with a total of 334,137 observations.
The essential elements of the high school assignment in Romania are described below.
The assignment to high school is nationally centralized by the Ministry of Education. 
At the end of grade 8, students submit a transition score and a complete ranking of preferences for high schools.
The transition score is an average of the student's performance on a national exam taken in grade 8
and
the student's grade point average during grades 5-8.
The Ministry of Education  ranks students by their transition score and no other criteria. The mechanism
assigns the student ranked first to her most preferred school, the student ranked second
to her most preferred school, etc.
Students cannot decline their assignment,
and they have incentives to truthfully reveal their preference rankings.

The observed variables are the town and year of student $i$, the transition score $X_i$, the school the student is assigned to, and the student's score on the 
``baccalaureate exam.''
This is an exam taken at the end of high school, and the grade on the exam is the outcome variable $Y_i$.
The quality of school $j$ (treatment dose $d_j$) is measured by the average transition score of the students attending that school $j$.
The cutoff $c_j$ for admission into a school $j$ is equal to the minimum transition score
among the students that are assigned to that school $j$.
The student's preferences in high schools are not observed in the data, which makes the RDD fuzzy.
For example, a student may have a score greater than the cutoff for the best school in her town,
but still be assigned to a different school because of her personal preferences.
For a transition score $X_i$, the treatment dose of eligibility $D(X_i)$ is equal to the largest $d$ among those schools with admission cutoff $c$ less than $X_i$.
The treatment dose received $D_i$ coincides with the treatment dose of eligibility $D(X_i)$ for 40\% of the students in the sample. 
Thus, the assignment is  fuzzy, and causal inference beyond ITT effects requires the methods of Section \ref{sec_case3}.
Following PU, I drop observations with missing values for $Y_i$.
I also drop cutoffs without enough observations around them to carry out the matrix inversions of the local polynomial regressions.
The dropping of cutoffs leaves the empirical distribution of outcomes, forcing variable, cutoffs, and treatment doses practically unchanged.
The estimation sample has 588 cutoffs with a total of 179,995 individuals from 769 schools in 121 towns and 3 years.
The variation of cutoff and dose values is displayed in Figure \ref{figure_schedule}.

\begin{figure}[H]
\caption{Variation in Cutoff and Dose Values}
\label{figure_schedule}
\begin{center}
  \begin{minipage}{3in}
    \centering
    (a)

    \includegraphics[width=3in]{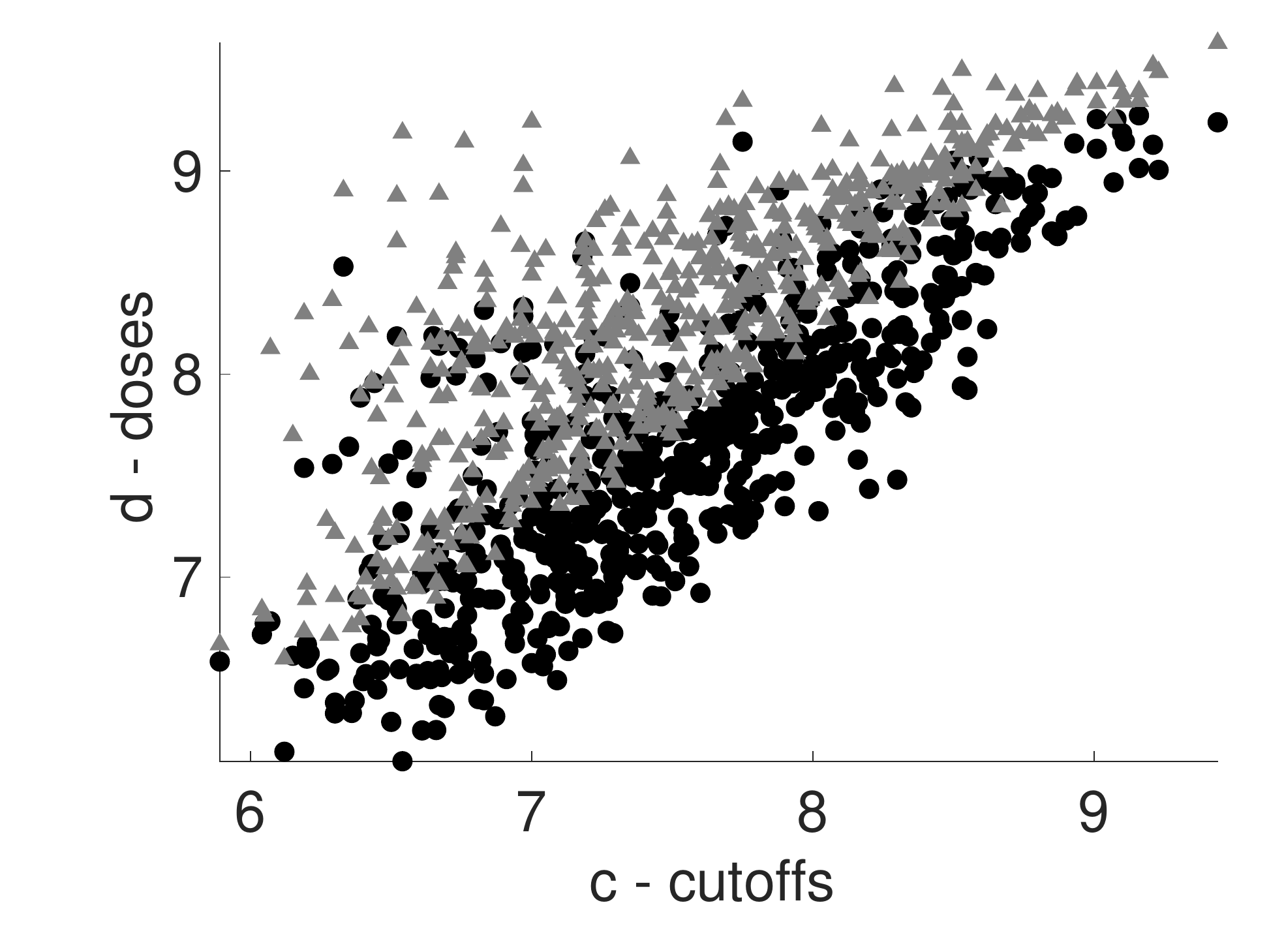}
  \end{minipage}%
  \begin{minipage}{3in}
    \centering
    (b)

    \includegraphics[width=3in]{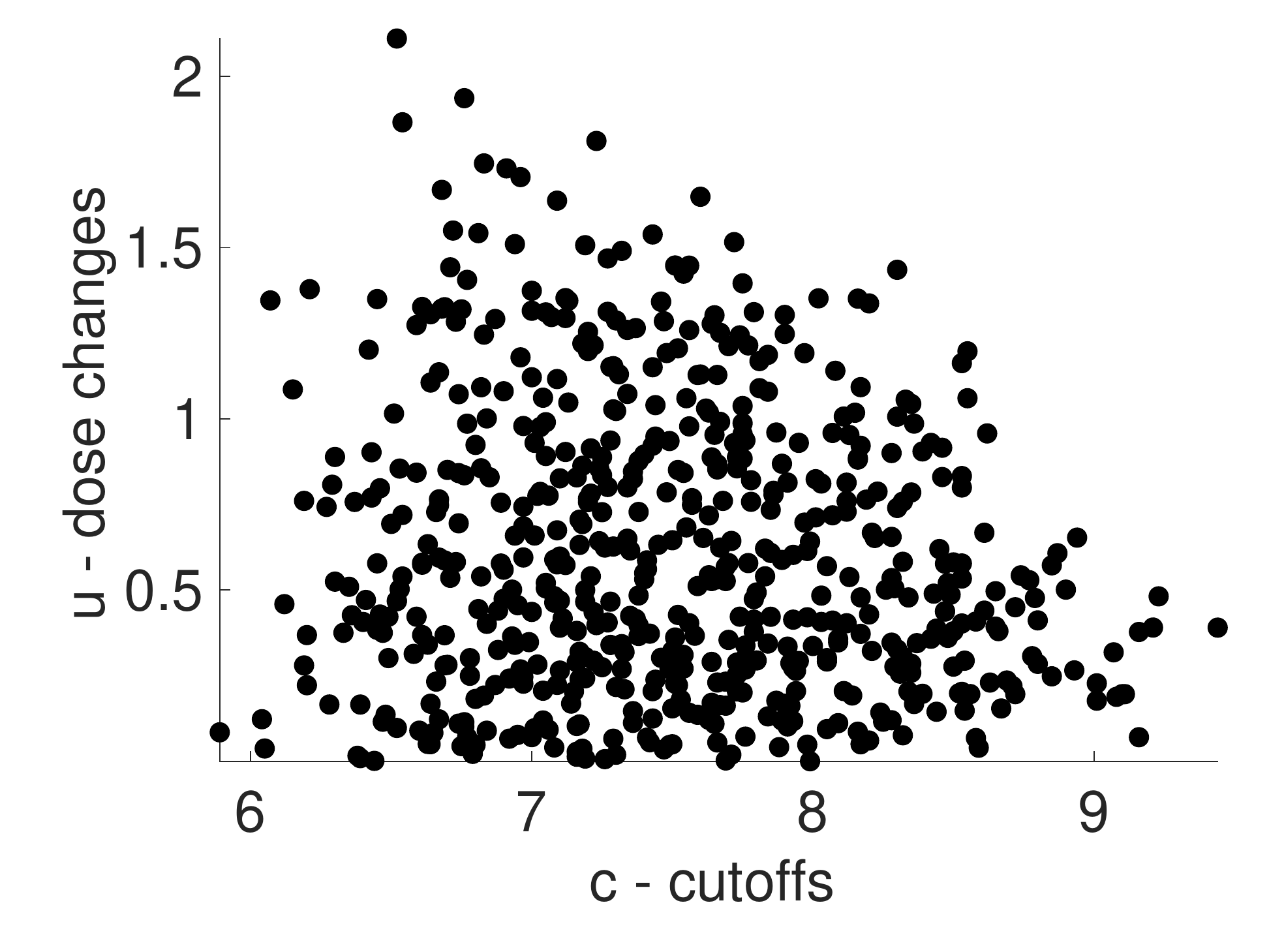}
  \end{minipage}%
\end{center}
\caption*{\footnotesize
Notes:
Scatter plot with cutoff values on the x-axis and  dose values on the y-axis 
for the $K= 588$ cutoffs in the Romanian data.
Panel (a) shows both doses before and after the cutoff, that is, $d_{j-1}$ in black and    $d_j$ in gray.
Panel (b) displays dose-change values on the y-axis, that is, $u_j = d_j - d_{j-1}$.
The peer-quality of school $j$ (treatment dose $d_j$) is measured by the average transition score of the students attending that school $j$.
The cutoff $c_j$ for admission into a school $j$ is equal to the minimum transition score among the students assigned to school $j$.
}
\end{figure}

Non-parametric identification of $\beta(\bc)$ is limited to the set $\m{C}$, which is the convex-hull of $\m{C}_{\infty}$.
The set $\m{C}_{\infty}$ is not entirely observed, and the researcher relies on the observation of $\m{C}_K$ (Figure \ref{figure_schedule}(a)).
For the sake of simplicity, I restrict $\beta$ to be a function of dose changes ($u=d'-d$) instead of doses before and after ($d$ and $d'$).
The restriction greatly simplifies the visualization and estimation of $\beta(c,d,d')$, because it implies that $\beta(c,d,d')=\phi(c)(d'-d)$, where $\phi$ is a continuously differentiable function.
Figure \ref{figure_schedule}(b) illustrates the variation of cutoff and dose-change values and  defines the limits on identification of policy counterfactuals.
For example, it is not possible to identify the effects of randomly assigning students with grades between 8 and 9 to a change in treatment dose of 2.
The support of such counterfactual distribution falls outside the observed variation of cutoff and dose-change values.
On the other hand, it is possible to identify the ATE  of randomly assigning students with grades between 6.5 and 8.5 to dose increases between 0 and 1.5.

The following policy question illustrates the ATE estimator proposed in this paper.
Suppose a new charter school is constructed in one of the towns in Romania.
The new charter school has more autonomy and better management than traditional public schools,
and admitted students experience an increase in school quality as if they were admitted to a school with better peers.
More specifically, the policy counterfactual is to give a 0.5 increase in peer-quality to a uniform distribution of scores between 6.5 and 8.5.
The ATE parameter is defined as 
\begin{gather}
\mu = \frac{1}{2} \bigintssss_{6.5}^{8.5} \phi(c) ~dc.
\label{def:mu_charter}
\end{gather}

I follow the estimation procedure suggested in Section \ref{sec_case2} and take into account the restriction $\beta(c,d,d')=\phi(c)(d'-d)$.
As in the binary treatment case, 
the restriction lowers the polynomial degree requirement on the second-step estimation 
to $\rho_2=1$. 
See Figure \ref{figure_rates} and the discussion that follows it.
The grid for $h_2$ has 32 equally-spaced points between 0.1 and 3.6, respectively, the smallest bandwidth for which the estimator is computable, 
and the maximum distance between two different cutoffs.
The MSE-optimal bandwidth choice is $h_2^* = 1.837$.
The new charter school has a bias-corrected ATE of $0.2964$ 
with standard error of $0.1296$, and it is statistically significant at 5\%.
Figure \ref{figure_beta}(a) plots $0.5 \ha{\phi}(c)$, that is, the effect of a $0.5$ increase in the treatment dose for various levels of $c$.
The graph reveals heterogeneous marginal effects of ability on returns to school quality.
Heterogeneity of treatment effects is \textit{a priori} unknown, and the ATE estimator proposed in this paper is consistent for $\mu$ regardless of the shape of $\phi(c)$.
This highlights the empirical relevance of Theorem \ref{theo_srd_kinf_est_int} and the importance of the second-step estimation.
In other words, the common strategy of normalizing all cutoffs to zero and estimating one discontinuity using the pooled data is not consistent for $\mu$
when $\phi(c)$ has such heterogeneity. 
\begin{figure}[H]
\caption{Treatment Effect Function}
\label{figure_beta}
\begin{center}
  \begin{minipage}{3in}
    \centering
    (a)

    \includegraphics[width=2.5in]{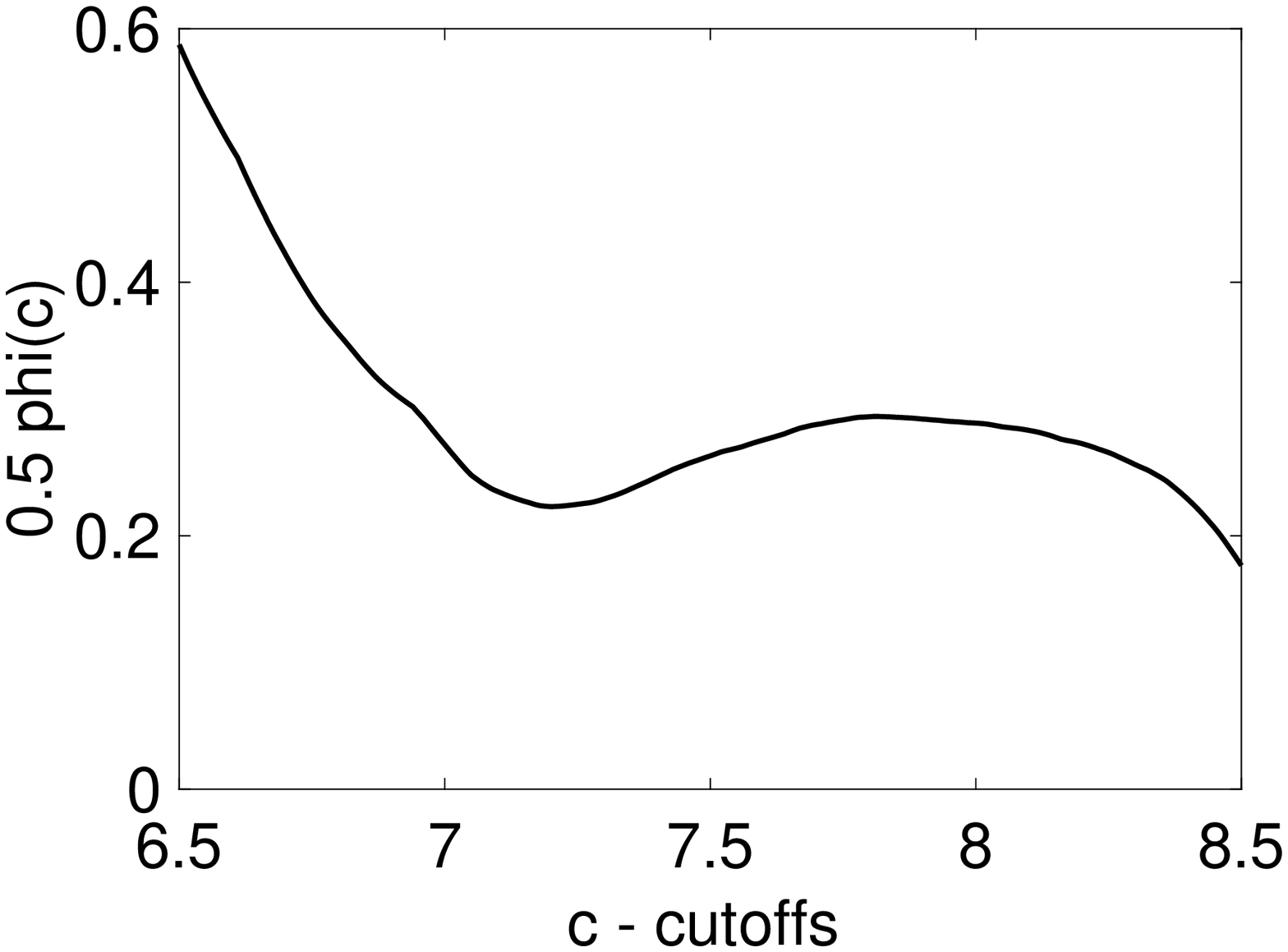}
  \end{minipage}%
  \begin{minipage}{3in}
    \centering
    (b)

    \includegraphics[width=2.5in]{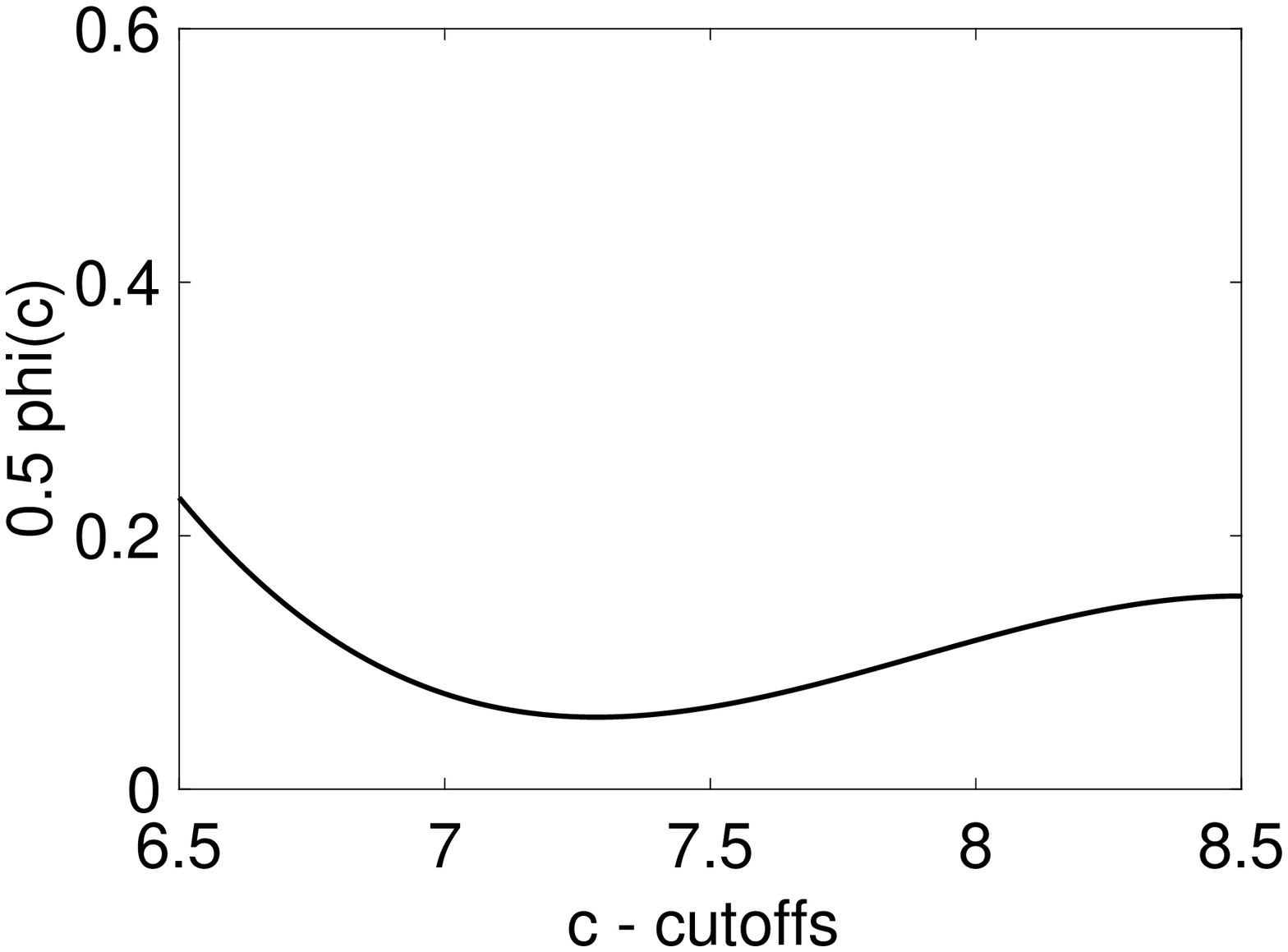}
  \end{minipage}%
\end{center}
\caption*{\footnotesize
Notes: Estimated average treatment effect function for a $0.5$ increase in school quality for students with score equal to $c$.
The figure plots $\ha{\beta}(c,d,d+0.5) = 0.5\ha{\phi}(c)$ for $c\in[6.5,8.5]$.
Panel (a) shows the ITT effect of a 0.5 increase in average peer performance for various levels of transition score.
The $\phi$ function is estimated non-parametrically with bias correction following Section \ref{sec_case2} (sharp case).
Panel (b) displays the effect on ever-compliers of the same uniform change in treatment dose.
The $\phi$ function is estimated parametrically with bias correction following the iterated
 procedure of Section \ref{sec_fuzzy_infer} in the supplemental appendix (fuzzy case).
}
\end{figure}

Estimation of treatment effects on ever-compliers requires a parametric functional form on $\beta_{ec}$ (Theorem \ref{theo_frd_id}).
I assume $\beta_{ec}(c,d,d') = \theta_1  (d'-d) + \theta_2 c (d'-d) +\theta_3 c^2 (d'-d)+  \theta_4 c^3 (d'-d)$
and carry out the iterative estimation procedure described in the supplemental appendix's 
Section \ref{sec_fuzzy_infer}.
The algorithm achieves convergence of $\theta$s within 30 iterations.
The iterated bias-corrected ATE on ever-compliers equals $0.107$ with standard error of $0.0109$.
The precision is substantially greater than the non-parametric case.
Figure \ref{figure_beta}(b) displays the treatment effect function on ever-compliers for a dose change of $0.5$.
Compared to ITT effects in Figure \ref{figure_beta}(a), the return of better schooling on ever-compliers is also positive,
but much less heterogeneous across ability levels.

\section{Conclusion}\label{sec_con}
\indent

Difficulty in gathering experimental data in many fields within the social sciences
makes quasi-experimental techniques such as RDD
extremely important to evaluate policies and social programs. 
RDD has been used in a wide range of applications in economics since the late 1990s.
More recently,
there has been an increasing number of applications with one forcing variable and multiple cutoffs,
assigning individuals to heterogeneous treatments.
The demand for multi-cutoff RDD methods is constantly growing, as richer data sets become ever more available.

This paper states conditions
under which multiple RDD effects are combined to infer ATE
over the entire range of cutoff values.
The proposed estimator is consistent and asymptotically normal for ATEs over
the entire support of variation in cutoffs and treatment doses.
Asymptotic results are derived under a large number of observations
and cutoffs in the sharp case of non-parametric treatment effect functions.
Sufficient conditions on the rate of growth of the number of cutoffs, relative to the number of observations, are given.
These rate conditions determine the feasible choice set of tuning parameters. 
This paper also shows that non-parametric identification in fuzzy RDD with multiple cutoffs is impossible unless 
the treatment effect function is finite-dimensional, or there is large variation of cutoff-dose values.  
A parametric specification provides an MSE-optimal ATE estimator for the fuzzy case that is consistent and asymptotically normal.

The relevance of the ATE estimators proposed in this paper is illustrated with the data of \cite{pop2013going} on high school assignment in Romania.
Of interest is the effect of high school quality on academic performance of students.
I find strong evidence of non-linearities in the returns to better schooling, as a function of students' ability level.
Monte Carlo simulations demonstrate that such non-linearities severely bias a naive
average of local effects that does not use the correction weighting
scheme proposed in this paper.
Applying the fuzzy RDD methods to the Romanian data
reveals causal effects on ever-compliers that are smaller and less heterogeneous than ITT effects.

The proposed estimator converges at the minimax optimal rate of root-$n$, as long as first-step bandwidths converge to zero at $1/K$ rate.
It would be interesting to learn about efficiency properties of the ATE estimator.
Theoretical tools commonly employed to derive efficiency lower-bounds may not be immediately applicable to the setting of this paper.
These tools are designed for regular estimators, and for data drawn from a population where the parameter of interest is identified.
In contrast, the fixed-cutoff RDD design relies on an ``identification at infinity argument'', 
and I wonder about the sufficient conditions that would obtain regularity of the ATE estimator.
A possibility for future work is to the generalize the uniform convergence tools from this paper to arrive at such conditions.


\section{Acknowledgements}

I am indebted to Han Hong, Caroline Hoxby, and Guido Imbens for invaluable advice.
The paper also benefited from feedback received from seminar participants at Stanford,
Boston University, Cambridge, Iowa,
Notre Dame, CORE-UcLouvain, UCSD, UCDavis,
FGV-EESP, FGV-EPGE, Insper, PUC-Rio, Toulouse, UIUC, and at various conferences.
I thank
Tim Bresnahan,
Arun Chandrasekhar,
Michael Dinerstein,
Ivan Fernandez-Val,
Ivan Korolev,
Michael Leung,
Huiyu Li,
Jessie Li,
Petra Moser,
Stephen Terry,
Xiaowei Yu,
and anonymous referees for suggestions and comments.
I gratefully acknowledge the
financial support received from the B.F. Haley and E.S. Shaw Fellowship at SIEPR-Stanford,
CORE-UcLouvain,
ISLA-Notre Dame,
and 
while visiting the Kenneth C. Griffin Department of Economics at the University of Chicago.

\begin{singlespace}
\bibliographystyle{econ}
\begingroup
    \setlength{\bibsep}{1.5pt}
    \bibliography{biblio}
\endgroup
\end{singlespace}

\appendix

\numberwithin{equation}{section}
\numberwithin{lemma}{section}

\numberwithin{theorem}{section}
\numberwithin{table}{section}

\begin{singlespace}

\section{Appendix}\label{sec_appen}

Throughout the appendices, $M$ is used as a generic finite and positive constant in the proofs.
For a $p \times q$ matrix A, the norm of A is induced by the Euclidean norm $\| \cdot \|$, i.e.
$\| A \|= \max_{x\in\mathbb{R}^q,x \neq 0}  \| Ax\| / \| x \| $.
The determinant of matrix $A$ is denoted $\det(A)$.
References to the supplemental appendix include B in the numbering; for example,
Lemma B.1, or Table B.2.

\subsection{Proof of Theorem \ref{theo_srd_est_avg} }

 Lemma \ref{lemma_porter} derives asymptotic normality of the bias-corrected jump-discontinuity estimator at one cutoff
based on local polynomial regressions of a vector $\bY_i$ on a scalar forcing variable $X_i$.
The proof of Theorem \ref{theo_srd_est_avg} is a straightforward generalization of  Lemma \ref{lemma_porter} in the particular case of a scalar $Y_i$.
As the sample size increases and the number of cutoff remains fixed, the jump-discontinuity estimators are independent across cutoffs.
First apply  Lemma \ref{lemma_porter} to each cutoff individually, and then aggregate over cutoffs.

$\square$

\subsection{Proof of Lemma \ref{lemma_srd_id_kinf}}
Define $\overline{\m{C}}=[\underline{\m{X}},\overline{\m{X}}] \times
[\underline{\m{D}},\overline{\m{D}}] \times
[\underline{\m{D}},\overline{\m{D}}]$.
Consider the partition of $\overline{\m{C}}$
made of the set of non-intersecting cubicles $T_n=\{ C_1, \ldots, C_M \}$ with $M=n^3$, $n=\{1,2,\ldots \}$.
Each $C_j$ is a half-open cubicle of the form
$[x_{l-1},x_l) \times [y_{m-1},y_m) \times [z_{o-1},z_o)$
with sides of lengths equal to $( \overline{\m{X}} - \underline{\m{X}})/n$,
$(\overline{\m{D}} - \underline{\m{D}})/n$, and $(\overline{\m{D}} - \underline{\m{D}})/n$.
Define the sub-collection
$U_n= \{C \in T_n:~ C \subset \m{C} \} = \{A_1,\ldots,A_{Q}\}$.
Since $\m{C}_{\infty}$ is dense in $\m{C}$,
for every $A_j \in U_n$, find a point $\bc_j \in \m{C}_{\infty} \cap A_j$ for which $\beta(\bc_j)$
is known.
The sum $
\mu_n = \sum_{j=1}^{Q} \omega(\bc_j) \beta(\bc_j)
\left( \overline{\m{X}} - \underline{\m{X}} \right)
\left( \overline{\m{D}} - \underline{\m{D}} \right)^2 / n^3
$ 
converges to $\mu^{c}$ as $n\to\infty$ because $\omega(\bc) \beta(\bc)$ is Riemann integrable
on $\m{C}$.

$\square$

\subsection{Proof of Theorem \ref{theo_srd_kinf_est_int} }
\label{sec:proof:theo_srd_kinf_est_int}

The proof combines arguments from the proof of  Lemma \ref{lemma_porter}
with lemmas on the uniform convergence of empirical processes from  Sections \ref{sec:supp:app:uniform}
and \ref{sec_int_app}.
Define $\mu^*$, $\ti{\mu}$, and $\mu_n$ as follows:
\begin{align}
\mu^* 
= &\sum_{j=1}^K 
	\Delta_j \Big\{
		e_1'\mme[ G_n^{j +} ] 
		\frac{1}{nh_{1j}}
		\sum_{i=1}^n
			k\left( \frac{X_i - c_{j}}{h_{1j}} \right)
			v_i^{j +} Y_i
			\ti{H}_{i}^{j}
\notag
\\*
& \hspace{1.2cm}
		-e_1'\mme[ G_n^{j -} ] 
		\frac{1}{nh_{1j}}
		\sum_{i=1}^n
			k\left( \frac{X_i - c_{j}}{h_{1j}} \right)
			v_i^{j -} Y_i
			\ti{H}_{i}^{j}
	\Big\}
\\*
= & \sum_{j=1}^K 
	\Delta_j
	\frac{1}{nh_{1j}}
	\sum_{i=1}^n
		k\left( \frac{X_i - c_{j}}{h_{1j}} \right)
		Y_i
		e_1'
		\left(
			v_i^{j +} \mme[ G_n^{j +} ]  
			-
			v_i^{j -} \mme[ G_n^{j -} ]  
		\right)
		\ti{H}_{i}^{j}
\\
\ti{\mu} = & \sum_{j=1}^K 
	\Delta_j \left\{
		e_1'  G_n^{j +}
		\mme\left[  
			\frac{1}{nh_{1j}}
			\sum_{i=1}^n
				k\left( \frac{X_i - c_{j}}{h_{1j}} \right)
				v_i^{j +} Y_i^{j+} 
				\ti{H}_{i}^{j}
		\right]
	\right.
\notag
\\*
& \hspace{1.2cm} \left.
		-e_1'  G_n^{j -}  
		\mme\left[
			\frac{1}{nh_{1j}}
			\sum_{i=1}^n
				k\left( \frac{X_i - c_{j}}{h_{1j}} \right)
				v_i^{j -} Y_i^{j-}  
				\ti{H}_{i}^{j}
		\right]
	\right\}
\notag
\\
\mu_n  &=
\sum_{j=1}^K \Delta_j B_j.
\end{align}

Write

\begin{align}
\frac{\ha{\mu} - \m{B}_{1n}^c - \m{B}_{2n}^c - \mu}{ (\m{V}_n^c)^{1/2}} = &
\frac{\mu^* - \mme[\mu^* |\m{X}_n] }{(\m{V}_n^c)^{1/2}}
\label{eq:kinf:clt}
\\
+ &
{\frac{ \ti{\mu} - \m{B}_{1n}^c  }{(\m{V}_n^c)^{1/2}} }
\label{eq:kinf:bias}
\\
+ &
\frac{ \mu_n - \m{B}_{2n}^c - \mu  }{(\m{V}_n^c)^{1/2}}
\label{eq:kinf:int}
\\
+ &
\frac{ \ha{\mu} - \mme[\ha{\mu} |\m{X}_n] - \left( \mu^* - \mme[\mu^* |\m{X}_n] \right) }{(\m{V}_n^c)^{1/2}}
\label{eq:kinf:op1:a}
\\
+ &
\frac{ \mme[\ha{\mu} - \mu_n |\m{X}_n] - \ti{\mu}  }{(\m{V}_n^c)^{1/2}}.
\label{eq:kinf:op1:b}
\end{align}

The proof in this appendix 
applies a central limit theorem (CLT) to show that
Part \eqref{eq:kinf:clt} converges in distribution to a standard normal;
it demonstrates that $\m{B}_{1n}$ approximates the first-step bias, that is, that part \eqref{eq:kinf:bias} converges in probability to zero;
and it shows that $\m{B}_{2n}$ approximates the second-step bias (integration error), that is, that part \eqref{eq:kinf:int} converges to zero.
 Lemma \ref{lemma_rest_theo} shows that 
parts \eqref{eq:kinf:op1:a} and \eqref{eq:kinf:op1:b} converge in probability to zero.

\begin{nopgbreak}
\begin{center}
\textbf{\underline{Part \eqref{eq:kinf:clt} }} 
\end{center}

First, find the rate that $(\m{V}_n^c)^{-1/2}$ grows.
Define $\phi_n$ and rewrite $\m{V}_n^c$ as follows:
\end{nopgbreak} 
\begin{align}
\phi_n(X_i) = &
		\sum_{j=1}^K 
			\frac{\Delta_j}{nh_{1j}}
			k\left( \frac{X_i - c_{j}}{h_{1j}} \right)
			e_1'
			\left(
				v_i^{j +} \mme[ G_n^{j +} ]  
				-
				v_i^{j -} \mme[ G_n^{j -} ]  
			\right)
			\ti{H}_{i}^{j}				
\\*
\m{V}_n^c = &\sum_{i=1}^n \mme \left[
	\eps_i^2 \phi_n(X_i)^2
\right].
\end{align}
Choose alternative bandwidths  $h_{1j}^*$, $j=1,\ldots,K$, such that 
(i) there exists $\delta>0$ (independent of $n$) such that  $\delta < h_{1j}^* / h_{1j} \leq 1~ \forall j$ ;
and
(ii) $[c_j - h_{1j}^*, c_j + h_{1j}^*] \cap  [c_{j'} - h_{1j'}^*, c_{j'} + h_{1j'}^*] = \emptyset$ for any $j \neq j'$.
\begin{align}
\m{V}_n^c & = n \mme \left[\zeta^2(X_i) \phi_n^2(X_i) \right] 
 \geq n \sum_{j=1}^K \bigintss_{c_j - h_{1j}^* }^{c_j + h_{1j}^*}
	\zeta^2(x) \phi_n^2(x) f(x) ~ dx
\\
& = n \sum_{j=1}^K \bigintss_{c_j - h_{1j}^* }^{c_j + h_{1j}^*}
	\zeta^2(x) 
	\Bigg(
		\frac{\Delta_j}{nh_{1j}} k\left(\frac{x-c_j}{h_{1j}}\right)
		e_1' \left( 
			\mmi\{x \geq 0\} \mme[ G_n^{j +} ]   
			-
			\mmi\{x < 0\} \mme[ G_n^{j -} ]  
		\right) 
\nonumber
\\*		
&		\hspace{4.5cm} H\left(\frac{x-c_j}{h_{1j}}\right)
	\Bigg)^2
	f(x) ~ dx
\\
& = \frac{1}{K n } \frac{1}{K} \sum_{j=1}^K \frac{K^2 \Delta_j^2}{h_{1j}} \bigintss_{ - h_{1j}^*/ h_{1j} }^{  h_{1j}^*/ h_{1j} }
	\zeta^2(c_j + u h_{1j}) 
	\Bigg(
		k\left( u \right)
		e_1' \left( 
			\mmi\{u \geq 0\} \mme[ G_n^{j +} ]   
			-
			\mmi\{u < 0\} \mme[ G_n^{j -} ]  
		\right) 
\nonumber 
\\* 
&		\hspace{7.8cm} H\left(  u  \right)
	\Bigg)^2
	f(c_j + u h_{1j}) ~ du	
~ \geq  ~ \frac{M}{K n \overline{h}_1}  
\label{eq:kinf:clt:s2:rate}
\end{align}
where the first inequality follows from the integrand being positive and $\cup_{j=1}^K [c_j - h_{1j}^*, c_j + h_{1j}^*] \subseteq \cup_{j=1}^K [c_j - h_{1j}, c_j + h_{1j}]$;
the third equality uses a change of variables $u = (x-c_j)/h_{1j}$;
and the last inequality follows because
(a) $h_{1j}\leq \overline{h}_1$;
(b) $K^2 \Delta_j^2 $ is bounded away from zero uniformly over $j$ (Lemma \ref{lemma_rate_int});
and 
(c) each integral is bounded away from zero over $j$
because the integration limits, $\zeta^2(c_j + u h_{1j}) $, 
$\mme[ G_n^{j \pm} ]$,
and
$f(c_j + u h_{1j})$
are uniformly close to quantities that are positive definite uniformly over $j$ 
(see  Lemma \ref{lemma_app} and recall that $f$ and $\zeta$ are bounded away from zero
because of 
Assumptions \ref{assu_srd_est_fx} and \ref{assu_srd_kinf_est_int}).
The inequality in \eqref{eq:kinf:clt:s2:rate} implies that $(\m{V}_n^c)^{-1} = O(Kn\overline{h}_1)$ where 
$Kn\overline{h}_1 \to \infty$.

Second, write  part \eqref{eq:kinf:clt} as a weighted sum across $i$:
\begin{align}
\mu^* &=
\sum_{i=1}^n
Y_i 
\underset{\equiv \phi_{n}(X_i)}
{
\underbrace{
\sum_{j=1}^K 
\frac{\Delta_j}{nh_{1j}}
k\left( \frac{X_i - c_{j}}{h_{1j}} \right)
e_1'
\left(
v_i^{j +} \mme[ G_n^{j +} ]  
-
v_i^{j -} \mme[ G_n^{j -} ]  
\right)
\ti{H}_{i}^{j}
}
}
=\sum_{i=1}^n Y_i \phi_n(X_i)
\label{eq:kinf:clt:sumiid}
\end{align}
so that
\begin{align}
\frac{\mu^* - \mme[\mu^* |\m{X}_n] }{(\m{V}_n^c)^{1/2}} &=
\frac{
\sum_{i=1}^n
\left( Y_i - \mme[Y_i|X_i] \right)
\phi_n(X_i)
}
{
(\m{V}_n^c)^{1/2}
}
=
\frac{
\sum_{i=1}^n
\eps_i
\phi_n(X_i)
}
{
(\m{V}_n^c)^{1/2}
}.
\label{eq:kinf:clt:sumiid:center}
\end{align}
Equation \ref{eq:kinf:clt:sumiid:center} is a sum of iid random variables with zero mean,
where $\m{V}_n^c$ is the variance of the numerator.
The Lindeberg condition is verified next.
Take an arbitrary $\delta>0$. 
\begin{align}
& \sum_{i=1}^{n} \mme\left[
	(\m{V}_n^c)^{-1} \eps_i^2 \phi_n(X_i)^2 
	\mmi\left\{
				\left| (\m{V}_n^c)^{-1/2} \eps_i \phi_n(X_i) \right| > \delta
	\right\}
\right]
\\
\leq & \sum_{i=1}^{n} \mme\left[
	M K n \overline{h}_1  \phi_n(X_i)^2 
	\mmi\left\{
				 M' \left( K n \overline{h}_1 \right)^{1/2} \left| \phi_n(X_i) \right| > \delta
	\right\}
\right]
\\
\leq & \sum_{i=1}^{n} \mme\left[
	M K n \overline{h}_1 \left(   K n \underline{h}_1 \right)^{-2}
	\mmi\left\{
				M' \left( K n \overline{h}_1 \right)^{1/2} \left(K n \underline{h}_1 \right)^{-1}  > \delta
	\right\}
\right]
\\
\leq & 
	\left(   K \underline{h}_1 \right)^{-1}
	\mmi\left\{
				M'  \left(K n \underline{h}_1 \right)^{-1/2}  > \delta
	\right\} =o(1)
\end{align}
where the first inequality relies on the fact that $\eps_i$ is a.s. bounded (Assumption \ref{assu_srd_kinf_est_int}),
and that  $(\m{V}_n^c)^{-1} = O \left( Kn\overline{h}_1 \right)$ (Equation \ref{eq:kinf:clt:s2:rate}).
The second inequality uses that $\phi_n(x) = O\left(   K n \underline{h}_1 \right)^{-1}$ 
 uniformly over $x$. 
 In fact, $\phi_n(x)$ is a sum of $K$ components of which at most two are non-zero,
  $\Delta_j = O\left( K^{-1} \right)$ uniformly over $j$ (Lemma \ref{lemma_rate_int}),
  $k(\cdot)$ is bounded (Assumption \ref{assu_srd_est_kernel}),
  $\mme[ G_n^{j \pm} ]$ is uniformly close to $G^{j \pm}$ whose norm is bounded away from zero (Lemma \ref{lemma_app}).
  The last inequality relies on the rate condition $\overline{h}_1/\underline{h}_1 = O(1)$,
  and that the indicator becomes zero for large $n$.
The Lindeberg-Feller CLT says that Equation \ref{eq:kinf:clt:sumiid:center}, 
and thus part \eqref{eq:kinf:clt}, converges in distribution to a standard normal.

\begin{center}
\textbf{\underline{Part \eqref{eq:kinf:bias} }} 
\end{center}

First consider
\begin{align}
& \mme\left[
	\frac{1}{h_{1j}} 
	k\left( \frac{X_i - c_j}{h_{1j}} \right)
	v_{i}^{j+} \ti{H}_{i}^{j}  \mme\left[ Y_i^{j+} | X_i \right]
\right] 
\\
= & \mme\left[
	\frac{1}{h_{1j}} 
	k\left( \frac{X_i - c_j}{h_{1j}} \right)
	v_{i}^{j+} \ti{H}_{i}^{j}
	\frac{ \nabla^{(\rho_1 + 1 )}R(c_j,d_j)}{ (\rho_1 + 1 )! }
	\left( \frac{X_i-c_j}{h_{1j}} \right)^{\rho_1+1} h_{1j}^{\rho_1+1}  
\right]
\label{eq:kinf:bias:expansion:fo}
\\
 & \hspace{.5cm} + 
\mme\left[
	\frac{1}{h_{1j}} 
	k\left( \frac{X_i - c_j}{h_{1j}} \right)
	v_{i}^{j+} \ti{H}_{i}^{j}
	\frac{ \nabla^{(\rho_1 + 2 )}R(c_j^*,d_j)}{ (\rho_1 + 2 )! }
	\left( \frac{X_i-c_j}{h_{1j}} \right)^{\rho_1+2} h_{1j}^{\rho_1+2}  
\right]
\label{eq:kinf:bias:expansion:so}
\\
= & 
h_{1j}^{\rho_1+1}  
\frac{ \nabla^{(\rho_1 + 1 )}R(c_j,d_j)}{ (\rho_1 + 1 )! }
f(c_j) \gamma^*
+ O\left( \overline{h}_{1}^{\rho_1+2} \right)
\label{eq:kinf:bias:expansion}
\end{align}
where $\mme\left[ Y_i^{j+} | X_i \right]$ is the difference between $\mme[Y_i|X_i]$ 
and its $\rho_1$-th order Taylor expansion around $X_i=c_j$ (see  Equations \ref{eq:lemma_porter:def:yij:plus} and \ref{eq:lemma_porter:def:yij:minus}).
The expectations in Equations \ref{eq:kinf:bias:expansion:fo} and \ref{eq:kinf:bias:expansion:so},
without the $h_{1j}^{\rho_1+1} $ and $h_{1j}^{\rho_1+2}$ terms, are bounded over $j$
because the kernel, derivatives, and polynomials are bounded functions of $u = (x-c_j)h_{1j}^{-1}$
(Assumptions \ref{assu_srd_est_kernel} and \ref{assu_srd_kinf_est_int}).  
The remainder term $O\left( \overline{h}_{1}^{\rho_1+2} \right)$ is uniform over $j$.

Next, 
\begin{align} 
\frac{\ti{\mu} - \m{B}_{1n} }{(\m{V}_n^c)^{1/2}} = & (\m{V}_n^c)^{-1/2} \sum_{j=1}^K 
	\Delta_j 
		e_1'  G_n^{j +}
		\mme\left[  
			\frac{1}{nh_{1j}}
			\sum_{i=1}^n
				k\left( \frac{X_i - c_{j}}{h_{1j}} \right)
				v_i^{j +} \mme\left[ Y_i^{j+} | X_i \right] 
				\ti{H}_{i}^{j}
		\right]
\notag\\
& - (\m{V}_n^c)^{-1/2} \m{B}_{1n}^+
\label{eq:kinf:bias:plus}
\\*
& 	- (\m{V}_n^c)^{-1/2} \sum_{j=1}^K 
	\Delta_j 	
		e_1'  G_n^{j -}  
		\mme\left[
			\frac{1}{nh_{1j}}
			\sum_{i=1}^n
				k\left( \frac{X_i - c_{j}}{h_{1j}} \right)
				v_i^{j -} \mme\left[ Y_i^{j-} | X_i \right]  
				\ti{H}_{i}^{j}
		\right]
\notag\\
& + (\m{V}_n^c)^{-1/2} \m{B}_{1n}^-
\label{eq:kinf:bias:minus}
\end{align}
 where 
\begin{align}
\m{B}_{1n} = & \m{B}_{1n}^+ -  \m{B}_{1n}^-
\\
\m{B}_{1n}^+ = & {\left( (\rho_1 + 1 )! \right)^{-1} } \sum_{j=1}^K 
		h_{1j}^{\rho_1+1} \Delta_j f(c_j) 
		{ \nabla^{ \rho_1 + 1 }_x R(c_j,d_j)}
		e_1'  G_n^{j+}  \gamma^*
\\
\m{B}_{1n}^- = & {\left( (\rho_1 + 1 )! \right)^{-1} } \sum_{j=1}^K 
		h_{1j}^{\rho_1+1} \Delta_j f(c_j)  
		{ \nabla^{ \rho_1 + 1 }_x R(c_j,d_{j-1})}
		e_1'   G_n^{j-}   \gamma^*.
\end{align}

Consider part \eqref{eq:kinf:bias:plus}. Part \eqref{eq:kinf:bias:minus} follows a symmetric argument.
\begin{align}
\eqref{eq:kinf:bias:plus} = & 
(\m{V}_n^c)^{-1/2} \sum_{j=1}^K 
	\Delta_j 
		e_1'   G_n^{j +} 
		\mme\left[  
			\frac{1}{nh_{1j}}
			\sum_{i=1}^n
				k\left( \frac{X_i - c_{j}}{h_{1j}} \right)
				v_i^{j +} \mme\left[ Y_i^{j+} | X_i \right] 
				\ti{H}_{i}^{j}
		\right] - (\m{V}_n^c)^{-1/2} \m{B}_{1n}^+
\\
= & 
(\m{V}_n^c)^{-1/2} \left[ 
	\sum_{j=1}^K 
		\Delta_j 
			e_1'  G_n^{j +}
				h_{1j}^{\rho_1+1}  
				\frac{ \nabla^{(\rho_1 + 1 )}R(c_j,d_j)}{ (\rho_1 + 1 )! }
				f(c_j) \gamma^*
	- \m{B}_{1n}^+
\right]
\\
&+
(\m{V}_n^c)^{-1/2} \left[ 
	\sum_{j=1}^K 
		\Delta_j 
			e_1'  G_n^{j +}
			O\left( \overline{h}_{1}^{\rho_1+2} \right)		
\right]
\\
= & 
0
\\
&+
O\left( \left( K n \overline{h}_1 \right)^{1/2} \right)
K O(K^{-1}) O_P(1) O\left( \overline{h}_{1}^{\rho_1+2} \right)	
=o_P(1)
\end{align}
where the second equality uses the expansion in Equation \ref{eq:kinf:bias:expansion}.
The third equality  uses the definition of $\m{B}_{1n}^+$, that $\Delta_j = O(K^{-1})$ uniformly over $j$,
that $G_n^{j+}=O_P(1)$.
These terms are $o_P(1)$
because of the rate condition
$\left( K n \overline{h}_1 \right)^{1/2}\overline{h}_{1}^{\rho_1+1} =O(1) $.

\begin{center}
\textbf{\underline{Part \eqref{eq:kinf:int}}}
\end{center}

\begin{align}
\frac{ \mu_n - \m{B}_{2n} - \mu  }{(\m{V}_n^c)^{1/2}} & = 
O\left( \left( K n \overline{h}_1 \right)^{1/2} \right) \left( 
	\sum_{j=1}^K  \Delta_j  B_j
	- \m{B}_{2n}
	 -\int_{\m{C}} \omega(\bc) \beta(\bc)  ~ d(\bc) 
\right)
\\ 
&= 
O\left( \left( K n \overline{h}_1 \right)^{1/2} \right)
O\left( h_2^{\rho_2+2} \right) 
=O(1) O\left( h_2 \right)
=o(1)
\end{align}
where the first equality uses the rate on $(\m{V}_n^c)^{-1/2}$ (Equation \ref{eq:kinf:clt:s2:rate}).
The second equality applies  Lemma \ref{lemma_rate_int} and relies
on Assumption \ref{assu_srd_kinf_est_cut} (asymptotic behavior of $\{ \bc_j \}_j$)
and 
Assumption \ref{assu_srd_kinf_est_int} (smoothness of $\beta(\bc)$).
The third equality uses the rate condition $\left( K n \overline{h}_1 \right)^{1/2} h_2^{\rho_2+1} =O(1)$.
 Lemma \ref{lemma_rate_int} also shows that 
$\m{B}_{2n}=O\left( h_2^{\rho_2+1} \right) $,
which yields 
$(\m{V}_n^c)^{-1/2} \m{B}_{2n} =  O\left( \left( K n \overline{h}_1 \right)^{1/2} h_2^{\rho_2+1} \right)
= O(1)$.

 Lemma \ref{lemma_rest_theo} shows that 
parts \eqref{eq:kinf:op1:a} and \eqref{eq:kinf:op1:b} converge in probability to zero, which concludes the proof.

$\square$

\subsection{Proof of Theorem \ref{theo_srd_kinf_minimax}}
\begin{center}
\textbf{\underline{Part \eqref{theo_srd_kinf_minimax_a}}}
\end{center}

First, consider the ideal setting where estimators $\mu^*$ are functions of data observed from $\{ Y_i(d) \}_{d \in \m{D}}$ and $X_i$.
For a choice of loss function $L(\mu,\mu')$,
the minimax risk of estimating the parameter $\mu^c(P)$ is defined
as  $\inf_{\mu^*} \sup_{P \in \m{P}} \mme_P  \left[ L(\mu^*,\mu^c(P) ) \right]$.
Here, the 0-1 loss function is used, that is, $L_n(\mu,\mu') = \mmi\{ n^{r} |\mu - \mu' |>\epsilon \}$,
for a positive rate $r$ and $\epsilon$. In this case, 
$\mme_P  \left[ L_n(\mu^*, \mu^c(P) ) \right] = \mmp_P  \left[ n^{r} |\mu^* - \mu^c(P) | > \epsilon \right]$.
The minimax risk is the supremum probability over $\m{P}$ of an estimator
being farther than $\epsilon n^{-r}$ from the truth minimized over all possible estimators $\mu^*$. 
The rate $r$ is an upper bound on the rate of convergence if for small $\epsilon>0$
there exists a lower bound $L \in (0,1)$ such that 
$\inf_{\mu^*} \sup_{P \in \m{P}} \mmp_P  \left[ n^{r} |\mu^* - \mu^c(P) | > \epsilon \right] \geq L$ 
for large $n$.
The rate $r$ is the minimax optimal rate if it is an upper bound and achievable; 
that is, if there exists an estimator $\ha\mu$ that converges at rate $r$ uniformly.
The estimator $\ha\mu$ converges at rate $r$ uniformly if, for any small $\delta>0$, there exists large $\epsilon \in (0,\infty)$ such that
$
\sup_{P \in \m{P}} \mmp_{P}  \left[ n^r | \ha{\mu} - \mu^c(P) | > \epsilon \right] < \delta  ~~ \text{for large n}.
$
See discussion in Chapter 2 of \cite{tsybakov2009}.

One common approach to compute lower bounds for the minimax risk is to use Le Cam's  method.
For $\epsilon >0$, choose two models $P,Q \in \m{P}$ such that $|\mu^c(P) - \mu^c(Q)| > \epsilon n^{-r}$.
Le Cam's method leads to the following inequality:
$\inf_{\mu^*} \sup_{P \in \m{P}} \mmp_P  \left[ n^{r} |\mu^* - \mu^c(P)| > \epsilon/2 \right] \geq e^{-n KL(P,Q)}/4$,  
where $KL(P,Q)$ is the Kullback-Leibler divergence between $P$ and $Q$. 
See Equations (2.7), (2.9), and Theorem 2.2(iii) of \cite{tsybakov2009}.
This inequality is used to prove part \eqref{theo_srd_kinf_minimax_a} with $r=1/2$.

Consider the continuous counterfactual density $\omega^c(\bc)$. 
The researcher must be choose a counterfactual density
such that its marginal densities
$\int \omega^c(c,d,d')~d(d') $
and 
$\int \omega^c(c,d,d')~d(d) $
are different functions; otherwise, $\mu^c=0$.
Construct an infinitely differentiable bounded function $g(c,d) \geq 0$
such that $\int [g(c,d') - g(c,d)] ~ \omega(\bc) ~d\bc =1 $.

Construct two models $P,Q \in \m{P}$ as follows. 
Let $\eps_i \sim N(0,1)$ and $X_i \sim U[0,1]$ iid and independent of each other.
Pick $\xi>2\sqrt{\pi} \epsilon>0$.
For model $P$, define $Y_i(d)=\Phi\left( \xi n^{-1/2} g(X_i,d) + \eps_i \right)$,
where $\Phi$ is the standard normal cdf. For model $Q$, define   $Y_i(d)=\Phi\left( \eps_i \right)$.
The expectation of $Y_i(d)$ conditional on $X_i=c$, that is, $R(c,d)$, is an infinitely differentiable function.  
The variables have bounded support, and models $P$ and $Q$ satisfy all the conditions to be in $\m{P}$.
Under model $P$, 

$\beta(\bc;P)=\mme_P[Y_i(d') - Y_i(d)~|X_i=c]$

$ = \mme_P[\Phi\left( \xi n^{-1/2} g(X_i,d') + \eps_i \right) - \Phi\left( \xi n^{-1/2} g(X_i,d) + \eps_i \right) ~|X_i=c]$

$ = \mme_P[\phi\left( \eps_i^* \right)  \xi n^{-1/2}  \left( g(c,d')-g(c,d) \right)  ~|X_i=c]$

$ = \mme_P[\phi\left( \eps_i^* \right)]  ~ \xi n^{-1/2}  \left( g(c,d')-g(c,d) \right)  $

\noindent where $\phi$ is the standard normal pdf, and $\eps_i^*$ is in between $\eps_i + \xi n^{-1/2} g(c,d')$ and $\eps_i + \xi n^{-1/2} g(c,d)$.
As $n$ grows large, $\mme_P[\phi\left( \eps_i^* \right)] = \mme_P[\phi\left( \eps_i \right)] + o(1) = \frac{1}{2 \sqrt{\pi}} + o(1)$ where
the $o(1)$ term is uniform over $(c,d,d')$. Then, 

$\mu^c(P)=\frac{1}{2 \sqrt{\pi}} \xi n^{-1/2}  \int \left( g(c,d')-g(c,d) \right) ~ \omega(\bc) ~d\bc + o\left( n^{-1/2} \right)
=\frac{1}{2 \sqrt{\pi}} \xi n^{-1/2}  + o\left( n^{-1/2} \right). $

\noindent Under model $Q$, $\beta(\bc;Q)=0$. Therefore, 

$\mu^c(P) - \mu^c(Q) =\frac{1}{2 \sqrt{\pi}} \xi n^{-1/2}  + o\left( n^{-1/2} \right) > \epsilon n^{-1/2}$ 

\noindent for large $n$, because $\frac{1}{2 \sqrt{\pi}} \xi > \epsilon$.

\bigskip

Next, we use the following inequality to  show that $r=1/2$ is an upper bound on the rate of convergence, 

$\inf_{\mu^*} \sup_{P \in \m{P}} \mmp_P  \left[ n^{1/2} |\mu^* - \mu^c(P)| > \epsilon/2 \right] \geq e^{-n KL(P,Q)}/4$.  

\noindent Let $d^*$ be such that $g(c,d^*)>0$ for some $c$. 
For simple models like $P$ and $Q$, any function of the variables $\{ Y_i(d) \}_{d \in \m{D}}$ and $X_i$ can be rewritten as functions 
of $Y_i(d^*)$ and $X_i$ because $Y_i(d)$ is a deterministic function
of $Y_i(d^*)$ and $X_i$ for any d.
 It suffices to look at the distribution of  $Y_i(d^*)$ and $X_i$ instead of the distribution of $\{Y_i(d)\}_{d \in \m{D}}$ and $X_i$.
Consider the Kullback-Leibler divergence for the distributions $P$ and $Q$ of $(Y_i(d^*),X_i)$,

$KL(P,Q) = \bigintsss \log \left[ \frac{p(y,x)}{q(y,x)} \right] p(y,x) ~dydx$,

\noindent where $p(y,x)$ and $q(y,x)$ are the pdfs of $(Y_i(d^*),X_i)$ under $P$ and $Q$ respectively.
Define $\ti{Y}_i = \xi n^{-1/2} g(X_i,d^*) + \eps_i $ under $P$, and 
$\ti{Y}_i =  \eps_i $ under $Q$. 
It follows that $(Y_i(d^*),X_i) = ( \Phi(\ti{Y}_i) ,X_i)$ under both $P$ and $Q$.
The Kullback-Leibler divergence is invariant to such a transformation of variables.

$KL(P,Q) = \bigintsss \log \left[ \frac{\ti{p}(y,x)}{\ti{q}(y,x)} \right] \ti{p}(y,x) ~dydx$,

\noindent where $\ti{p}(y,x) = \phi\left( y-\xi n^{-1/2} g(x,d^*) \right)$
 and $\ti{q}(y,x)=\phi(y)$ are the pdfs of $(\ti{Y}_i,X_i)$ under $P$ and $Q$ respectively.

$KL(P,Q)=\bigintsss \log \left[ \frac{ \exp\left\{-(1/2) \left( y-\xi n^{-1/2} g(x,d^*) \right)^2 \right\}}{\exp\left\{-(1/2)  y^2 \right\} }\right] \ti{p}(y,x) ~dydx$

$=\bigintsss \log \left[  \exp\left\{  y \xi n^{-1/2} g(x,d^*) - (1/2)\xi^2 n^{-1} g(x,d^*)^2  \right\}  \right] \ti{p}(y,x) ~dydx$

$=\bigintsss  \left[ y \xi n^{-1/2} g(x,d^*) - (1/2)\xi^2 n^{-1} g(x,d^*)^2  \right] \ti{p}(y,x) ~dydx$

$=\bigintsss  (1/2) \xi^2 n^{-1} g(x,d^*)^2  ~dx$

$=(1/2) \xi^2 n^{-1} \bigintsss   g(x,d^*)^2   ~dx>0$

\noindent Pick $\eta >1$ such that $(1/2) \xi^2 \bigintsss   g(x,d^*)^2   ~dx < \log(\eta) $. Then,
$ e^{-n KL(P,Q)} / 4 > 1/(4\eta) >0$, and

$\inf_{\mu^*} \sup_{P \in \m{P}} \mmp_P  \left[ n^{1/2} |\mu^* - \mu^c(P)| > \epsilon/2 \right] \geq \frac{1}{4\eta}$.  

\bigskip

This is a minimax lower bound for estimators $\mu^*$ that are functions of
an ideal sample of $\{ Y_i(d) \}_{d \in \m{D}}$ and $X_i$.
In practice, only part of these variables are observed according to the schedule of cutoff-doses
$\{ \bc_j \}_{j=1}^K$. The set of all estimators $\ti{\mu}$
that are functions of the observed variables $(Y_i,X_i)$ is a subset of the set of all estimators
$\mu^*$. Therefore, the lower bound above is also a minimax lower bound for all estimators 
$\ti{\mu}$:

$\inf_{\ti{\mu}} \sup_{P \in \m{P}} \mmp_P  \left[ n^{1/2} |\ti{\mu} - \mu^c(P)| > \epsilon/2 \right] \geq \frac{1}{4\eta}$.  

\bigskip

\begin{center}
\textbf{\underline{Part \eqref{theo_srd_kinf_minimax_b}}}
\end{center}

Let $\ha{\mu}$ denote $\ha{\mu}^c$ and $\mu = \mu^c(P)$ for notational ease.
The goal is to show that, for any small $\delta>0$, there exists large $\epsilon \in (0,\infty)$ such that
$
\sup_{P \in \m{P}} \mmp_{P}  \left[ n^{1/2} | \ha{\mu} - \mu | > \epsilon \right] < \delta
$ for large $n$.
The choice of $\overline{h}_1$ plus the discussion preceding Equation \ref{eq:kinf:clt:s2:rate} lead
 to $(\m{V}^c_n)^{-1/2} \geq M n^{1/2}$ for large $n$.
Thus, 
$\mmp_{P}  \left[ n^{1/2} | \ha{\mu} - \mu | > \epsilon/M \right] \leq \mmp_{P}  \left[ (\m{V}^c_n)^{-1/2} | \ha{\mu} - \mu | > \epsilon \right]$
uniformly over $\m{P}$ for large $n$.
Theorem \ref{theo_srd_kinf_est_int} breaks $(\m{V}^c_n)^{-1/2} | \ha{\mu} - \mu | $
into four components: 
the CLT component $N_n$, that converges in distribution to a standard normal (part \eqref{eq:kinf:clt}); 
the first-step bias component $B_n$, that converges in probability to zero (part \eqref{eq:kinf:bias});
the integration error component $I_n$, that converges in probability to zero (part \eqref{eq:kinf:int});
and the remainder terms $R_n$, that converge in probability to zero (parts \eqref{eq:kinf:op1:a} and \eqref{eq:kinf:op1:b}). 
It is true that

$\mmp_P\left(  (\m{V}^c_n)^{-1/2} | \ha{\mu} - \mu | > \epsilon \right) \leq
\mmp_P\left( | N_n | > \epsilon/4 \right)
+
\mmp_P\left( | B_n | > \epsilon/4 \right)
+
\mmp_P\left( | I_n | > \epsilon/4 \right)
+
\mmp_P\left( | R_n | > \epsilon/4 \right)
$.

\noindent Hence, for  each of the four components, it suffices to show that for a choice of $\delta>0$ small, there exist
large $n$ and large $\epsilon>0$ such that the supremum probability over $\m{P}$ is less than $\delta$.
The restrictions placed  in the class of models $\m{P}$ along with the proof of Theorem 
\ref{theo_srd_kinf_est_int} give the result.

$N_n$-\textbf{term}: 
part \eqref{eq:kinf:clt} has zero mean and unit variance (see Equation \ref{eq:kinf:clt:sumiid:center}).
Chebyshev's inequality implies that the supremum probability of  the absolute value of part \eqref{eq:kinf:clt} being greater than $\epsilon/4$ is smaller
than $16/\epsilon^{2}$ uniformly over $\m{P}$.

$B_n$-\textbf{term}: $B_n$ is the sum of $B_{n}^+$ (part \eqref{eq:kinf:bias:plus}),
 and $B_{n}^-$ (part \eqref{eq:kinf:bias:minus}). $B_n^+$ converges in probability to zero uniformly over $\m{P}$
 because
 the approximations of  Lemma \ref{lemma_app},
 the bounds on the derivatives of $R(x,d)$, on $f(x)$, on $\sigma^2(x,d)$, and on the rate of $(\m{V}_n^c)^{-1/2}$
 hold uniformly over $\m{P}$.
 The weights $\Delta_j$ do not depend on $P$.
 The same idea applies to $B_{n}^-$.
 Thus, for  $\epsilon>0$,
$\sup_{P \in \m{P}} \mmp_P \left( | B_n | > \epsilon/4 \right)$ converges to zero.

$I_n$-\textbf{term}: 
uniform bounds on the partial derivatives of $\beta(\bc)$ yield
a uniform bound on the approximation error of the numerical integral.
 See  Lemma \ref{lemma_rate_int}. 
The bounds on the rate of $(\m{V}_n^c)^{-1/2}$ also hold uniformly over  $\m{P}$.
 For every $\epsilon>0$
there exists a large $n$ for which $| I_n | \leq \epsilon/4$ holds uniformly over $\m{P}$.

$R_n$-\textbf{term}: $R_n$ is the sum of $R_n^a$ (part \eqref{eq:kinf:op1:a}) and  $R_n^b$ (part \eqref{eq:kinf:op1:b}).
 Lemma \ref{lemma_rest_theo} shows that both converge in probability to zero.
They also converge in probability to zero uniformly over $\m{P}$ for the same reasons
that the $B_n$-term above does.  
Therefore, for  $\epsilon>0$,
$\sup_{P \in \m{P}} \mmp_P \left( | R_n | > \epsilon/4 \right)$ converges to zero.

$\square$

\subsection{Proof of Theorem \ref{theo_frd_id}}
Define $\delta_{j,l}=\mmi\{ {\m{U}}_i(c_j)=d_l \}$.
Assumption \ref{assu_frd_id_nodef} (no ever-defiers) implies the following facts:
(i) $\mmp\left[\delta_{j-1,l}=0,~ \delta_{j,l}=1 \right]=0$  for $\forall l \neq j$;
(ii) $\mmp\left[\delta_{j-1,l}=1,~ \delta_{j,l}=0 \right]=0$  for   $ l = j$;

\noindent
(iii) $\mmp\left[\delta_{j-1,l}=1,~ \delta_{j,l}=0,~ \delta_{j,u}=1 \right]=0$  for  $\forall u \neq j$ and  $u \neq l$.

Fix a small $e>0$ and use fact (i) to obtain

$\mme[Y_i | X_i=c_{j}+e ]
=\sum_{l=0 }^{K} \mme\left[ \delta_{j,l} Y_i(d_{l}) |  X_i =c_{j}+ e  \right]$

$=\sum_{l=0 }^{K} \mme\left[ Y_i(d_{l}) |  X_i =c_{j}+ e,~ \delta_{j,l} = 1,~ \delta_{j-1,l} = 1  \right]
\mmp \left[ \delta_{j,l} = 1,~ \delta_{j-1,l} = 1 |  X_i =c_{j}+ e\right]
$

$
+
\sum_{l=0 }^{K} \mme\left[ Y_i(d_{l}) |  X_i =c_{j}+ e,~ \delta_{j,l} = 1,~ \delta_{j-1,l} = 0  \right]
\mmp \left[ \delta_{j,l} = 1,~ \delta_{j-1,l} = 0 |  X_i =c_{j}+ e\right]
$

$
=\sum_{l=0 }^{K} \mme\left[ Y_i(d_{l}) |  X_i =c_{j}+ e,~ \delta_{j,l} = 1,~ \delta_{j-1,l} = 1  \right]
\mmp \left[ \delta_{j,l} = 1,~ \delta_{j-1,l} = 1 |  X_i =c_{j}+ e\right]
$

$
+
\mme\left[ Y_i(d_{j}) |  X_i =c_{j}+ e,~ \delta_{j,j} = 1,~ \delta_{j-1,j} = 0  \right]
\mmp \left[ \delta_{j,j} = 1,~ \delta_{j-1,j} = 0 |  X_i =c_{j}+ e\right].
$

Take the limit as $e \downarrow 0$. Use that $\{ \delta_{j,l} = 1,~ \delta_{j-1,l} = 1  \}$
and $\{ \delta_{j,j} = 1,~ \delta_{j-1,j} = 0  \}$
are finite unions of measurable sets of the form $\{\m{U}_i = \bar{\m{U}}  \}$, $\bar{\m{U}} \in \m{U}^*$.
The conditional expectation and probability are continuous functions
of $x$ conditional on these sets (Assumption \ref{assu_frd_id_nodef}).

$
\lim_{e \downarrow 0}\mme[Y_i | X_i=c_{j}+e ]
$

$
=\sum_{l=0 }^{K} \mme\left[ Y_i(d_{l}) |  X_i =c_{j},~ \delta_{j,l} = 1,~ \delta_{j-1,l} = 1  \right]
\mmp \left[ \delta_{j,l} = 1,~ \delta_{j-1,l} = 1 |  X_i =c_{j} \right]
$

$
+
\mme\left[ Y_i(d_{j}) |  X_i =c_{j},~ \delta_{j,j} = 1,~ \delta_{j-1,j} = 0  \right]
\mmp \left[ \delta_{j,j} = 1,~ \delta_{j-1,j} = 0 |  X_i =c_{j} \right]
$

$
=\sum_{l=0 }^{K} \mme\left[ Y_i(d_{l}) |  X_i =c_{j},~ \delta_{j,l} = 1,~ \delta_{j-1,l} = 1  \right]
\mmp \left[ \delta_{j,l} = 1,~ \delta_{j-1,l} = 1 |  X_i =c_{j}\right]
$

$
+
\sum_{l=0, l \neq j }^{K} \mme\left[ Y_i(d_{j}) |  X_i =c_{j},~ \delta_{j,j} = 1,~\delta_{j-1,l}=1  \right]
\mmp \left[ \delta_{j,j} = 1,~\delta_{j-1,l}=1 |  X_i =c_{j} \right]
$

Similarly, use fact (ii) for the left-hand-side limit,
$
\lim_{e \downarrow 0} \mme[Y_i | X_i=c_{j}-e ]
$

$
=\sum_{l=0 }^{K} \mme\left[ Y_i(d_{l}) |  X_i =c_{j},~ \delta_{j,l} = 1,~ \delta_{j-1,l} = 1  \right]
\mmp \left[ \delta_{j,l} = 1,~ \delta_{j-1,l} = 1 |  X_i =c_{j} \right]
$

$
+
\sum_{l=0, l \neq j }^{K} \mme\left[ Y_i(d_{j}) |  X_i =c_{j},~ \delta_{j,l} = 0,~ \delta_{j-1,l} = 1  \right]
\mmp \left[ \delta_{j,l} = 0,~ \delta_{j-1,l} = 1 |  X_i =c_{j} \right]
$

Use fact (iii) to get

$
=\sum_{l=0 }^{K} \mme\left[ Y_i(d_{l}) |  X_i =c_{j},~ \delta_{j,l} = 1,~ \delta_{j-1,l} = 1  \right]
\mmp \left[ \delta_{j,l} = 1,~ \delta_{j-1,l} = 1 |  X_i =c_{j} \right]
$

$
+
\sum_{l=0, l \neq j }^{K} \mme\left[ Y_i(d_{l}) |  X_i =c_{j},~ \delta_{j,j} = 1,~ \delta_{j-1,l} = 1 \right]
\mmp \left[ \delta_{j,j} = 1,~ \delta_{j-1,l} = 1 |  X_i =c_{j} \right].
$

The difference between right and left hand side limits is $B_j$

$=
  \sum_{l=0, l \neq j }^{K} \mme\left[ Y_i(d_{j}) - Y_i(d_{l}) |  X_i =c_{j}, \delta_{j,j} = 1,\delta_{j-1,l}=1  \right]
\mmp \left[ \delta_{j,j} = 1, \delta_{j-1,l}=1 |  X_i =c_{j} \right]
$

$
=\sum\limits_{l=0, l \neq j }^{K} \beta_{ec}(c_j,d_l,d_j)
\mmp \left[ \delta_{j,j} = 1,~ \delta_{j-1,l} = 1 |  X_i =c_{j} \right].
$

Next, it is shown that $\mmp \left[ \delta_{j,j} = 1,~ \delta_{j-1,l} = 1 |  X_i =c_{j} \right] = \omega_{j,l}$,
for $l \neq j$.

$
\mmp \left[ \delta_{j,j} = 1,~ \delta_{j-1,l} = 1 |  X_i =c_{j} \right]
=\mmp[\delta_{j,l} =0,\delta_{j-1,l} =1 | X_i =c_{j} ]
$

$
=\mmp[\delta_{j,l} =0| X_i =c_{j} ] - \mmp[\delta_{j-1,l} =0 | X_i =c_{j} ]
$

$
=\lim_{e \downarrow 0}
\left\{
\mmp[ {\m{U}}_i(c_j) \neq d_l | X_i =c_{j}+e ] - \mmp[{\m{U}}_i(c_{j-1}) \neq d_l | X_i =c_{j}-e ] \right\}
$

$
=\lim_{e \downarrow 0} \left\{
\mmp[ D_i=d_l | X_i =c_{j}-e ] - \mmp[ D_i= d_l | X_i =c_{j}+e ]
\right\}
$

where facts (i) and (ii) are used. This proves the first part of the theorem.

If $\beta_{ec}$ belongs to the class of functions of Assumption $\ref{assu_srd_id_param}$,
then $ B_{j} ={\ti{W}_j} \btheta_0$. If the matrix
$\bWt'\bWt=\sum\limits_{j} \widetilde{W}_{j} \widetilde{W}_j '$
is invertible, then the second part of the theorem follows.

Conversely, suppose that the $p>K$ elements in
$\{ \beta_{ec}(c_j,d_l,d_j) \text{ for } (j,l): \omega_{j,l}>0 \}$ are identified
for every fuzzy assignment $\tilde \bc_1 = (c_1,d_0,d_1) $, ..., $\tilde \bc_p = (c_K,d_{K-1},d_K) $.
Identification means that there is an unique solution to the following constrained linear system:
\begin{gather*}
\left[
\begin{array}{c}
B_1
\\
\vdots
\\
B_K
\end{array}
\right]
=
\left[
\begin{array}{cccccccccc}
\omega_{1,0} & \ldots & \omega_{1,K} & 0                   & 0       & 0                   & \ldots & 0 & \ldots & 0
\\
0                   & \ldots & 0                   & \omega_{2,0} & \ldots & \omega_{2,K} & \ldots & 0 & \ldots & 0
\\
\vdots            & \ddots&                      &                      &          &                      &         & \vdots   & \ddots         & \vdots
\\
0                   & \ldots & 0                   & 0                   & \ldots & 0                    & \ldots & \omega_{K,0} & \ldots & \omega_{K,K-1}
\end{array}
\right]
\left[
\begin{array}{c}
\beta_{1}
\\
\beta_{2}
\\
\vdots
\\
\beta_{p}
\end{array}
\right]
\\
\text{ such that } (\beta_{1},\ldots, \beta_{p}) \in  \m{G}.
\end{gather*}

The $K \times p$ matrix of coefficients has rank equal to $K$ because the assignment is fuzzy. Since
$p>K$, the unconstrained system has infinitely many nonzero solutions of the form
$\boup{b} = \boup{b}^p + \sum_{m=1}^{p-K} \lambda_m \boup{b}^s_m$
for any $(\lambda_1, \ldots, \lambda_{p-K}) \in \mmr^{p-K}$, where $\{ \boup{b}^s_m \}_{m=1}^{p-K}$ are the basis vectors of the null-space of the unconstrained system, and $\boup{b}^p$ is a particular solution.
By assumption, the constrained system has one unique solution $\boup{b}^* \in \m{G}$, so $\boup{b}^* + \boup{b}^s_m \not\in \m{G} ~~\forall m$. This implies
that $\boup{b}^s_m \not\in \m{G} ~~\forall m$ because $\m{G}$ is a vector subspace of $\mmr^p$. This is a set of $p-K$ linearly independent vectors in $\mmr^p$ not in $\m{G}$. Therefore, the $dim \m{G} \leq p-(p-K)=K$, and the third part of the theorem follows.
$\square$

\end{singlespace}


\newpage
\setcounter{page}{1}

\bigskip

\begin{center}
 \Large ``Regression Discontinuity Design with Many Thresholds''

\normalsize Marinho Bertanha

\end{center}

\bigskip

\section{Supplemental Appendix}\label{sec_supappen}
\indent


\subsection{Local Polynomial Regressions}
\label{sec_supappen_lpr}
\indent

The first lemma is a straightforward generalization of \cite{porter2003}'s Theorem 3(a).
It derives the asymptotic distribution of the Local Polynomial Regression (LPR) estimator for the difference
in side-limits of a conditional mean.
The lemma considers the mean of the $q \times 1$ vector
$\bY_i$ rather than a scalar $Y_i$ in order to cover the CLT proof in the fuzzy case (Theorem \ref{theo_frd_est})
as a special case with $\bY_i=[Y_i ~~\boup{\m{W}}(c_{j},D_i)']'$.
At a cutoff $c_{j}$, the difference in conditional mean is $\bJ_{j}$, for
$j =1 ,\ldots, K$.
Given a choice of a bandwidth $h_{1j}>0$ for the cutoff $c_j$,
a kernel density function $k(u)$, and a polynomial
order $\rho_1 \in \mathbb{Z}_+$,
the $l$-coordinate of $\bJ_{j}$ is denoted $J_{j,l}$ and estimated as follows:
\begin{align}
\ha{J}_{j,l} = & \ha{a}_{j+,l} - \hat{a}_{j-,l}
\\
(\ha{a}_{j+,l}, \ha{\boup{b}}_{j+,l} )
= &
\argmin\limits_{(a, \boup{b} )}
	\sum\limits_{i=1}^n \Bigg\{ 
		k \left( \frac{X_i - c_{j}}{h_{1j} } \right) v_{i}^{j+}
\nonumber\\
&\hspace{2.5cm}
		\bigg[ Y_{j,l} - a - b_1 (X_i - c_{j}) - \ldots -b_{\rho_1} (X_i - c_{j})^{\rho_1} \bigg]^2
\Bigg\}
\\
\nonumber\\
(\hat{a}_{j-,l}, \hat{\boup{b}}_{j-,l} )
= & 
\argmin\limits_{(a, \boup{b} )}
	\sum\limits_{i=1}^n \Bigg\{
		k \left( \frac{X_i - c_{j}}{h_{1j} } \right) v_{i}^{j-}
\nonumber\\
& \hspace{2.5cm}
		\bigg[ Y_{j,l} - a - b_1 (X_i - c_{j}) - \ldots -b_{\rho_1} (X_i - c_{j})^{\rho_1} \bigg]^2
\Bigg\}
\end{align}
where
\begin{align}
v_i^{j+}= & \mmi\left\{c_{j} \leq X_i < c_{j}+h_{1j}  \right\}
\\
v_i^{j-} = & \mmi\left\{c_{j}-h_{1j} < X_i < c_{j} \right\}
\\
\boup{b}= & (b_1,\ldots,b_{\rho_1}).
\end{align}

\begin{lemma} \label{lemma_porter}

For each $j=1,\ldots,K$, assume the following conditions hold:

\begin{itemize}
\item[(i)] The kernel density function $k:\mathbb{R}\to\mathbb{R}$ is symmetric around zero, 
has compact support $[-M,M]$ for some $M \in (0,\infty)$,
and it is Lipschitz continuous;
\item[(ii)] The distribution of $X_i$ has probability density function $f(x)$ that is continuous
 and has bounded support  $\m{X}=[\underline{ \m{X}},\overline{ \m{X}}]$; the cutoff $c_j$
 belongs to $(\underline{ \m{X}},\overline{ \m{X}})$;
\item[(iii)]
Define
\begin{gather}
\bJ_{j}=\lim_{e \downarrow 0} \left\{ \mme[ \boup{Y}_i|X_i=c_j+e] - \mme[\boup{Y}_i|X_i=c_j-e] \right\}
\\
\bmm(x)=\mme[ \bY_i | X_i=x ] - \sum_{j=1}^K \mmi\{ c_{j} \leq x \} \bJ_{j}
\end{gather}
The function $\bmm(x)$ is at least $\rho_1 +1$ times continuously
differentiable wrt $x$ for all $x$ in a compact interval centered at $c_{j}$ except for  $x=c_{j}$;
there exists
left and right side derivatives at $x=c_{j}$ up to the same order;
its $\rho_1+1$-th partial derivative wrt x is denoted as $\nabla_x^{(\rho_1+1)}\bmm(x)$ and the side limits of the derivatives
are denoted as $\lim\limits_{x \to c_j^{\pm} } \nabla_x^{(\rho_1+1)}\bmm(x) = \nabla_x^{(\rho_1+1)}\bmm(c_j^{\pm}) $;

\item[(iv)] Define
\begin{align}
\beps_i=\bY_i - \mme[ \bY_i | X_i]
\\
\bzeta(x) = \mme\left[\beps_i \beps_i'|X_i=x \right],
\end{align}
and assume $\mme[\| \beps_i \|^3 | X_i] $ is bounded.
The matrix-valued function  $\bzeta(x)$
is continuous wrt $x$ for all $x$ in a compact interval
centered at $c_{j}$  except for  $x=c_{j}$;
there exists left and right side limits at $x=c_{j}$
denoted
$\lim\limits_{x \to c_j^{\pm} }  \bzeta(x) =  \bzeta(c_j^{\pm})$,
where $\bzeta(c_j^{\pm})$ is positive-definite;

\item[(v)] As $n\to\infty$ and  $h_{1j} \to 0$, assume $n h_{1j} \to \infty$ 
and $\sqrt{n h_{1j}} ~ h_{1j}^{p + 1} \to C \in [0,\infty)$
\end{itemize}

\bigskip

Then, for each $j$

\begin{align}
(\bmV_{nj})^{-1/2} \left( \ha{\bJ}_j - \bmB_{nj} - \bJ_j \right)  \dto N(\bzero,\bI)
\end{align}
with
$(\bmV_{nj})^{-1/2}$ being the inverse of the square root of the symmetric and positive-definite matrix $\bmV_{nj}$,
$(\bmV_{nj})^{1/2} =  O \left( \left( n h_{1j} \right)^{1/2}  \right)$, 
$(\bmV_{nj})^{1/2} \bmB_{nj}  =  O_P \left( \left( n h_{1j} \right)^{1/2} h_{1j}^{\rho_1 + 1}  \right)$,
where $\bzero$ is the $q \times 1$ vector of zeros, and 
$\bI$ is the $q\times q$ identity matrix.
The bias $\bmB_{nj}$ and variance $\bmV_{nj}$ terms are characterized as follows,
\begin{align}
\bmB_{nj} = & 
			\frac{ h_{1j}^{\rho_1+1}  f(c_j)} { (\rho_1 + 1 )! } \be_1'  \left[ 
				   \bG_n^{j +}  \bgam^* \nabla^{\rho_1 + 1 }_x \bmm (c_j^+) 
				- 
				\bG_n^{j -}  \bgam^* \nabla^{\rho_1 + 1 }_x\bmm (c_j^-)
			\right]
\label{def:veclpr:bias}
\\
\bmV_{nj} = &n \mme \Bigg\{
	\left[
		\frac{1}{n h_{1j}}
		k\left( \frac{X_i - c_{j}}{h_{1j}} \right)
	\right]^2
	\left[
		\be_1'
		\left(
			v_i^{j +} \mme[ \bG_n^{j +} ]  
			-
			v_i^{j -} \mme[ \bG_n^{j -} ]  
		\right)
		\bHt_{i}^{j}				
	\right]
	\beps_i
\notag \\
& \hspace{1cm}
	\beps_i'
	\left[
		\bHt_{i}^{j'}				
		\left(
			v_i^{j +} \mme[ \bG_n^{j + '} ]  
			-
			v_i^{j -} \mme[ \bG_n^{j - '} ]  
		\right)
		\be_1	
	\right]
\Bigg\}
\label{def:veclpr:var}
\end{align}
where $\beps_i = \bY_i - \mme[\bY_i|X_i]$;
and
\begin{align}
\gamma^*= & [\gamma_{\rho_1 + 1} ~~ \ldots ~~\gamma_{2 \rho_1 +1}]'
\\
\bgam^* =&\bI_q \otimes \gamma^*
\text{,  where }\bI_q \text{ is the }q \times q \text{ identity matrix, and}
 \\
& \otimes \text{ denotes the Kronecker product;} 
\notag \\
\be_1= & \bI_q \otimes e_1 \text{, where }
\\
& e_1  \text{ is the } (\rho_1+1 \times 1) \text{ vector } e_1=[1~0~0~\cdots~0]'
\notag
\\
\gamma_d= & \int_0^1 k(u) u^d du 
\\
H(u) = & \left[ 1 ~~ u ~~ \ldots ~~ u^{\rho_1} \right]'
\\
H_{i}^{j} = & H( X_i - c_{j}  )
\\
\ti{H}_{i}^{j} = & H\left(  \frac{X_i - c_{j}}{ h_{1j} }   \right)
\\
\bH(u) = & \bI_q \otimes H(u)
\\
\bH_i^{j} = & \bH(X_i - c_{j})
\\
\bHt_i^{j} = & \bH\left(  \frac{X_i - c_{j}}{ h_{1j} }   \right)
\\
\bG_n^{j \pm} = & \left[ \frac{1}{n h_{1j} } \sum_{i=1}^n k\left( \frac{X_i - c_{j}}{ h_{1j} } \right)
v_i^{j \pm} \bHt_{i}^{j }  \bHt_{i}^{j'}  \right]^{-1}
\\
\bG^{j \pm} = & f(c_{j}) ^{-1} \bI_{q} \otimes \Gamma_{\pm}^{-1}
\\
\Gamma_+ = & \Gamma, ~~ \Gamma_- =  \{ (-1)^{j+l} \Gamma_{j,l} \}_{j,l}
\\
\Gamma = & \left[
	\begin{array}{ccc}
		\gamma_0 & \ldots & \gamma_{\rho_1}
\\
		\vdots & \vdots & \vdots
\\
		\gamma_{\rho_1} & \ldots & \gamma_{2 \rho_1}
	\end{array}
\right]
\end{align}
where $vec(A_{m \times n})=[a_{1,1}, \ldots, a_{m,1}, a_{1,2},\ldots, a_{m,n}]'$
which makes $\bphi^{j \pm}$ a $q(\rho_1+1) \times 1$ vector.

Moreover,

\begin{align}
\left( \bG_n^{j \pm} - \bG^{j \pm} \right) = & O_P \left( \frac{1}{ \sqrt{ nh_{1j} } } \right)
\\
\mme\left[ \bG_n^{j \pm} \right] = & \bG^{j \pm} + O(h_{1j})
\end{align}

\end{lemma}

\begin{proof}

Following \cite{porter2003},  the
jump estimator is equal to $\ha{\bJ}_{j} = \ha{\ba}_{j}^+ - \ha{\ba}_{j}^-$,
where $\ha{\ba}_{j}^{\pm} = [\ha{a}_{j,1}^{\pm},\ldots,\ha{a}_{j,q}^{\pm}]'$.

\begin{align}
\ha{\ba}_{j}^{\pm} = & \be_1' \left[ \frac{1}{n h_{1j} } \sum_{i=1}^n k\left( \frac{X_i - c_{j}}{ h_{1j} } \right)
v_i^{j \pm} \bH_{i}^{j } \bH_{i}^{j' } \right]^{-1}
\left[ \frac{1}{n h_{1j} } \sum_{i=1}^n k\left( \frac{X_i - c_{j}}{ h_{1j} } \right) v_i^{j \pm} \bH_{i}^{j } \bY_i \right]
\\
= &
\be_1' \bG_n^{j \pm}
\left[ \frac{1}{n h_{1j} } \sum_{i=1}^n k\left( \frac{X_i - c_{j}}{ h_{1j} } \right) v_i^{j \pm} \bHt_{i}^{j } \bY_i \right]
\\
\ha{\bJ}_{j} = &
\be_1' \bG_n^{j +}
\left[ \frac{1}{n h_{1j} } \sum_{i=1}^n k\left( \frac{X_i - c_{j}}{ h_{1j} } \right) v_i^{j +} \bHt_{i}^{j } \bY_i \right]
\notag\\
-&
\be_1' \bG_n^{j -}
\left[ \frac{1}{n h_{1j} } \sum_{i=1}^n k\left( \frac{X_i - c_{j}}{ h_{1j} } \right) v_i^{j -} \bHt_{i}^{j } \bY_i \right]
\end{align}
and note that $\bH_i^{j }$ changes to $\bHt_i^{j}$ because $\be_1$ only takes the first elements of each of the $q$
stacked $\rho_1+1$ vectors.
Define $\bJ_j^*$, $\ti{\bJ}_j$  as follows,
\begin{align}
\bJ_j^* 
= &
		\be_1'\mme[ \bG_n^{j +} ] 
		\frac{1}{nh_{1j}}
		\sum_{i=1}^n
			k\left( \frac{X_i - c_{j}}{h_{1j}} \right)
			v_i^{j +} 
			\bHt_{i}^{j}
			\bY_i
\notag
\\*
& \hspace{1.2cm}
		-\be_1'\mme[ \bG_n^{j -} ] 
		\frac{1}{nh_{1j}}
		\sum_{i=1}^n
			k\left( \frac{X_i - c_{j}}{h_{1j}} \right)
			v_i^{j -} 
			\bHt_{i}^{j}
			\bY_i
	\Big\}
\\
= & 
	\sum_{i=1}^n
		\underset{\bphi_{n}(X_i)}{
			\underbrace{
				\frac{1}{nh_{1j}}
				k\left( \frac{X_i - c_{j}}{h_{1j}} \right)
				\be_1'
				\left(
					v_i^{j +} \mme[ \bG_n^{j +} ]  
					-
					v_i^{j -} \mme[ \bG_n^{j -} ]  
				\right)
				\bHt_{i}^{j}
			}
		}
				\bY_i
\\
= & 
\sum_{i=1}^n \bphi_{n}(X_i) \bY_i
\\
\ti{\bJ}_j = & 
		\be_1'  \bG_n^{j +}
		\mme\left[  
			\frac{1}{nh_{1j}}
			\sum_{i=1}^n
				k\left( \frac{X_i - c_{j}}{h_{1j}} \right)
				v_i^{j +} 
				\bHt_{i}^{j}
				\bY_i^{j+} 
		\right]
\notag
\\*
& \hspace{1.2cm} \left.
		-\be_1'  \bG_n^{j -}  
		\mme\left[
			\frac{1}{nh_{1j}}
			\sum_{i=1}^n
				k\left( \frac{X_i - c_{j}}{h_{1j}} \right)
				v_i^{j -} 
				\bHt_{i}^{j}
				\bY_i^{j-}  
		\right]
	\right\}
\end{align}
where $\bY_i^{j \pm}$ are defined by  
\begin{align}
\bY_i^{j +} = &  \bY_i - \bH_i^{j'} \bphi^{j+} = \bmm(X_i) + \sum_{l=1}^K \mmi\{ c_{l} \leq X_i \}  \bJ_{l} + \beps_i - \bH_i^{j'} \bphi^{j+} 
\label{eq:lemma_porter:def:yij:plus}
\\
\bY_i^{j -} = & \bY_i - \bH_i^{j'} \bphi^{j-} = \bmm(X_i) + \sum_{l=1}^K \mmi\{ c_{l} \leq X_i \} \bJ_{l}  + \beps_i - \bH_i^{j'} \bphi^{j-} 
\label{eq:lemma_porter:def:yij:minus}
\\
\bphi^{j+} = & vec \left[ \bmm({c_{j}}) + \sum_{l=1}^{j} \bJ_{l}
~~ \nabla_x \bmm(c_{j}^+) ~~ \ldots ~~ \nabla_x^{\rho_1} \bmm(c_{j}^+)/\rho_1 ! \right]'
\\
\bphi^{j-} = & vec \left[ \bmm({c_{j}}) + \sum_{l=1}^{j-1} \bJ_{l} ~~ \nabla_x \bmm(c_{j}^-) ~~ \ldots ~~ \nabla_x^{\rho_1} \bmm(c_{j}^-)/\rho_1 ! \right]'.
\end{align}

Write

\begin{align}
(\bmV_{nj})^{-1/2} \left(  \ha{\bJ}_j - \bmB_{nj} - \bJ_j \right)   
= &
(\bmV_{nj})^{-1/2} \left( \bJ_j^* - \mme[\bJ_j^* |\m{X}_n] \right)
\label{eq:veclpr:clt}
\\
+ &
(\bmV_{nj})^{-1/2} \left( \ti{\bJ}_j - \bmB_{nj}  \right)
\label{eq:veclpr:bias}
\\
+ &
(\bmV_{nj})^{-1/2} \left( \ha{\bJ}_j - \mme[ \ha{\bJ}_j |\m{X}_n] - \left( \bJ_j^* - \mme[\bJ_j^* |\m{X}_n] \right) \right)
\label{eq:veclpr:op1:a}
\\
+ &
(\bmV_{nj})^{-1/2} \left( \mme[\ha{\bJ}_j - \bJ_j |\m{X}_n] - \ti{\bJ}_j  \right).
\label{eq:veclpr:op1:b}
\end{align}

The proof applies a central limit theorem (CLT) to show that
part \eqref{eq:veclpr:clt} converges in distribution to a standard normal;
it demonstrates that $\m{B}_{nj}$ approximates the first-order bias, that is, part \eqref{eq:veclpr:bias} converges in probability to zero;
and that parts \eqref{eq:veclpr:op1:a} and \eqref{eq:veclpr:op1:b} converge in probability to zero.

\begin{center}
\textbf{\underline{Part \eqref{eq:veclpr:clt} }} 
\end{center}

\indent

First, find the rate that $(\bmV_{nj})^{-1/2}$ grows.
Use the change of variables $u = (x - c_j)/h_{1j}$ to evaluate the expectation:
\begin{align}
\bmV_{nj} = &\frac{1}{n h_{1j}} \bigintss_{-1}^{1} 
	k\left( u \right)^2
	\left[
		\be_1'
		\left(
			\mmi\{u \geq 0\} \mme[ \bG_n^{j +} ]  
			-
			\mmi\{u < 0\} \mme[ \bG_n^{j -} ]  
		\right)
		\bH(u)				
	\right]
	\bzeta(c_j + u h_{1j})
\notag \\
& \hspace{1.5cm}
	\left[
		\bH(u)							
		\left(
			\mmi\{u \geq 0\} \mme[ \bG_n^{j +'} ]  
			-
			\mmi\{u < 0\} \mme[ \bG_n^{j -'} ]  
		\right)
		\be_1	
	\right]
	f(c_j + u h_{1j})  ~ du
\\
\left\| \bmV_{nj} \right\| > & \frac{M}{n h_{1j}}.
\label{eq:veclpr:clt:s2:rate}
\end{align}
because $\mme[ \bG_n^{j \pm} ]$ is approximately equal to a positive-definite matrix $\bG^{j \pm}$
so that the integral evaluates to a positive-definite matrix.

Second, 
\begin{align}
(\bmV_{nj})^{-1/2} \left( \bJ_j^* - \mme[\bJ_j^* |\m{X}_n] \right)
&=
(\bmV_{nj})^{-1/2}
\sum_{i=1}^n
\bphi_n(X_i) \left( \bY_i - \mme[\bY_i|X_i] \right)
\\
&=
(\bmV_{nj})^{-1/2}
\sum_{i=1}^n
\bphi_n(X_i) \beps_i 
\label{eq:veclpr:clt:sumiid:center}
\end{align}
Equation \ref{eq:veclpr:clt:sumiid:center} is a sum of iid random vectors with zero mean,
where $\bmV_{nj}$ is the variance of the numerator.
The Lindeberg condition is verified next.
Take an arbitrary $\delta>0$. 
\begin{align}
& \sum_{i=1}^{n} \mme\left[
	\left\| \bmV_n \right\|^{-1} \left\|  \bphi_n(X_i)\beps_i \right\|^2 
	\mmi\left\{
				\left\| \bmV_n \right\|^{-1/2} \left\|  \bphi_n(X_i) \beps_i \right\|  > \delta
	\right\}
\right]
\\
= &
 n \mme\left[
	\left\| \bmV_n \right\|^{-1} \left\|  \bphi_n(X_i)\beps_i \right\|^2 
	\mmi\left\{
		\left\|  \bphi_n(X_i) \beps_i \right\|  > \delta \left\| \bmV_n \right\|^{1/2}
	\right\}
\right]
\\
\leq &
 n \mme\left[
	\left\| \bmV_n \right\|^{-1} \left\|  \bphi_n(X_i)\beps_i \right\|^3 
	\delta^{-1} \left\| \bmV_n \right\|^{-1/2}
\right]
\\
\leq &
 M n (n h_{1j})^{3/2} \mme\left[
	\left\|  
					\frac{1}{nh_{1j}}
				k\left( \frac{X_i - c_{j}}{h_{1j}} \right)
				\be_1'
				\left(
					v_i^{j +} \mme[ \bG_n^{j +} ]  
					-
					v_i^{j -} \mme[ \bG_n^{j -} ]  
				\right)
				\bHt_{i}^{j}
				\beps_i 
	\right\|^3 
	\delta^{-1} 
\right]
\\
\leq &
 M \frac{n (n h_{1j})^{3/2}}{n^3 h_{1j}^2} \mme\Bigg[
				\frac{1}{h_{1j}}
				\left|   k\left( \frac{X_i - c_{j}}{h_{1j}} \right) \right|^3			
				\left\|  
					v_i^{j +} \mme[ \bG_n^{j +} ]  
					-
					v_i^{j -} \mme[ \bG_n^{j -} ]  
				\right\|^3
				\left\|   
					\bH\left(\frac{X_i - c_j}{h_{1j}} \right)
				\right\|^3
\notag\\
&\hspace{3cm}
				\mme[\left\|  
					\beps_i 
				\right\|^3 | X_i ] 
	~\delta^{-1} ~
\Bigg]
\\
= &
 M (n h_{1j})^{-1/2} \bigintss_{-1}^1 \Bigg[
				\left|   k\left( u \right) \right|^3			
				\left\|  
					\mmi\{u \geq 0 \} \mme[ \bG_n^{j +} ]  
					-
					\mmi\{u < 0 \} \mme[ \bG_n^{j -} ]  
				\right\|^3
				\left\|   
					\bH\left(u \right)
				\right\|^3
\notag \\*
& \hspace{3.5cm}
 f(c_j + u h_{1j}) \Bigg] ~ du
\\
\leq &  M (n h_{1j})^{-1/2}  =o(1)
\end{align}
where the inequality $x^2 \mmi\{|x|> \delta \} \leq x^3 \delta^{-1}$,
boundedness of $\mme[\left\|  \beps_i \right\|^3 | X_i ]$,
and the rate of $\bmV_{nj}^{-1}$ are used.
The multivariate Lindeberg-Feller CLT says that Equation \ref{eq:veclpr:clt:sumiid:center}, 
and thus part \eqref{eq:veclpr:clt}, converges in distribution to a standard normal.

\begin{center}
\textbf{\underline{Part \eqref{eq:veclpr:bias} }} 
\end{center} 
 
First consider,
\begin{align}
& \mme\left[
	\frac{1}{ h_{1j} } 
	k\left( \frac{X_i - c_j}{ h_{1j} } \right)
	v_{i}^{j+} \bHt_{i}^{j}  \mme\left[ \bY_i^{j+} | X_i \right]
\right] 
\\
= & \mme\left[
	\frac{1}{h_{1j}} 
	k\left( \frac{X_i - c_j}{h_{1j}} \right)
	v_{i}^{j+}  \bHt_{i}^{j}
	\frac{ \nabla^{\rho_1 + 1 }_x \bmm(c_j^+)}{ (\rho_1 + 1 )! }
	\left( \frac{X_i-c_j}{h_{1j}} \right)^{\rho_1+1} h_{1j}^{\rho_1+1}  
\right]
\label{eq:veclpr:bias:expansion:fo}
\\
 & \hspace{.5cm} + 
\mme\left[
	\frac{1}{h_{1j}} 
	k\left( \frac{X_i - c_j}{h_{1j}} \right)
	v_{i}^{j+} \bHt_{i}^{j}
	\frac{ \nabla^{\rho_1 + 2 }_x\bmm(c_j^*)}{ (\rho_1 + 2 )! }
	\left( \frac{X_i-c_j}{h_{1j}} \right)^{\rho_1+2} h_{1j}^{\rho_1+2}  
\right]
\label{eq:veclpr:bias:expansion:so}
\\
= & 
h_{1j}^{\rho_1+1}  
f(c_j) \bgam^*
\frac{ \nabla^{\rho_1 + 1 }_x\bmm (c_j^+)}{ (\rho_1 + 1 )! }
+ O\left( h_{1j}^{\rho_1+2} \right),
\label{eq:veclpr:bias:expansion}
\\
& \text{ and }
\nonumber
\\
\bmB_{nj} = &  \bmB_{nj}^+ -  \bmB_{nj}^-
\\
\bmB_{nj}^+ = & 			
			\frac{ h_{1j}^{\rho_1+1} f(c_j) }{ (\rho_1 + 1 )! }
			\be_1'  \bG_n^{j +} \bgam^*
			\nabla^{\rho_1 + 1 }_x\bmm (c_j^+)
\\
\bmB_{nj}^- = & 			
			\frac{ h_{1j}^{\rho_1+1} f(c_j) }{ (\rho_1 + 1 )! }
			\be_1'  \bG_n^{j -} \bgam^*
			\nabla^{\rho_1 + 1 }_x\bmm (c_j^-)
\end{align}
where $\mme\left[ \bY_i^{j+} | X_i \right]$ is the difference between $\mme[\bY_i|X_i]$ 
and its $\rho_1$-th order Taylor expansion around $X_i=c_j$ (see Equations \ref{eq:lemma_porter:def:yij:plus} and \ref{eq:lemma_porter:def:yij:minus}).
The expectations in Equations \ref{eq:veclpr:bias:expansion:fo} and \ref{eq:veclpr:bias:expansion:so},
without the $h_{1j}^{\rho_1+1} $ and $h_{1j}^{\rho_1+2}$ terms, are bounded over $j$
because the kernel, derivatives, and polynomials are bounded functions of $u = (x-c_j)h_{1j}^{-1}$.  

Next, 
\begin{align} 
(\bmV_{nj})^{-1/2} \left( \ti{\bJ}_j - \bmB_{nj}  \right) = &
(\bmV_{nj})^{-1/2}	\be_1'  \bG_n^{j +}
		\mme\left[  
			\frac{1}{nh_{1j}}
			\sum_{i=1}^n
				k\left( \frac{X_i - c_{j}}{h_{1j}} \right)
				v_i^{j +} 
				\bHt_{i}^{j}
				\bY_i^{j+} 
		\right]
\notag\\
-&
(\bmV_{nj})^{-1/2} \bmB_{nj}^+
\label{eq:veclpr:bias:plus}
\\
- & 
(\bmV_{nj})^{-1/2}	\be_1'  \bG_n^{j -}  
		\mme\left[
			\frac{1}{nh_{1j}}
			\sum_{i=1}^n
				k\left( \frac{X_i - c_{j}}{h_{1j}} \right)
				v_i^{j -} 
				\bHt_{i}^{j}
				\bY_i^{j-}  
		\right]
\notag\\
+&
(\bmV_{nj})^{-1/2} \bmB_{nj}^-
\label{eq:veclpr:bias:minus}
\end{align}

Consider part \eqref{eq:veclpr:bias:plus}. Part \eqref{eq:veclpr:bias:minus} follows a symmetric argument.
Use \eqref{eq:veclpr:bias:expansion} and write
\begin{align}
(\ref{eq:veclpr:bias:plus}) = & 
(\bmV_{nj})^{-1/2}	\left[ 
			\be_1'  \bG_n^{j +} \bgam^* 
			h_{1j}^{\rho_1+1}  
			\frac{ \nabla^{\rho_1 + 1 }\bmm (c_j^+)}{ (\rho_1 + 1 )! }
			f(c_j) 
			- \bmB_{nj}^+
		\right]
\\
+ & 
(\bmV_{nj})^{-1/2}	\be_1'  \bG_n^{j +}
				 O\left( h_{1j}^{\rho_1+2} \right)
\\
=& 
0 + O_P\left( \sqrt{n h_{1j}}  h_{1j}^{\rho_1+2} \right) = o_P(1)
\end{align}
where the second equality uses the definition of 
$\bmB_{nj}^+$, the fact that $\bG_n^{j +}=O_P(1)$, and the rate condition
$\left( n h_{1j} \right)^{1/2} h_{1j}^{\rho_1+1} =O(1) $.

\bigskip

\begin{center}
\textbf{\underline{Part \eqref{eq:veclpr:op1:a} }} 
\end{center} 

\begin{align}
\eqref{eq:veclpr:op1:a} = & (\bmV_{nj})^{-1/2}	\Bigg[
		\be_1' \bG_n^{j +}  
		\frac{1}{nh_{1j}}
		\sum_{i=1}^n
			k\left( \frac{X_i - c_{j}}{h_{1j}} \right)
			v_i^{j +} 
			\bHt_{i}^{j}
			\beps_i
\notag
\\*
& \hspace{2cm}
		-\be_1' \bG_n^{j -}  
		\frac{1}{nh_{1j}}
		\sum_{i=1}^n
			k\left( \frac{X_i - c_{j}}{h_{1j}} \right)
			v_i^{j -} 
			\bHt_{i}^{j}
			\beps_i
\Bigg]
\\
- &(\bmV_{nj})^{-1/2}	\Bigg[
		\be_1' \mme\left[ \bG_n^{j +}\right]
		\frac{1}{nh_{1j}}
		\sum_{i=1}^n
			k\left( \frac{X_i - c_{j}}{h_{1j}} \right)
			v_i^{j +} 
			\bHt_{i}^{j}
			\beps_i
\notag
\\*
& \hspace{2cm}
		-\be_1' \mme\left[ \bG_n^{j -}\right] 
		\frac{1}{nh_{1j}}
		\sum_{i=1}^n
			k\left( \frac{X_i - c_{j}}{h_{1j}} \right)
			v_i^{j -} 
			\bHt_{i}^{j}
			\beps_i
\Bigg]
\\
= & (\bmV_{nj})^{-1/2}	
		\be_1' \left[ \bG_n^{j +}  - \mme\left[ \bG_n^{j +}\right]  \right]
		\frac{1}{nh_{1j}}
		\sum_{i=1}^n
			k\left( \frac{X_i - c_{j}}{h_{1j}} \right)
			v_i^{j +} 
			\bHt_{i}^{j}
			\beps_i
\\*
-& (\bmV_{nj})^{-1/2}
		\be_1' \left[ \bG_n^{j -}  - \mme\left[ \bG_n^{j -}\right]   \right]  
		\frac{1}{nh_{1j}}
		\sum_{i=1}^n
			k\left( \frac{X_i - c_{j}}{h_{1j}} \right)
			v_i^{j -} 
			\bHt_{i}^{j}
			\beps_i
\\
= & O\left( (nh_{1j})^{1/2} \right) O_P\left( (nh_{1j})^{-1/2} \right) o_P(1)
\\
+&O\left( (nh_{1j})^{1/2} \right) O_P\left( (nh_{1j})^{-1/2} \right) o_P(1)
\\
=& o_P(1)
\end{align}
because of $\left[ \bG_n^{j \pm}  - \mme\left[ \bG_n^{j \pm}\right]  \right]=O_P\left( (nh_{1j})^{-1/2} \right)$,
and the fact that the zero mean terms $(nh_{1j})^{-1}$
		$\sum_{i=1}^n$
			$k\left( \frac{X_i - c_{j}}{h_{1j}} \right)$
			$v_i^{j \pm} $
			$\bHt_{i}^{j}$
			$\beps_i $
			converge in probability to zero since their variances are $O\left( (nh_{1j})^{-1}\right)$.

\begin{center}
\textbf{\underline{Part \eqref{eq:veclpr:op1:b} }} 
\end{center}

Use the definitions of  $\bphi^{j \pm}$ and $\bY_i^{j \pm}$ to write:
\begin{align}
\ha{\bJ}_{j} - \bJ_{j} = & \ha{\ba}_{j}^{+} - \ha{\ba}_{j}^{-} - \bJ_{j}
\\
=& \be_1' \left[ \frac{1}{n h_{1j} } \sum_{i=1}^n k\left( \frac{X_i - c_{j}}{ h_{1j} } \right)
v_i^{j +} \bH_{i}^{j } \bH_{i}^{j' } \right]^{-1}
\left[ \frac{1}{n h_{1j} } \sum_{i=1}^n k\left( \frac{X_i - c_{j}}{ h_{1j} } \right) v_i^{j +} \bH_{i}^{j } \bY_i \right]
\notag \\
- &
\be_1' \left[ \frac{1}{n h_{1j} } \sum_{i=1}^n k\left( \frac{X_i - c_{j}}{ h_{1j} } \right)
v_i^{j -} \bH_{i}^{j } {\bH_{i}^{j' }} \right]^{-1}
\left[ \frac{1}{n h_{1j} } \sum_{i=1}^n k\left( \frac{X_i - c_{j}}{ h_{1j} } \right) v_i^{j -} \bH_{i}^{j } \bY_i \right]
\notag \\
- & \be_1' \left( \bphi^{j+} - \bphi^{j-} \right)
\\
= & \be_1' \left[ \frac{1}{n h_{1j} } \sum_{i=1}^n k\left( \frac{X_i - c_{j}}{ h_{1j} } \right)
v_i^{j +} \bH_{i}^{j } {\bH_{i}^{j' }} \right]^{-1}
\left[ \frac{1}{n h_{1j} } \sum_{i=1}^n k\left( \frac{X_i - c_{j}}{ h_{1j} } \right) v_i^{j +} \bH_{i}^{j }
\bY_i^{j+}
\right]
\notag \\*
- &
\be_1' \left[ \frac{1}{n h_{1j} } \sum_{i=1}^n k\left( \frac{X_i - c_{j}}{ h_{1j} } \right)
v_i^{j -} \bH_{i}^{j} {\bH_{i}^{j' }} \right]^{-1}
\left[ \frac{1}{n h_{1j} } \sum_{i=1}^n k\left( \frac{X_i - c_{j}}{ h_{1j} } \right) v_i^{j -} \bH_{i}^{j }
\bY_i^{j-}
\right]
\\
= & \be_1' \bG_n^{j +}
\left[ \frac{1}{n h_{1j} } \sum_{i=1}^n k\left( \frac{X_i - c_{j}}{ h_{1j} } \right) v_i^{j +} \bHt_{i}^{j} \bY_i^{j+}  \right]
\nonumber
\\
- &
\be_1' \bG_n^{j -}
\left[ \frac{1}{n h_{1j} } \sum_{i=1}^n k\left( \frac{X_i - c_{j}}{ h_{1j} } \right) v_i^{j -} \bHt_{i}^{j} \bY_i^{j-}  \right]
\end{align}

Thus, part \eqref{eq:veclpr:op1:b} becomes
\begin{align}
\eqref{eq:veclpr:op1:b} =&
(\bmV_{nj})^{-1/2} \be_1' \bG_n^{j +}
\left[ \frac{1}{n h_{1j} } \sum_{i=1}^n k\left( \frac{X_i - c_{j}}{ h_{1j} } \right) v_i^{j +} \bHt_{i}^{j} \mme[ \bY_i^{j+} | X_i]  \right]
\nonumber
\\
-&
(\bmV_{nj})^{-1/2} \be_1' \bG_n^{j +}
\mme \left[ \frac{1}{n h_{1j} } \sum_{i=1}^n k\left( \frac{X_i - c_{j}}{ h_{1j} } \right) v_i^{j+} \bHt_{i}^{j} \mme[ \bY_i^{j+} | X_i]  \right]
\label{eq:veclpr:op1:b:plus}
\\
- &
(\bmV_{nj})^{-1/2} \be_1' \bG_n^{j -}
\left[ \frac{1}{n h_{1j} } \sum_{i=1}^n k\left( \frac{X_i - c_{j}}{ h_{1j} } \right) v_i^{j -} \bHt_{i}^{j} \mme[ \bY_i^{j-} | X_i]  \right]
\notag\\
+&
(\bmV_{nj})^{-1/2} \be_1' \bG_n^{j -}
\mme \left[ \frac{1}{n h_{1j} } \sum_{i=1}^n k\left( \frac{X_i - c_{j}}{ h_{1j} } \right) v_i^{j -} \bHt_{i}^{j} \mme[ \bY_i^{j-} | X_i]  \right]
\label{eq:veclpr:op1:b:minus}
\end{align}

The next steps show that part \eqref{eq:veclpr:op1:b:plus} converges in probability to zero.
A symmetric proof shows that part \eqref{eq:veclpr:op1:b:minus} also converges in probability to zero.
\begin{align}
\eqref{eq:veclpr:op1:b:plus} =&
(\bmV_{nj})^{-1/2} \be_1' \bG_n^{j +} 
	\Bigg\{
		\frac{1}{n h_{1j} } \sum_{i=1}^n  
		k\left( \frac{X_i - c_{j}}{ h_{1j} } \right) v_i^{j +} \bHt_{i}^{j} \mme[ \bY_i^{j+} | X_i] 
\nonumber
\\
&\hspace{3cm} 
		-\mme \left[ \frac{1}{n h_{1j} } \sum_{i=1}^n   
			k\left( \frac{X_i - c_{j}}{ h_{1j} } \right) v_i^{j +} \bHt_{i}^{j} \mme[ \bY_i^{j+} | X_i]  
		\right]
\Bigg\}
\\
=&
(\bmV_{nj})^{-1/2} \be_1' \bG_n^{j +} h_{1j}^{\rho_1+1}
	\Bigg\{
		\frac{1}{n h_{1j} } \sum_{i=1}^n  
		k\left( \frac{X_i - c_{j}}{ h_{1j} } \right) v_i^{j +} \bHt_{i}^{j} \mme[ \bY_i^{j+} | X_i] h_{1j}^{-(\rho_1+1)}
\nonumber
\\
&\hspace{4cm} 
		-\mme \left[ \frac{1}{n h_{1j} } \sum_{i=1}^n   
			k\left( \frac{X_i - c_{j}}{ h_{1j} } \right) v_i^{j +} \bHt_{i}^{j} \mme[ \bY_i^{j+} | X_i]  
		\right] h_{1j}^{-(\rho_1+1)}
\Bigg\}
\\
&= O\left( (nh_{1j})^{1/2} \right) O_P(1) h_{1j}^{ \rho_1+1 } O_P\left( (nh_{1j})^{-1/2} \right) = o_P(1)
\end{align}
where the zero mean term in curly brackets is normalized by $h_{1j}^{\rho_1+1}$ 
(see Equation \ref{eq:veclpr:bias:expansion}), and its variance after the normalization decreases at $(nh_{1j})^{-1}$. 

\end{proof}

\subsection{Uniformity with Large Number of Cutoffs}\label{sec:supp:app:uniform}
\indent

A class of sets $\mathcal{S}$ of a space $\Omega$ is said to shatter a $n$-point subset of $\Omega$, $D_n$, if for
every subset of $D_n$, $D^{(i)}_n$, there exists a set in $\mathcal{S}$, $S$, such that $S \cap D^{(i)}_n = D^{(i)}_n$. A
class of sets $\mathcal{S}$ is said to be a VC class if there exists a finite non-negative integer $v$ such that 
 no $v$-point set $D_v$ is shattered by $\mathcal{S}$. 
In this case, the index of the VC class is $v$. 
For a class of functions from $\Omega$ to $\mathbb{R}$, $\m{F}$,
call the class of graphs of $\m{F}$, $g \m{F} = \{ (x,t) \in \Omega \times \mathbb{R} : t \leq f(x) \leq 0 \text{ or }
0 \leq t \leq f(x) \text{ for }f \in \m{F} \}$.
A class of functions $\m{F}$ is called a VC-subgraph class
 if $gF$ is a VC class.\footnote{One may define VC subgraph using
 alternative definitions of class of graphs, but those lead to definitions of 
VC subgraph that are equivalent to ours. See \cite{van1996weak}'s Problem 2.6.11.}
The class $\m{F}$ is enveloped by function $F$
if $\forall f \in \m{F}$, $|f(x)|\leq F(x)$.
Let $(\Omega,\mathcal{A},Q)$ be a probability space. A covering number $N_1(\varepsilon,Q, \m{F})$
is defined to be the smallest non-negative integer $m$ for which there exists functions $f_1, \ldots, f_m$ in $ \m{F}$ such
that $\min\limits_j E_Q |f-f_j | \leq \varepsilon$ for every $f \in  \m{F}$.

It is possible to build a complex VC-subgraph class by combining basic VC-subgraph classes.
Any class of functions made of a finite union or intersection of VC-subgraph classes
is also VC subgraph (\cite{pollard1984convergence}'s Lemma 2.15).
Let $\phi:\mathbb{R} \to \mathbb{R}$ be a monotone function. Define the class of functions
which consists of translations of this monotone function $\phi$. That is,
$ \m{F}=\{ f:\mathbb{R}\to\mathbb{R}\text{ with }  f(x)=\phi \left( x-c \right)~\forall c \in \mathbb{R}\}$.
Then, $\m{F}$ is a VC-subgraph class with index equal to 2 (\cite{van1996weak}'s Lemma 2.6.16).
Moreover, if $\m{G}$ is VC subgraph, then $\phi \circ \m{G} = \{\phi(g) : g\in \m{G} \}$
is VC subgraph (\cite{van1996weak}'s Lemma 2.6.18).
A VC-subgraph class $\m{F}$ of uniformly bounded
functions has covering number $N_1(\varepsilon,Q, \m{F}) \leq A \varepsilon ^ {-W}$, where the constants $A,W$
depend only on the VC index of the class of functions and on the uniform bound
(\cite{pollard1984convergence}'s Lemma 2.25). The next lemma lists more properties. 
\begin{lemma}\label{lemma_covnum_sumprod}
Let $\m{F}$ and $\m{G}$ be VC-subgraph classes of functions uniformly bounded by a constant $0<M<\infty$. Define
$\m{H}_{+}=\{ f+g : f \in \m{F}, g \in \m{G} \}$ and $\m{H}_{\times }=\{ fg : f \in \m{F}, g \in \m{G} \}$.
For a fixed Lipschitz continuous function $\phi$ with Lipschitz constant $C$, define
$\m{H}_{\phi}=\{ \phi(f) : f \in \m{F} \}$.
Then,
\begin{enumerate}
	\item $N_1(\varepsilon,Q,\m{H}_{+}) \leq N_1(\varepsilon/2,Q,\m{F}) N_1(\varepsilon/2,Q,\m{G})$
	\item $N_1(\varepsilon,Q,\m{H}_{\times }) \leq N_1(\varepsilon/2M,Q,\m{F}) N_1(\varepsilon/2M,Q,\m{G})$
	\item $N_1(\varepsilon,Q,\m{H}_{\phi }) \leq N_1(\varepsilon/C,Q,\m{F}) $
\end{enumerate}
\end{lemma}

\begin{proof}
Slightly modified from Theorem 3 in \cite{andrews1994empirical}.

Fix $\varepsilon>0$, pick any $h \in \m{H}_+$. It is known that $h=f+g$.

Use $f_i + g_j$ to approximate $f+g$, where $E_Q|f-f_i|\leq \varepsilon/2$
and $E_Q|g-g_j|\leq \varepsilon/2$, $1 \leq i \leq N_1(\varepsilon/2,Q,\m{F}) $, $1 \leq j \leq N_1(\varepsilon/2,Q,\m{G}) $.
It is known that these two covering numbers are finite since $\m{F}$ and $\m{G}$ are VC-subgraph.
Call $h_l=f_i + g_j$, with $1 \leq l \leq N_1(\varepsilon/2,Q,\m{F}) N_1(\varepsilon/2,Q,\m{G})$.

$E_Q|h-h_l| = E_Q|f+g -(f_i+g_j)| \leq E_Q|f -f_i| + E_Q| g -g_j| \leq \varepsilon$

Therefore, $N_1(\varepsilon,Q,\m{H}_{+}) \leq N_1(\varepsilon/2,Q,\m{F}) N_1(\varepsilon/2,Q,\m{G})$.
\bigskip

Now, pick any $h \in \m{H}_{\times }$. It is known that $h=fg$.

Use $f_i  g_j$ to approximate $fg$, where $E_Q|f-f_i|\leq \varepsilon/2M$
and $E_Q|g-g_j|\leq \varepsilon/2M$, $1 \leq i \leq N_1(\varepsilon/2M,Q,\m{F}) < \infty$, $1 \leq j \leq N_1(\varepsilon/2M,Q,\m{G})< \infty $.
Call $h_l=f_i  g_j$, with $1 \leq l \leq N_1(\varepsilon/2M,Q,\m{F}) \allowbreak N_1(\varepsilon/2M,Q,\m{G})$.

\begin{gather*}
E_Q|h-h_l| = E_Q|fg -f_i g_j| = E_Q|fg -f_i g_j -f_i g +f_i g|   \\
\leq E_Q|f - f_i| |g| + E_Q|g_j - g| |f_i | \leq M \left(  E_Q|f - f_i| + E_Q|g_j - g|  \right) \leq \varepsilon
\end{gather*}

Therefore, $N_1(\varepsilon,Q,\m{H}_{\times }) \leq N_1(\varepsilon/2M,Q,\m{F}) N_1(\varepsilon/2M,Q,\m{G})$.

\bigskip

Lastly, pick $h \in \m{H}_{\phi}$, so that $h = \phi(f)$ for some $f \in \m{F}$.
Use $f_i$ to approximate $f$, where $E_Q|f-f_i|\leq \varepsilon/C$,
 $1 \leq i \leq N_1(\varepsilon/C,Q,\m{F}) < \infty$.
Call $h_i=\phi(f_i)$ for each $i$.

$E_Q|h-h_i| = E_Q|\phi(f) - \phi(f_i)| \leq C E_Q|f-f_i|\leq \varepsilon$.

Therefore, $N_1(\varepsilon,Q,\m{H}_{\phi}) \leq N_1(\varepsilon/C,Q,\m{F}) $.

\end{proof}

Consider a set of $K+2$ positive bandwidth sequences $\underline{h}_1$, $h_{1j}$, $j=1,\ldots, K$, 
$\overline{h}_1$ that depend on $n$.
Assume $\underline{h}_1 \leq h_{1j} \leq \overline{h}_1$ for every $j$,
and that both $\underline{h}_1$ and $\overline{h}_1$ converge to zero at the same rate.
Define $v_{c,h}^{+}(x)=\mmi\{ c \leq x < c+h \}$,
and $v_{c,h}^{-}(x) = \mmi\{c-h < x < c\}$
for any $c \in \mathbb{R}$ and $h>0$, so that $v_i^{j \pm}$ (used in the main text) becomes $v_{c_j,h_{1j}}^{\pm}(X_i)$.

\begin{lemma}\label{lemma_classf}
Consider the classes of functions defined below $\m{F}_{j}^{\pm}$, $j=1,\ldots, 4$.
They depend on $n$ because the bandwidth sequences 
$\underline{h}_1$, $h_{1j}$, $j=1,\ldots, K$, 
and $\overline{h}_1$
enter their definitions.
\begin{enumerate}
\item $ \m{F}_{1}^{\pm}=
\left\{ 
	f_{c,h} : \m{X}  \to \mathbb{R} \text{ st } f_{c,h}(x)= v_{c,h}^\pm(x)
		k\left(\frac{x-c}{h}\right)
		~, c \in \m{X}
		~ , h \in [\underline{h}_1, \overline{h}_1]   
\right\}$
for a kernel density function $k(\cdot)$ that satisfies Assumption \ref{assu_srd_est_kernel};
\item $\m{F}_{2}^{\pm}=
\left\{
	f_{c,h}: \m{X} \times [-M;M] \to \mathbb{R} \text{ st }
		f_{c,h}(x,y)=v_{c,h}^\pm(x)
		k\left(\frac{x-c}{h} \right) y
		~, c \in \m{X}
		~ , h \in [\underline{h}_1, \overline{h}_1]   
\right\}$
for any $M \in (0, \infty)$;

\item
$ \m{F}_{3}^{\pm}=
\left\{
	f_{c,h}:\m{X} \to \mathbb{R} \text{ st } f_{c,h}(x)= v_{c,h}^\pm(x)
	k\left(\frac{x-c}{h}\right)
	r^{\pm}(x)]
	~, c \in \m{X}
	~ , h \in [\underline{h}_1, \overline{h}_1]   
\right\}$ 
where
$r^{\pm}(x)=\sum_{j =1}^K  v_{c_{j},h_{1j}}^\pm(x) \mme[Y_i^{j \pm}|X_i=x]$
and $Y_i^{j \pm}$ is defined in Lemma \ref{lemma_porter} (scalar case);

\item
$\m{F}_{4}^{\pm}=
\left\{
	f_{c,h}:\m{X}  \to \mathbb{R} \text{ st } f_{c,h}(x)= v_{c,h}^\pm(x)
	k\left(\frac{x-c}{h}\right)
	\left(\frac{x-c}{h}\right)^l
	~, c \in \m{X}
	~ , h \in [\underline{h}_1, \overline{h}_1]
\right\}$
for any positive integer $l\in \mathbb{Z}_+$.
\end{enumerate}

If Assumptions \ref{assu_srd_est_kernel}, \ref{assu_srd_est_mx}, and \ref{assu_srd_kinf_est_int} hold,
these functions are bounded for large $n$.
 The covering number of each of these classes satisfies
$N_1(\varepsilon,Q,\m{F}_{j}) \leq A_j \varepsilon^{-W_j}$, $j=1,2,3,4$,
where the positive constants $A_j$ and $W_j$
are independent of $n$ and $Q$.
\end{lemma}
\begin{proof}

First, note that all these functions are bounded.
The functions in the first two classes are bounded because the kernel and the indicator functions are bounded.
For the third class of functions,
\begin{gather*}
\left|
r^{+}(x)\right|
=
\left|\sum_{j=1}^K  v_{c_{j},h_{1j}}^+ (x)
\mme[Y_i^{j +}|X_i=x]
\right|
\\
\leq \max_{j} \left| \mme[Y_i^{j+}|X_i=x] \right|
=\max_{j} \left|
\left[ \nabla_x^{\rho_1+1} R(c_{j}^*(x),d_{j}) /(\rho_1+1)! \right](x-c_{j})^{\rho_1+1}
\right|
\end{gather*}
where
$c_{j}^*(x)\in(c_{j},x)$.
The function $ r^{+}(x)$ is bounded because $\nabla_x^{\rho_1+1} R(\cdot)$ is bounded (Assumption \ref{assu_srd_kinf_est_int}). An analogous argument bounds $r^{-}(x)$.
For the fourth class of functions,
\begin{gather*}
0 \leq v_{c,h}^+ \left( \frac{x-c}{h} \right)^l < 1
\\
-1 < v_{c,h}^- \left( \frac{x-c}{h} \right)^l < 0
\end{gather*}

Second, note that each of these classes is  made out of the product of the following (uniformly bounded) classes of functions:
\begin{enumerate}

\item 
$\m{G}_{1}^{\pm} = \left\{f_{c,h}: \m{X} \to \mmr \text{ st } 
	f_{c,h}(x)=
	v_{c,h}^\pm(x)
	\left(\frac{x-c}{h}\right)^l, ~ c \in \m{X},
	~ h \in [\underline{h}_1, \overline{h}_1] 
\right\}$ ;

\item $\m{G}_{2}^\pm = \left\{f_{c,h}: \m{X} \to \mmr \text{ st } 
	f_{c,h}(x)=
	 v_{c,h}^\pm(x), ~ c \in \m{X}, 
	 ~ h \in [\underline{h}_1, \overline{h}_1]
\right\}$;

\item $\m{G}_{3}=\left\{f: [-M,M] \to \mmr \text{ st }
	f(y)=y
\right\}$, that is,  only one function $f$;

\item $\m{G}_{4}^\pm=\left\{f: \m{X} \to \mmr \text{ st }
	f(x) = r^\pm(x) 
\right\}$, that is, only one function $r^\pm$;

\item $\m{G}_{5}^{\pm} = \left\{f_{c,h}: \m{X} \to \mmr \text{ st } 
	f_{c,h}(x) = v_{c,h}^\pm(x) k\left(\frac{x-c}{h}\right), ~ c \in \m{X}, 
	~ h \in [\underline{h}_1, \overline{h}_1] 
\right\}$.

\end{enumerate}

Lemma \ref{lemma_covnum_sumprod} says that it suffices to show that each of these classes has a polynomial bound on the covering number with constants
that are independent of $n$ and $Q$. 

\bigskip
$\m{G}_{1}^{\pm} = \left\{f_{c,h}: \m{X} \to \mmr \text{ st } 
	f_{c,h}(x)=
	v_{c,h}^\pm(x)
	\left(\frac{x-c}{h}\right)^l, ~ c \in \m{X},
	~ h \in [\underline{h}_1, \overline{h}_1] 
\right\}$ \\*
Take $\m{G}_{1}^{+} $ WLOG.
A function $v_{c,h}^+(x) 	\left(\frac{x-c}{h}\right)$
is a line connecting the point $(c,0)$ to $(c+h,1)$ with support $[c,c+h)$.
The class of functions $\m{G}_{1}^{*} = \Big\{ f_{c,h}: \m{X} \to \mmr \text{ st } 
	f_{c,h}(x)=
	v_{c,h}^+(x)
	\left(\frac{x-c}{h}\right), ~ c \in \m{X},
	~ h \in [\underline{h}_1, \overline{h}_1] 
\Big\}$
is VC subgraph because no 4-point set is shattered. 
It has covering number $N_1(\eps,Q,\m{G}_1^*)$
bounded by a polynomial in $\eps$ whose constants do not depend on $Q$ or $n$.
The function $\phi(x)=x^l$ defined over $[0,1]$ is Lipschitz continuous with constant equal to $l$.
Since   $\m{G}_{1}^{+} = \Big\{\phi(g) : g \in \m{G}_{1}^{*} \Big\}$,
Lemma \ref{lemma_covnum_sumprod} says $\m{G}_{1}^{+}$
has covering number
bounded above by $A_1 \varepsilon ^ {-W_5}$ with $A_1,W_1$ independent of $n$ or $Q$.

\bigskip
$\m{G}_{2}^\pm = \left\{f_{c,h}: \m{X} \to \mmr \text{ st } 
	f_{c,h}(x)=
	 v_{c,h}^\pm(x), ~ c \in \m{X}, 
	 ~ h \in [\underline{h}_1, \overline{h}_1]
\right\}$\\*
For either $v_{c,h}^+$ or $v_{c,h}^-$, no 3-point set is shattered by the graphs of either
$\m{G}_2^+$ or $\m{G}_2^-$.
Hence, $\m{G}_{2}^\pm$ is VC subgraph
with covering number bounded above by $A_2^{\pm} \varepsilon ^ {-W_2^{\pm}}$ where
$A_2^\pm,W_2^\pm$ are independent of $n$ or $Q$.

\bigskip
$\m{G}_{3}=\left\{f: [-M,M] \to \mmr \text{ st }
	f(y)=y
\right\}$\\*
It is straightforward to see that the graphs of this class of functions is VC with index 2. Therefore, the covering number
of $\m{G}_{3}$ is bounded above by $A_3 \varepsilon ^ {-W_3}$ with $A_3,W_3$ independent of $n$ or $Q$.

\bigskip
$\m{G}_{4}^\pm=\left\{f: \m{X} \to \mmr \text{ st }
	f(x) = r^\pm(x) 
\right\}$\\*
Consider $\m{G}_{4}^+$ WLOG. 
For each $n$,
$r^{+}(x)=\sum_{j=1}^K v_{c_{j},h_{1j}}^+ (x) \mme[Y_i^{j +}|X_i=x]$
is a fixed function.
Similar to $\m{G}_{3}$, the covering number
of $\m{G}_{4}^+$ is bounded above by $A_4 \varepsilon ^ {-W_4}$ with $A_4,W_4$ independent of $n$ or $Q$.

\bigskip

$\m{G}_{5}^{\pm} = \left\{f_{c,h}: \m{X} \to \mmr \text{ st } 
	f_{c,h}(x) = v_{c,h}^\pm(x) k\left(\frac{x-c}{h}\right), ~ c \in \m{X}, 
	~ h \in [\underline{h}_1, \overline{h}_1] 
\right\}$\\*
Take $\m{G}_{5}^{+}$ WLOG.
Define
$\m{G}_{5}^{*} = \Big\{ 
	f = k\left(g \right), ~ g \in \m{G}_1^*
\Big\}$,
where 
$\m{G}_1^*$ is the VC subgraph class of functions defined above.
Given that $k(\cdot)$ is Lipschitz continuous (Assumption \ref{assu_srd_est_kernel}), 
Lemma \ref{lemma_covnum_sumprod} says that $\m{G}_{5}^{*}$ has covering number
$N_1(\eps,Q,\m{G}_5^*)$
bounded by a polynomial in $\eps$ whose constants do not depend on $Q$ or $n$.
Note that $\m{G}_{5} = \Big\{ g h : ~ g \in \m{G}_2^+, ~ h \in \m{G}_{5}^{*}
\Big\}$ because $v_{c,h}^+(x) k \left(v_{c,h}^+(x) \left(\frac{x-c}{h}\right) \right) 
=
v_{c,h}^+(x) k \left( \frac{x-c}{h}  \right)
$.
Therefore, $\m{G}_{5}$ has covering number
bounded above by $A_5 \varepsilon ^ {-W_5}$ with $A_5,W_5$ independent of $n$ or $Q$
(Lemma \ref{lemma_covnum_sumprod}).

\end{proof}

Lemma \ref{lemma_pollard} below is a slightly modified version of \cite{pollard1984convergence}'s Theorem 2.37.

\begin{lemma}\label{lemma_pollard}
For each $n$, let $ \m{F}_n$ be a class of uniformly bounded functions whose covering numbers satisfy
\begin{align*}
\sup_Q N_1(\varepsilon, Q,  \m{F}_n)\leq A \varepsilon^{-W} \text{ for } 0<\varepsilon<1
\end{align*}
with constants $A$ and $W$ not depending on $n$.
Let $\delta_n$ be a positive decreasing sequence
such that $\frac{\log n}{ n \delta_{n}^2} \to 0$. If $[\mme (f^2)]^{1/2} \leq \delta_n$ for $\forall f \in  \m{F}_n$, then

\begin{align*}
\sup_{f \in  \m{F}_{n}} | E_n(f) - \mme(f)| = O_P\left( \delta_n^2 \sqrt{\frac{\log n}{n \delta_n^2 }} \right)
\end{align*}
where $E_n(f)$ is the expected value of $f$ wrt the empirical distribution of the variables in the domain of $f$.
\end{lemma}
\begin{proof}
The proof is almost the same as \cite{pollard1984convergence}'s Theorem 2.37.
There are two main differences.
The first, he has an arbitrary sequence $\alpha_n$ that weakly decreases to zero such that
$\frac{n \delta_n^2 \alpha_n^2}{\log n} \to \infty$, and I take this sequence to be $\alpha_n^2 = \frac{\log n}{n \delta_n ^2} \to 0$. Note that this sequence of $\alpha_n$ does not satisfy
$\frac{n \delta_n^2 \alpha_n^2}{\log n} \to \infty$, but this is not needed here.  The second difference,
he shows almost sure convergence, and I only show the expression to be bounded in probability.

That said, it is to be shown that for $\forall \gamma>0$, there exists $M_{\gamma}>0$ and $n_{\gamma}$ such that
\begin{align*}
\mmp \left\{ \sup_{f \in  \m{F}_{n}}  | E_n(f) - \mme(f) | > M_{\gamma} \delta_n^2 \alpha_n \right\} < \gamma \text{ for } n \geq n_{\gamma}
\end{align*}

Taking $\varepsilon_n=\varepsilon \delta_n^2 \alpha_n$,

\begin{gather*}
\frac{\mmv(E_n(f))}{(4 \varepsilon_n)^2} \leq \frac{\mme (f^2)}{16n \varepsilon^2 \delta_n^4 \alpha_n^2} \\
\leq \frac{M}{16 \varepsilon^2 n \delta_n^2 \alpha_n^2} = \frac{1}{16 \varepsilon \log n}
\end{gather*}

For large $n$, this is smaller than $1/2$, so that Equation (30) on page 31 of \cite{pollard1984convergence} is used to get:

\begin{gather}
\mmp \left\{ \sup_{f \in  \m{F}_{n}}  | E_n(f) - \mme(f) | > 8 \varepsilon \delta_n^2 \alpha_n  \right\} \leq
4 \mmp \left\{ \sup_{f \in  \m{F}_{n}}  | E_n^\circ(f) | > 2 \varepsilon_n  \right\} \label{P2.37_0}
\end{gather}

where $E_n^\circ(f)$ is the signed measure defined there. Using the same approximation argument that led to Equation (31)
on page 31 for functions $g_j \in  \m{F}_n$:
\begin{gather*}
\mmp \left\{ \sup_{f \in  \m{F}_{n}}  | E_n^\circ(f) | > 2 \varepsilon_n  \right\}
\leq 2 N_1(\varepsilon_n, P_n,  \m{F}_n) \exp \left[ \frac{(-1/2) n \varepsilon^2_n}{\max\limits_j E_n (g_j^2)} \right] \\
\end{gather*}
where $P_n$ is the probability measure that weights each observation by $1/n$.
This inequality is used to rewrite the right-hand side of Equation \ref{P2.37_0}):

\begin{align}
4\mmp \left\{ \sup_{f \in  \m{F}_{n}}  | E_n^\circ(f) | > 2 \varepsilon_n  \right\} =&
4\mmp \left\{ \sup_{f \in  \m{F}_{n}}  | E_n^\circ(f) | > 2 \varepsilon_n ,
\sup_{f \in  \m{F}_{n}}  | E_n(f^2) | \leq 64 \delta_n^2 \right\} \nonumber \\
&+4\mmp \left\{ \sup_{f \in  \m{F}_{n}}  | E_n^\circ(f) | > 2 \varepsilon_n,
\sup_{f \in  \m{F}_{n}}  | E_n(f^2) | > 64 \delta_n^2 \right\} \nonumber \\
\leq &~8 N_1(\varepsilon_n, P_n,  \m{F}_n) \exp \left[ \frac{(-1/2) n \varepsilon^2_n}{64 \delta_n^2 } \right] \label{P2.37_1}\\
&+4\mmp \left\{ \sup_{f \in  \m{F}_{n}}  | E_n(f^2) | > 64 \delta_n^2  \label{P2.37_2} \right\}
\end{align}

For part \eqref{P2.37_1}), use the fact that $N_1(\varepsilon_n, P_n,  \m{F}_n) \leq A \varepsilon^{-W}$, and rearrange it into
\begin{gather*}
\ref{P2.37_1}) \leq 8A \varepsilon^{-W} \exp [ W \log (1/{\delta_n^2 \alpha_n}) - n \varepsilon^2 \delta_n^2 \alpha_n^2 /128 ]
\end{gather*}

For part \eqref{P2.37_2}) use Lemmas 33 and 36 in chapter 2 of \cite{pollard1984convergence} to get:
\begin{gather*}
\ref{P2.37_2}) \leq 16 \mme\left[ \min \left\{ N_2(\delta_n,P_n, \m{F}_n) \exp(-n \delta_n^2) ~;~ 1\right\} \right] \\
\leq 16 \mme\left[ \min \left\{ N_1(\delta_n^2/2,P_n, \m{F}_n) \exp(-n \delta_n^2) ~;~ 1\right\} \right] \\
\leq 16  \min \left\{ A \left( \frac{\delta_n^2}{2} \right)^{-W} \exp(-n \delta_n^2) ~;~ 1\right\}  \\
= 16  \min \left\{ A 2^{W} \exp \left[- (W \log \delta_n^2 + n \delta_n^2) \right] ~;~ 1\right\}
\end{gather*}

Hence,

\begin{gather}
\mmp \left\{ \sup_{f \in  \m{F}_{n}}  | E_n(f) - \mme(f)| > 8 \varepsilon \delta_n^2 \alpha_n  \right\} \nonumber \\
< 8A \varepsilon^{-W} \exp [ W \log (1/{\delta_n^2 \alpha_n}) - n \varepsilon^2 \delta_n^2 \alpha_n^2 /128 ]\label{P2.37_A}\\
+16  \min \left\{ A 2^{W} \exp \left[- (W \log \delta_n^2 + n \delta_n^2) \right] ~;~ 1\right\}  \label{P2.37_B}
\end{gather}

For the case here, it suffices to show that there is a $\varepsilon$ such that the sum of the two
bounds above (\ref{P2.37_A}) and (\ref{P2.37_B}) converge to zero as $n \to \infty$. For part (\ref{P2.37_A}),
note that for large $n$, $n \log n \geq n \alpha_n$, since $\alpha_n$ is decreasing. Then,
\begin{gather*}
\log \left( \frac{1}{\delta_n^2 \alpha_n} \right) = \log \left( \frac{n \alpha_n}{\log n} \right) \leq \log \left( \frac{n \log n}{\log n} \right) = \log n
\end{gather*}
Using this, $(\ref{P2.37_A}) \leq 8A \varepsilon^{-W} \exp \left[ \left( W-\frac{\varepsilon^2}{128} \right) \log n \right]$. If $\varepsilon$ is made small enough, this expression goes to zero.

For part \eqref{P2.37_B}, $\frac{\log n }{n \delta_n^2} \to 0$ leads to $\frac{\delta_n^2}{n^{-1}} = n \delta_n^2  \to \infty$. For big $n$, these imply: (i) $\log (\delta_n^2) \geq \log n^{-1} = -\log n$ and (ii) $n \delta_n^2 \geq (W+1) \log n$. Hence,

\begin{gather*}
(\ref{P2.37_B}) \leq 16  \min \left\{ A 2^{W} \exp (- \log n ) ~;~ 1\right\} \to 0
\end{gather*}

\end{proof}

In what follows, I use the Euclidean norm $\| \cdot \|$ with real-valued vectors.
 For matrices, the norm is  induced by the Euclidean norm.
That is, for a $p \times q$ matrix $A$, $\| A \| = \sup_{x \in \mathbb{R}^q, \|x\|=1 } \| Ax \|$.
Such a matrix norm has the following properties:
(i) for a matrix $A$ and a vector $x$, $\| A x \| \leq \| A \| \| x \|$;
(ii) for matrices $A$ and $B$ such that $AB$ is defined,
$\| A B \| \leq \| A \| \| B \|$;
(iii) for $A$ invertible, $\| A  \|^{-1} \leq \| A^{-1} \| $.
The determinant of matrix $A$ is denoted $\det(A)$.
Another useful result is that (iv) convergence
in the matrix norm is equivalent to convergence of all elements of the matrix.

\begin{lemma}\label{lemma_mat_inv}
Consider a random process $X_n(c)$ in $\mathbb{R}^{q \times q}$, and a fixed (non-random)
function $X(c)$ also in $\mathbb{R}^{q \times q}$. Suppose $\sup\limits_{c} \| X(c)\| \leq L_0 < \infty$
and $\inf\limits_{c} | \det(X(c)) | \geq L_1 > 0$.

If for some sequence $\alpha_n \downarrow 0$
\begin{gather*}
\sup\limits_{c} \left\| X_n(c) - X(c) \right\| = O_P (\alpha_n)
\end{gather*}
then
\begin{gather*}
\sup\limits_{c} \left\| X_n(c)^{-1} - X(c)^{-1} \right\| = O_P (\alpha_n)
\end{gather*}
\end{lemma}

\begin{proof}
Consider the compact subset of $\mathbb{R}^{q \times q}$:
\[
A= \left\{X \in \mathbb{R}^{q \times q} :~  \|X \| \leq 2 L_0,~ |\det(X)| \geq L_1 \right\}
\]
Note that $X(c) \in A ~~\forall (c)$, and that any continuous function on $A$ is uniformly continuous
because $A$ is a compact set.
The function $f: A \to \mathbb{R}^{q \times q}$, $f(X)=X^{-1}$ is uniformly continuous.

For any $\gamma>0$, find $M_{\gamma}>0$ such that
\[
\mmp \left\{ \sup\limits_{c}
\alpha_n^{-1} \left\|  X_n(c)^{-1} - X(c)^{-1} \right\| >   M_{\gamma}  \right\}
< \gamma
\]

\begin{gather}
\mmp \left\{  \alpha_n^{-1} \sup\limits_{c}
\left\|  X_n(c)^{-1} - X(c)^{-1} \right\| > M_{\gamma}  \right\}
\nonumber
\\
\leq
\mmp \left\{  \sup\limits_{c}
\left\|  X_n(c)^{-1} - X(c)^{-1} \right\| > \alpha_n M_{\gamma},~ X_n(c) \in A ~\forall c  \right\}
\label{eq_mat_inv_1}
\\
+
\mmp \left\{ X_n(c) \notin A \text{ for some } c  \right\}
\label{eq_mat_inv_2}
\end{gather}

\textbf{Part \eqref{eq_mat_inv_1}}

Since $f(X)=X^{-1}$ is uniformly continuous in $A$, for any choice of $M_{\gamma}>0$, and for
a given sample size, there exists a
$\delta(\alpha_n M_{\gamma})>0$ such that

\[
\forall X_n(t), X(t) \in A, ~  \|X_n(t)^{-1} - X(t)^{-1} \| > \alpha_n M_\gamma
\Rightarrow \|X_n(t) - X(t) \| >  \delta( \alpha_n M_{\gamma})
~~ \forall n
\]


\begin{gather}
(\ref{eq_mat_inv_1}) \leq
\mmp \left\{ \sup\limits_{c,p}
\left\|  X_n(c) - X(c) \right\| > \delta(\alpha_n M_{\gamma}),~ X_n(c) \in A ~\forall (c)  \right\}
\nonumber
\\
\leq
\mmp \left\{ \sup\limits_{c,p}
\left\|  X_n(c) - X(c) \right\| > \delta(\alpha_n M_{\gamma})  \right\}
\label{eq_mat_inv_1A}
\end{gather}

By assumption, it is possible to find $M^*$ such that

\[
\mmp \left\{ \sup\limits_{c,p} \left\| X_n(c) - X(c) \right\| > \alpha_n M^* \right\} < \gamma/2
\]
for large $n$.
So pick $M_\gamma$ to be such that $\delta( \alpha_n M_{\gamma}) \geq \alpha_n M^*$ which makes
$ (\ref{eq_mat_inv_1A}) \leq \gamma/2$.

\textbf{Part \eqref{eq_mat_inv_2}}

\begin{gather*}
(\ref{eq_mat_inv_2}) \leq
\mmp \left\{ \| X_n(c) \|  > 2 L_0 \text{ for some } c \right\}
\\
+
\mmp \left\{ | \det (X_n(c)) |  <  L_1 \text{ for some } c \right\}
\\
\leq
\mmp \left\{ \| X_n(c) - X(c) \|  > L_0 \text{ for some } c \right\}
\\
+ \mmp \left\{ \| X(c) \|  > L_0 \text{ for some } c \right\}
\\
+
\mmp \left\{ | \det (X_n(c)) - \det (X(c)) |  <  L_1/2 \text{ for some } c \right\}
\\
+
\mmp \left\{ | \det (X(c)) |  <  L_1/2 \text{ for some } c \right\}
\\
=
\mmp \left\{ \| X_n(c) - X(c) \|  > L_0 \text{ for some } c \right\}
\\
+
\mmp \left\{ | \det (X_n(c)) - \det (X(c)) |  <  L_1/2 \text{ for some } c \right\}
\end{gather*}
which is made smaller than $\gamma/2$ for large $n$ since $X_n(c)$ converges in probability
to $X(c)$, and so does  $\det( X_n(c))$ to $\det( X(c))$.
Therefore, $(\ref{eq_mat_inv_1})+(\ref{eq_mat_inv_2}) \leq \gamma$.

\end{proof}

An application of Lemma \ref{lemma_pollard} to the classes of functions in
Lemma \ref{lemma_classf} gives the rates at which certain terms
in the proof of Theorem \ref{theo_srd_kinf_est_int} are uniformly bounded
in probability.

\begin{lemma}\label{lemma_app}
Consider the definitions of $G^{j \pm}$, $G_n^{j \pm}$, $\ti{H}_i^{j }$, $H$, and $Y_i^{j \pm}$
from Lemma \ref{lemma_porter} (scalar case).
Suppose Assumptions
\ref{assu_srd_est_kernel},
\ref{assu_srd_est_fx},
\ref{assu_srd_est_mx}, and 
\ref{assu_srd_kinf_est_int} hold.
Assume the rate conditions of Theorem \ref{theo_srd_kinf_est_int}.
Then,

\begin{align}
&\max\limits_{j} \left\| G^{j \pm} - \mme\left[ G_n^{j \pm} \right] \right| = O(\overline{h}_1)
\label{eq:lemma_app:1}
\\
&\max\limits_{j} \left\|
	G_n^{j \pm} - \mme\left[ G_n^{j \pm} \right]
\right\|
= O_P\left( \sqrt{\frac{\log n}{n \overline{h}_1}} \right)
\label{eq:lemma_app:2}
\\
&\max \limits_{j } \left\|
	{\frac{1}{n h_{1j}}\sum_{i=1}^n v_i^{j\pm}
		k\left(\frac{X_i - c_j}{h_{1j}}\right) \ti{H}_i^{j }  \varepsilon_i}
\right\| = O_P\left( \sqrt{\frac{\log n}{n \overline{h}_1}} \right)
\label{eq:lemma_app:3}
\\
&\max \limits_{j }  \left\|
	\frac{1}{n h_{1j} } \sum_{i=1}^n \left\{
		v_i^{j\pm} 
		k\left(\frac{X_i - c_j}{ h_{1j} }\right) \ti{H}_i^{j } \mme[Y_i^{j \pm} | X_i]
	\right.
\right.
\notag
\\
&\hspace{3.5cm} \left.
	\left.
	 -\mme\left[ v_i^{j\pm}  k\left(\frac{X_i - c_j}{ h_{1j} }\right) \ti{H}_i^{j} Y_i^{j \pm} \right]
	~ \right\}
~ \right\|= O_P\left( \sqrt{\frac{\log n}{n \overline{h}_1}} \right)
\label{eq:lemma_app:4}
\end{align}

\end{lemma}

\begin{proof}
Consider the positive parts with $v_i^{j+} = v_{c_j,h_{1j}}^+(X_i)$,  $G_n^{j +}$, and $G^{j +}$
WLOG.

\bigskip
\textbf{Part \eqref{eq:lemma_app:1}}

First, we show that  
$\max\limits_{j} \left\|  \mme\left[ G_n^{j +} \right] - G^{j +}   \right\| = O(\overline{h}_1)$ using Lemma \ref{lemma_mat_inv}.

We have that $G^{j +} = f(c_{j})^{-1} \Gamma^{-1}$ is a bounded function of $j$ and has a determinant uniformly bounded away from zero. Using Lemma \ref{lemma_mat_inv}, it suffices to show

$\max\limits_{j} \left\|   \left[ \mme\left( G_n^{j +} \right) \right]^{-1} - \left( G^{j +} \right)^{-1}   \right\| = O(\overline{h}_1)$.

Also, convergence in
the matrix norm is equivalent to convergence in each element of the matrix. Hence, it suffices to show that

$\max\limits_{j} \left| 
	\frac{1}{h_{1j}} \mme\left[ 
		v_i^{j+}  k\left(\frac{X_i - c_j}{h_{1j}}\right) \left(\frac{X_i - c_j}{h_{1j}}\right)^l
	\right]
	- { f(c_j) \gamma_l }  
\right| = O(\overline{h}_1)$.

The LHS above is bounded by

$
\sup \limits_{c \in \m{X}, h \in [\underline{h}_1, \overline{h}_1]}
\left| \frac{1}{h} \mme\left[ v_{c,h}(X_i)  k\left(\frac{X_i - c}{h}\right) \left(\frac{X_i - c}{h}\right)^l \right]
- { f(c) \gamma_l }  \right| 
$,

\noindent
which we show to be $O(\overline{h}_1)$.

Take an arbitrary sequence $h \in [\underline{h}_1, \overline{h}_1]$,

$
\left| \frac{1}{h} \mme\left[ v_{c,h}(X_i)  k\left(\frac{X_i - c}{h}\right) \left(\frac{X_i - c}{h}\right)^l \right]
- {f(c)  \gamma_l}  \right|
$

$
=  \left| \int_0^1 k\left( u \right) u^l f(c+u h) du  - {f(c) \gamma_l } \right| 
$

$
=  h \left| \int_0^1 k\left( u \right) u^{l+1} \nabla_x f(c_{u h}^*) du  \right| \leq M h  = O(\overline{h}_1)
$

where Assumption \ref{assu_srd_est_fx} bounds the derivative of $f$. Therefore,
the supremum above is $O(\overline{h}_1)$, and the result follows.

\bigskip
\textbf{Part \eqref{eq:lemma_app:2}}

The goal is to show that:

$\max\limits_{j}
\left\| G_n^{j \pm} - \mme\left[ G_n^{j \pm} \right] \right\|
= O_P\left( \sqrt{\frac{\log n}{n \overline{h}_1}} \right)$.

Note that part \eqref{eq:lemma_app:1} implies that
$\mme\left[ G_n^{j +} \right]$ is a bounded function of $j$ and has a determinant uniformly bounded away from zero for large $n$. Using Lemma
\ref{lemma_mat_inv}, it suffices to show that
$
\max\limits_{j}
\left\| {G_n^{j +}}^{-1} - {\mme\left[ G_n^{j +} \right]}^{-1} \right\|
=O_P\left( \sqrt{\frac{\log n}{n \overline{h}_1}} \right)
$. In fact, it suffices to show uniform convergence of each point 
of the matrix:
\begin{gather*}
\max\limits_{j}
\left|
\frac{1}{nh_{1j}} \sum_{i=1}^n
\left\{~ 
v_i^{j +}
 k\left( \frac{X_i - c_{j}}{h_{1j}} \right)
\left( \frac{X_i - c_{j}}{h_{1j}} \right)^l
-
\mme\left[ v_i^{j +}  k\left(\frac{X_i - c}{h_{1j}}\right) \left(\frac{X_i - c}{h_{1j}}\right)^l \right]~ 
\right\}~ 
\right|
\\*
= O_P\left( \sqrt{\frac{\log n}{n \overline{h}_1}} \right)
\end{gather*}
for an arbitrary $l$. The LHS is bounded by
\begin{gather}
\mathsmaller{
\frac{1}{\underline{h}_1}
	\sup \limits_{c \in \m{X}, ~ h \in [\underline{h}_1, \overline{h}_1] }
	\left| \frac{1}{n }\sum_{i=1}^n 
		\left\{ v_{c,h}^+(X_i)  k\left(\frac{X_i - c}{h}\right) \left(\frac{X_i - c}{h}\right)^l
			- \mme\left[ v_{c,h}^+(X_i)  k\left(\frac{X_i - c}{h}\right) \left(\frac{X_i - c}{h}\right)^l \right] ~ 
		\right\}~
	\right|
}
\label{lemma_app_p1_1}
\end{gather}
and we apply Lemma \ref{lemma_pollard} to this part.
Lemma \ref{lemma_classf} says that the class of functions (over which the sup is being taken)
satisfies the conditions of Lemma \ref{lemma_pollard}.
For the second moment bound $\delta_n^2$, take an arbitrary sequence 
$h \in [\underline{h}_1, \overline{h}_1]$ and note that
\begin{gather*}
\mme \left\{
\left[ v_{c,h}^+(X_i)
k\left(\frac{X_i - c}{h}\right) \left(\frac{X_i - c}{h}\right)^l  \right]^2
\right\}
= h \int_0^1 k(u)^2 u^{2l} f(c+u h) du \leq M h  \leq M \overline{h}_1
\end{gather*}
where $f(\cdot)$ and $k(\cdot)$ are uniformly bounded (Assumptions \ref{assu_srd_est_kernel} and \ref{assu_srd_est_fx}).
Hence, for the purposes of Lemma \ref{lemma_pollard}, $\delta_n^2=M \overline{h}_1$, which satisfies $\frac{\log n}{n \delta_n^2} \to 0$
because of the rate condition $\frac{\sqrt{K} \log n}{\sqrt{n \overline{h}_1}} \to 0$.
Therefore, applying  Lemma \ref{lemma_pollard} to Equation \ref{lemma_app_p1_1}
makes it 
$\underline{h}_1^{-1} O_P\left(\overline{h}_1 \sqrt{\frac{\log n}{n \overline{h}_1 }}\right) =
O_P\left( \sqrt{\frac{\log n}{n \overline{h}_1 }}\right) $, because $\underline{h}_1^{-1} \overline{h}_1 =O(1)$.

\bigskip

\textbf{Part \eqref{eq:lemma_app:3}}

Similar to above, convergence in
the matrix norm is equivalent to convergence in each element of the matrix,
so it suffices
to show that

\[
\frac{1}{\underline{h}_1}
\sup \limits_{c \in \m{X}, h \in [\underline{h}_1, \overline{h}_1]}
	\left| {\frac{1}{n}\sum_{i=1}^n v_{c,h}^+(X_i)
			k\left(\frac{X_i - c}{h}\right) \left(\frac{X_i - c}{h}\right)^l \varepsilon_i}
	\right|
	= O_P\left(\sqrt{\frac{\log n}{n \overline{h}_1}} \right)
\]
for any positive integer $l$.
Take an arbitrary sequence 
$h \in [\underline{h}_1, \overline{h}_1]$
\begin{gather*}
\mme \left\{
\left[ v_{c,h}^+(X_i)  k\left(\frac{X_i - c}{h}\right) \left(\frac{X_i - c}{h}\right)^l
\varepsilon_i \right]^2
\right\}
\\
\leq
M \mme\left[ v_{c,h}^+(X_i)  k\left(\frac{X_i - c}{h}\right)^2 \left(\frac{X_i - c}{h}\right)^{2l} \right]
\\
= M h \int_0^1 k(u)^2 u^{2l} f(c+u h)  du \leq M \overline{h}_1
\end{gather*}
where it is used that $\varepsilon_i=Y_i-R(X_i,D_i)$ is a.s. uniformly bounded (Assumption \ref{assu_srd_kinf_est_int}).
Hence,  $\delta_n^2= M  \overline{h}_{1}$.
The expectation
$\mme\left[ v_{c,h}^+(X_i)  k\left(\frac{X_i - c}{h}\right)
\left(\frac{X_i - c}{h}\right)^l \varepsilon_i \right] = 0$, and the sup is over a
class of functions that satisfies the conditions of Lemma \ref{lemma_pollard}, which  gives the result.

\bigskip

\textbf{Part \eqref{eq:lemma_app:4}}

It suffices to show that
\begin{align*}
&\frac{1}{ \underline{h}_1} \sup \limits_{c \in \m{X}, ~ h \in [\underline{h}_1, \overline{h}_1] }
	\left|
		\frac{1}{n } \sum_{i=1}^n \left\{
			v_{c, h}^+ (X_i)
			k\left(\frac{X_i - c}{h}\right) \left(\frac{X_i - c}{h}\right)^l \mme[Y_i^{j +} | X_i]
		\right.
	\right.
\\
&	\left.
		\left.
			\hspace{4cm} - \mme\left[
				v_{c,h}^+(X_i)
				k\left(\frac{X_i - c}{h}\right) \left(\frac{X_i - c}{h}\right)^l Y_i^{j \pm}
		\right]~
	\right\}~
\right|
\\
=&
\frac{1}{ \underline{h}_1} \sup \limits_{c \in \m{X}, ~ h \in [\underline{h}_1, \overline{h}_1] }
\left|
	\frac{1}{n } \sum_{i=1}^n \left\{
		v_{c,h}^+ (X_i)
		k\left(\frac{X_i - c}{h}\right) \left(\frac{X_i - c}{h}\right)^l r^+(X_i)
	\right.
\right.
\\
&\left.
	\left.
		\hspace{4cm} - \mme\left[
			v_{c,h}^+(X_i)
			k\left(\frac{X_i - c}{h}\right) \left(\frac{X_i - c}{h}\right)^l r^+(X_i)
		\right]~
	\right\}~
\right|
\\
= & O_P\left( \sqrt{\frac{\log n}{n \overline{h}_1}} \right)
\end{align*}
for any positive integer $l$. Choose $\delta_n^2$ similarly as before. The sup is over a
class of functions that satisfies the conditions of Lemma \ref{lemma_pollard}, which gives the result.

\end{proof}

\bigskip

\begin{lemma}
\label{lemma_rest_theo}
Assume the conditions of Theorem \ref{theo_srd_kinf_est_int} hold.
Then, parts \eqref{eq:kinf:op1:a} and \eqref{eq:kinf:op1:b} in the proof of Theorem \ref{theo_srd_kinf_est_int}
(Section \ref{sec:proof:theo_srd_kinf_est_int}) converge in probability to zero.
\end{lemma}
\begin{proof}

 \indent
 
\begin{center}
\textbf{\underline{Part \eqref{eq:kinf:op1:a} }} 
\end{center}

\begin{align}
& \left| \frac{ \ha{\mu} - \mme[\ha{\mu} |\m{X}_n] - \left( \mu^* - \mme[\mu^* |\m{X}_n] \right) }{(\m{V}_n^c)^{1/2}} \right|
\\*
\leq & O\left( \left( K n \overline{h}_1 \right)^{1/2} \right)
\Bigg| \sum_{j=1}^K 
	\Delta_j \Big\{
		e_1' \left( G_n^{j +}  - \mme\left[ G_n^{j +} \right] \right)
		\frac{1}{nh_{1j}}
		\sum_{i=1}^n
			k\left( \frac{X_i - c_{j}}{h_{1j}} \right)
			v_i^{j +} \eps_i
			\ti{H}_{i}^{j}
\notag
\\*
& \hspace{4cm}
		-e_1' \left( G_n^{j -}  - \mme\left[ G_n^{j -} \right] \right)
		\frac{1}{nh_{1j}}
		\sum_{i=1}^n
			k\left( \frac{X_i - c_{j}}{h_{1j}} \right)
			v_i^{j -} \eps_i
			\ti{H}_{i}^{j}
	\Big\}
\Bigg|
\\
\leq & O\left( \left( K n \overline{h}_1 \right)^{1/2} \right)
\sum_{j=1}^K 
	\left| \Delta_j \right| 
	\left\|
		G_n^{j +}  - \mme\left[ G_n^{j +} \right] 
	\right\|
	\left\|
		\frac{1}{nh_{1j}}
		\sum_{i=1}^n
			k\left( \frac{X_i - c_{j}}{h_{1j}} \right)
			v_i^{j +} \eps_i
			\ti{H}_{i}^{j}
	\right\|
\notag
\\
& \hspace{4cm}
	+\left\|
		G_n^{j -}  - \mme\left[ G_n^{j -} \right] 
	\right\|
	\left\|
		\frac{1}{nh_{1j}}
		\sum_{i=1}^n
			k\left( \frac{X_i - c_{j}}{h_{1j}} \right)
			v_i^{j -} \eps_i
			\ti{H}_{i}^{j}
	\right\|
\\
 \leq & O\left( \left( K n \overline{h}_1 \right)^{1/2} \right)
 K
O\left( K^{-1} \right)
O_P\left( \left(\frac{ \log n }{ n \overline{h}_1 }\right)^{1/2} \right)
O_P\left( \left(\frac{ \log n }{ n \overline{h}_1 }\right)^{1/2} \right)
 \\
 = &
 O_P \left( K^{1/2} \frac{ \log n }{ \left( n \overline{h}_1 \right)^{1/2} }  \right) = o_P(1)
\end{align} 
 where the first inequality uses the rate on $(\m{V}_n^c)^{-1/2}$ (Equation \ref{eq:kinf:clt:s2:rate}); 
the third inequality relies on the uniform convergence rates of Lemma \ref{lemma_app}, and that $\Delta_j = O(K^{-1})$ uniformly over $j$ 
(Lemma \ref{lemma_rate_int});
the last equality uses the  rate condition ${K}^{1/2} \log n \left( n \overline{h}_1 \right)^{-1/2} =o(1)$.

\begin{center}
\textbf{\underline{Part \eqref{eq:kinf:op1:b} }} 
\end{center} 

\begin{align}
\frac{ \mme[\ha{\mu} - \mu_n |\m{X}_n] - \ti{\mu}  }{(\m{V}_n^c)^{1/2}} &
\\
= & (\m{V}_n^c)^{-1/2} \sum_{j=1}^K 
	\Delta_j 
	e_1'  G_n^{j +} \left\{ 
		\frac{1}{nh_{1j}}
		\sum_{i=1}^n
			k\left( \frac{X_i - c_{j}}{h_{1j}} \right)
			v_i^{j +} \mme[Y_i^{j+} | X_i ] 
			\ti{H}_{i}^{j}
\right. \notag \\*
& \hspace{3cm} \left.
	-
	\mme\left[  
		\frac{1}{nh_{1j}}
		\sum_{i=1}^n
			k\left( \frac{X_i - c_{j}}{h_{1j}} \right)
			v_i^{j +} Y_i^{j+} 
			\ti{H}_{i}^{j}
	\right]
\right\}
\label{eq:kinf:op1:b:a}
\\
- & (\m{V}_n^c)^{-1/2} \sum_{j=1}^K 
	\Delta_j 
	e_1'  G_n^{j -} \left\{ 
		\frac{1}{nh_{1j}}
		\sum_{i=1}^n
			k\left( \frac{X_i - c_{j}}{h_{1j}} \right)
			v_i^{j -} \mme[Y_i^{j-} | X_i ] 
			\ti{H}_{i}^{j}
\right. \notag \\*
& \hspace{3cm} \left.
	-
	\mme\left[  
		\frac{1}{nh_{1j}}
		\sum_{i=1}^n
			k\left( \frac{X_i - c_{j}}{h_{1j}} \right)
			v_i^{j -} Y_i^{j-} 
			\ti{H}_{i}^{j}
	\right]
\right\}
\label{eq:kinf:op1:b:b}
\end{align}

where

\begin{align}
(\ref{eq:kinf:op1:b:a}) 
= & (\m{V}_n^c)^{-1/2} \sum_{j=1}^K 
	\Delta_j 
	e_1'  \mme\left[ G_n^{j +} \right] \left\{ 
		\frac{1}{nh_{1j}}
		\sum_{i=1}^n
			k\left( \frac{X_i - c_{j}}{h_{1j}} \right)
			v_i^{j +} \mme[Y_i^{j+} | X_i ] 
			\ti{H}_{i}^{j}
\right. \notag \\*
& \hspace{3.7cm} \left.
	-
	\mme\left[  
		\frac{1}{nh_{1j}}
		\sum_{i=1}^n
			k\left( \frac{X_i - c_{j}}{h_{1j}} \right)
			v_i^{j +} Y_i^{j+} 
			\ti{H}_{i}^{j}
	\right]
\right\}
\label{eq:kinf:op1:b:a:a}
\\
= & (\m{V}_n^c)^{-1/2} \sum_{j=1}^K 
	\Delta_j 
	e_1'  \left( G_n^{j +} - \mme\left[ G_n^{j +} \right]    \right) \left\{ 
		\frac{1}{nh_{1j}}
		\sum_{i=1}^n
			k\left( \frac{X_i - c_{j}}{h_{1j}} \right)
			v_i^{j +} \mme[Y_i^{j+} | X_i ] 
			\ti{H}_{i}^{j}
\right. \notag \\*
& \hspace{5.2cm} \left.
	-
	\mme\left[  
		\frac{1}{nh_{1j}}
		\sum_{i=1}^n
			k\left( \frac{X_i - c_{j}}{h_{1j}} \right)
			v_i^{j +} Y_i^{j+} 
			\ti{H}_{i}^{j}
	\right]
\right\}
\label{eq:kinf:op1:b:a:b}
\end{align}

Part \eqref{eq:kinf:op1:b:a:a} is $o_P(1)$ because it has zero mean and zero limiting variance,
\begin{align}
\mmv[(\ref{eq:kinf:op1:b:a:a})]
= & (\m{V}_n^c)^{-1} \sum_{j=1}^K 
	\Delta_j^2 
	\frac{1}{n^2 }
	n \mmv \left\{ 
		\frac{1}{h_{1j}} k\left( \frac{X_i - c_{j}}{h_{1j}} \right)
		v_i^{j +} \mme[Y_i^{j+} | X_i ] 
		\left( e_1'  \mme\left[ G_n^{j +} \right] \ti{H}_{i}^{j} \right)
\right. \notag \\*
& \hspace{3.7cm} \left.
		-
		\mme\left[  
			\frac{1}{h_{1j}}
			k\left( \frac{X_i - c_{j}}{h_{1j}} \right)
			v_i^{j +} Y_i^{j+} 
			\left( e_1'  \mme\left[ G_n^{j +} \right] \ti{H}_{i}^{j} \right)
		\right]
\right\}
\\
\leq  & O(K n \overline{h}_1) \sum_{j=1}^K 
	\Delta_j^2 
	\frac{1}{n}
	\mme\left[ 
		\frac{1}{h_{1j}^2} k\left( \frac{X_i - c_{j}}{h_{1j}} \right)^2
		v_i^{j +} \mme[Y_i^{j+} | X_i ]^2 
		\left( e_1'  \mme\left[ G_n^{j +} \right] \ti{H}_{i}^{j} \right)^2
	\right]
\\
= & O(K ) \sum_{j=1}^K 
	O(K^{-2} ) 
	\mme\left[ 
		\frac{1}{h_{1j}} k\left( \frac{X_i - c_{j}}{h_{1j}} \right)^2
		v_i^{j +}  O \left( h_{1j}^{\rho_1 +  1}  \right)^2 
		\left( e_1'  \left[ G^{j +} + O(h_{1j}) \right] \ti{H}_{i}^{j} \right)^2
	\right]
	\label{eq:kinf:op1:b:a:a:expectbounded}
\\
= &O \left( \overline{h}_{1}^{2 \rho_1 +  2}  \right) 	= o(1)
\end{align}
where it is used the rate on $(\m{V}_n^c)^{-1}$ (Equation \ref{eq:kinf:clt:s2:rate});
that $\Delta_j^2 = O(K^{-2})$ holds uniformly over $j$ (Lemma \ref{lemma_rate_int});
expansion \eqref{eq:kinf:bias:expansion}; 
$\mme\left[ G_n^{j +} \right] $ is uniformly close to $G^{j +}$ (Lemma \ref{lemma_app});
and that the expected value in \eqref{eq:kinf:op1:b:a:a:expectbounded} without the $O\left( h_{1j}^{\rho_1 +  1}  \right)^2$ term 
is a bounded quantity.

Part \eqref{eq:kinf:op1:b:a:b} is $o_P(1)$ because 
\begin{align}
\left| \left( \ref{eq:kinf:op1:b:a:b} \right) \right| \leq  & 
O\left( \left(K n \overline{h}_1 \right)^{1/2} \right) \sum_{j=1}^K 
	| \Delta_j |
	\left\|  G_n^{j +} - \mme\left[ G_n^{j +} \right]    \right\| 
	\left\| 
		\frac{1}{nh_{1j}}
		\sum_{i=1}^n
			k\left( \frac{X_i - c_{j}}{h_{1j}} \right)
			v_i^{j +} \mme[Y_i^{j+} | X_i ] 
			\ti{H}_{i}^{j}
\right. \notag \\*
& \hspace{8cm} \left.
	-
	\mme\left[  
		\frac{1}{h_{1j}}
		k\left( \frac{X_i - c_{j}}{h_{1j}} \right)
		v_i^{j +} Y_i^{j+} 
		\ti{H}_{i}^{j}
	\right]
\right\|
\\
= & O\left( \left(K n \overline{h}_1 \right)^{1/2} \right)
K
O\left( K^{-1} \right)
O_P\left( \left( \log n \right)^{1/2} \left( n \overline{h}_1\right)^{-1/2} \right)
O_P\left( \left( \log n \right)^{1/2} \left( n \overline{h}_1\right)^{-1/2} \right)
\\
=& 
O_P\left( K^{1/2} \left( \log n \right) \left( n \overline{h}_1 \right)^{-1/2}  \right)
=o_P(1)
\end{align}
which relies on the rate conditions of $(\m{V}_n^c)^{-1/2}$ (Equation \ref{eq:kinf:clt:s2:rate}),
that $\Delta_j = O(K^{-1})$ uniformly over $j$ (Lemma \ref{lemma_rate_int}),
and on the rate conditions of Lemma \ref{lemma_app}'s  parts \eqref{eq:lemma_app:2} and \eqref{eq:lemma_app:4}.
Therefore, (\ref{eq:kinf:op1:b:a}) is $o_P(1)$, and a symmetric proof shows that
(\ref{eq:kinf:op1:b:b}) is $o_P(1)$. Hence, part \eqref{eq:kinf:op1:b} is $o_P(1)$.

\end{proof}

\subsection{Integral Approximation}
\label{sec_int_app}
\indent

This section proves results on the error of approximated integrals. Let
$R:\mmr^2 \to \mmr$ be a Riemann integrable function; for an open and convex
 set $\m{C} \subset \mmr^3$, define
$\beta: \m{C} \to \mmr$ such that
$\beta(\bo{x}) = R(x_1,x_3) - R(x_1,x_2)$
(i.e. treatment effect function on the main text).
There are observations of the value of the $\beta(.)$ function for $K$ points $\bc_1,\ldots,\bc_K$, that is,
$\beta_1=\beta(\bc_1), \ldots, \beta_K=\beta(\bc_K)$,
for $\bc_j=(c_{1,j},c_{2,j},c_{3,j})$. Interest lies on the integral $\mu = \int_{\m{C}}\beta(\bx) d(\bx)$
which is approximated by a finite weighted sum $\widehat \mu =  \sum_j \Delta_j \beta_j$. A procedure to compute
the integral approximation is given below. More importantly, there is a result that gives the rate of decay of the approximation error of this procedure as the number of points $K \to \infty$.
The procedure consists of using a multivariate local polynomial regression in a first step to obtain an
approximated function $\widehat \beta(\bx)$. The second step integrates $\widehat \beta(\bx)$
over set $\m{C}$ to obtain an approximated integral $\widehat \mu$.

For the first step, run a weighted regression of $\beta_j$s on $J \times 1$ vectors $E_j(\bx)$.
Each $E_j(\bx)$ is made of polynomials evaluated at $(\bx - \bc_j)$ of order $\rho_2$
at most.
To define $E_{j}(\bx)$ and $J$, first consider the multi-index notation for vectors:
for 
$\bx=(x_1,x_2,x_3) \in \mathbb{R}^3$ and
$\bgam=(\gamma_1,\gamma_2,\gamma_3) \in \mathbb{Z}^3_{+}$, let
\begin{gather*}
|\bgam|=\sum_{i=1}^{3} \gamma_i
\\
\bgam! = \prod_{i=1}^3 {\gamma_i} !
\\
\bx^{\bgam} = \prod_{i=1}^3 x_i^{\gamma_i}
\\
\nabla^{|\bgam|} \beta(\bx) =
  \frac{\partial^{|\bgam|}}{\partial {x_{1}}^{\gamma_1} \partial {x_{2}}^{\gamma_2} \partial {x_{3}}^{\gamma_3}}
  \beta(\bx)
\end{gather*}
Each entry in $E_{j}(\bx)$ is a polynomial of the form $p_{\bgam}(\bx-\bc_j)=\prod_{i=1}^3 (x_{i}-c_{i,j})^{\gamma_i} $
with $\bgam$ such that
$|\bgam| \leq \rho_2$
and
$\min \{\gamma_2, \gamma_3 \} = 0$. There is no  $\gamma_2, \gamma_3 >0$
because $\beta(\bx)$ is the difference $R(x_1,x_3)-R(x_1,x_2)$
whose polynomial approximation does not include interactions between $x_2$ and $x_3$.
The dimension of $E_j(\bx)$
is $J \times 1$ with $J=2 {\rho_2 +2 \choose 2} - (\rho_2+1)$, and the first entry in $E_{j}(\bx)$
is the polynomial of degree zero (i.e. $p_{\bo{0}}(\bx-\bc_j)=1$).
Next, stack $E_{1}(\bx)' \ldots E_{K}(\bx)'$
into the $K \times J$ matrix $\bE(\bx)$, and
$\beta_1, \ldots, \beta_K$ into the $K \times 1$ vector $\bB$.
The regression of $\bB$ on $\bE$ is kernel weighted depending on the distance
between a fixed point $\bx \in \m{C}$ and $\bc_j$.
For a choice of bandwidth $h_2>0$,
and a kernel density function that satisfies Assumption \ref{assu_srd_est_kernel},
the $K \times K$ matrix $\bOmg(\bx;h_2)$ is the diagonal matrix of kernel weights:

\[
\bOmg(\bx;{h_2}) = \diag \left\{ \Omega_{j}(\bx;h_2) \right\}_{j}
= \diag \left\{
\prod_{i=1}^3
k \left( \frac{x_{i} - c_{i,j}}{h_2 } \right)
\right\}_{j}
\]
The first-step regression consists of solving the following problem.
\begin{gather*}
\ha{ \bbeta}
= \argmin\limits_{\bbeta} \left(  \bB - \bE(\bx) \bbeta \right)' \bOmg(\bx;h_2) \left( \bB - \bE(\bx) \bbeta \right)
\\
\widehat \beta(\bx) = e_1' \widehat \bbeta = \widehat\eta_1
\end{gather*}
where $\bbeta$ is a $J \times 1$ vector of parameters, and
$\eta_1$ is the first coordinate of the vector $\bbeta$ (intercept coefficient).

In the second step, integrate the estimated function $\widehat \beta(\bx)$
over $\m{C}$. Note that the approximated integral $\widehat \mu$ is written as a weighted sum of $\beta_j$.
\begin{gather*}
 \int_{\m{C}} \widehat \beta(\bx) ~~ d\bx =
\bigintsss_{\m{C}} ~ e_1'
 \left( \bE(\bx)' \bOmg(\bx;h_2) \bE(\bx) \right)^{-1}
 \sum_j \Omega_{j}(\bx;h_2) E_{j}(\bx) \beta_j ~~d(\bx)
 \\
 =
 \sum_j
 \bigintsss_{\m{C}} ~ e_1'
 \left( \bE(\bx)' \bOmg(\bx;h_2) \bE(\bx) \right)^{-1}
 \Omega_{j}(\bx;h_2) E_{j}(\bx) ~~d(\bx) ~~ \beta_j
 \\
 = \sum_j \Delta_j \beta_j
\end{gather*}
The expression for the correction weight $\Delta_j$ is
\begin{gather*}
 \Delta_j = \bigintsss_{\m{C}} ~ e_1'
 \left( \bE(\bx)' \bOmg(\bx;h_2) \bE(\bx) \right)^{-1}
 \Omega_{j}(\bx;h_2) E_{j}(\bx) ~~d(\bx)
 \\
 = \bigintsss_{\m{C}} ~
 \frac{det \left( \bE(\bx)' \bOmg(\bx;h_2) \bE_{\bo{0} \leftarrow e_j}(\bx) \right)}
   {det \left( \bE(\bx)' \bOmg(\bx;h_2) \bE(\bx) \right)}
   ~~d(\bx)
\end{gather*}
where the Cramer rule is used in the second equality, and $\bE_{\bo{0} \leftarrow e_j}(\bx)$
is the matrix valued function $\bE(\bx)$ except for the first column which is replaced by the $K \times 1 $ vector $e_j$
that is zero everywhere except for the $j$-th entry which is equal to 1.

The approximation error of such a procedure is well-behaved if $R(x,y)$
is a continuously differentiable function of order up to $\rho_2+1$ on  $\m{C}$.
This implies that,
$\nabla^{|\bgam|} \beta(\bx)$
is a continuous function for every $\bgam$ such that $|\bgam|=\rho_2+1$.
Lemma
\ref{lemma_mls} below states the approximation error of using a multivariate local polynomial regression
on a finite number of points to obtain $\widehat \beta(\bx)$.  This result is Theorem 3.1 of \cite{lipman2006},
and here account is given to the fact that $\beta$ is the difference of two functions.

\begin{lemma}\label{lemma_mls}
Let $\m{C} \subset \mmr^3$ be open and convex. Let $R: \mmr^2 \to \mmr$ be a $\rho_2+1$ times continuously differentiable function on $\m{C}$,
and define $\beta(\bx)=R(x_1,x_3)-R(x_1,x_2)$.
 For $\bx \in \m{C}$, assume $\widehat \beta(\bx)$ is constructed as above, and that the matrix
 $\bE(\bx)' \bOmg(\bx;h_2) \bE(\bx)$ is invertible for some choice of $h_2>0$. Then, there exists
 $\xi_{j} \in (0,1)$ $j=1,\ldots,K$, such that

\begin{align}
   \widehat \beta(\bx) - \beta(\bx) & =
   \sum\limits_{
   \substack{
    |\bgam|=\rho_2+1
    \\
    \min \{\gamma_2,\gamma_3 \}=0}
    }
    \sum\limits_{ j=1 }^{K}  \Bigg\{  \frac{1}{\bgam!}   (\bc_j-\bx)^{\bgam}
    \nabla^{|\bgam|} \beta\bigg(
     \xi_{j}(\bc_{j}- \bx) + \bx
      \bigg)
      \notag
      \\
    & \hspace{4cm} \frac{det \left( \bE(\bx)' \bOmg(\bx;h_2) \bE_{\bo{0} \leftarrow e_j}(\bx) \right)}
   {det \left( \bE(\bx)' \bOmg(\bx;h_2) \bE(\bx) \right)} \Bigg\}
   \label{lemma_mls:1}
\end{align}
where the $\bE(\bx)$, $\bE_{\bo{0} \leftarrow e_j}(\bx)$,  and $\bOmg(\bx;h_2)$ matrices have been described above.
\end{lemma}
\begin{proof}
See the proof of Theorem 3.1 in \cite{lipman2006} and use the fact
that  $\nabla^{|\bgam |} \beta(\bx) = 0$ if $\gamma_2 >0$ and $\gamma_3 >0$.
\end{proof}

Suppose there is an increasing number of function evaluation points.
The values of $\bc_j$
are thought as coming from a triangular array indexed by $K$: $\left\{ \bc_{j,K} \right\}_j$.
The approximation error $\widehat \beta(\bx) - \beta(\bx)$ decreases to zero as $K$ grows large.
Lemma \ref{lemma_rate_int} below uses regularity conditions on the function $\beta$ and 
on the triangular array of points to determine the rate at which the approximation error converges to zero.
\begin{lemma}\label{lemma_rate_int}
Assume the conditions of Lemma \ref{lemma_mls} hold.
Furthermore, assume that 
\begin{itemize}
\item[(i)] $K \to \infty $, $h_2 \to 0$, $1/(K h_2^3) = O(1)$; 
\item[(ii)] there exists a positive definite $J \times J$ matrix $\bQ$ such that

$ \sup\limits_{  \bx \in \m{C}   }
\left\| {K h_2^3}
 \left[ \bE(\bx/h_2)' \bOmg(\bx;h_2) \bE(\bx/h_2) \right]^{-1}
- \bQ \right\|
=o\left( {1} \right)$; and

\item[(iii)] the function $\beta(\bx)$ has bounded derivatives on $\m{C}$
of order up to $\rho_2+2$.
\end{itemize}
Then,
\begin{align}
& \int_{\m{C}} \widehat \beta(\bx) - \beta(\bx) ~~d\bx  - \m{B}_{K} =  O\left( h_2^{\rho_2+2} \right)
\label{lemma_rate_int:1}
\\
& \text{ where } \notag
\\
& \m{B}_{K} =  \bigints_{\m{C}}  
	\sum\limits_{
		\substack{
			|\bgam|=\rho_2+1\\\min \{\gamma_2,\gamma_3 \}=0}
    	}
    	\sum\limits_{ j=1 }^{K}  \Bigg\{ 
    		\frac{1}{\bgam!}   (\bc_j-\bx)^{\bgam}
    		\nabla^{|\bgam|} \beta(	\bx )
\notag      \\
    & \hspace{4.5cm} \frac{det \left( \bE(\bx)' \bOmg(\bx;h_2) \bE_{\bo{0} \leftarrow e_j}(\bx) \right)}
   {det \left( \bE(\bx)' \bOmg(\bx;h_2) \bE(\bx) \right)} \Bigg\}
~~d\bx.
\label{lemma_rate_int:2}
\end{align}
and $\m{B}_{K} = O\left( h_2^{\rho_2+1} \right)$.

Moreover, there exists a $J \times 1$  vector $\Theta$ such that $e_1' \bQ \Theta > 0$, and 
\begin{align}
& \max_{1\leq j \leq K} \left\|
 h_2^{-3} \int_{\m{C}} \Omega_j(\bx;h_2) E_j(\bx / h_2) ~d \bx - \Theta
\right\| = o(1)
\label{lemma_rate_int:3}
\\
& \max_{1 \leq j \leq K} | K \Delta_j - e_1' \bQ \Theta | = o\left(1 \right).
\label{lemma_rate_int:4}
\end{align}

\end{lemma}
\begin{proof}

\textbf{Parts \eqref{lemma_rate_int:1} and  \eqref{lemma_rate_int:2} : } 

Start with Equation \ref{lemma_mls:1}.
Do a first-order Taylor expansion  of $\nabla^{|\bgam|} \beta(\xi_{j}(\bc_{j}- \bx) + \bx)$
around $\bx$ and substitute in Equation \ref{lemma_mls:1} to obtain
\begin{align}
\widehat \beta(\bx) - \beta(\bx) & =
	\sum\limits_{
		\substack{
    		|\bgam|=\rho_2+1
			\\
    		\min \{\gamma_2,\gamma_3 \}=0
    		}
    }
    \sum\limits_{ j=1 }^{K}  \Bigg\{
    	\frac{1}{\bgam!}   (\bc_j-\bx)^{\bgam}
    	\nabla^{|\bgam|} \beta(\bx)
      \notag
\\
& \hspace{3.5cm} 
	\frac{det \left( \bE(\bx)' \bOmg(\bx;h_2) \bE_{\bo{0} \leftarrow e_j}(\bx) \right)}
   		{det \left( \bE(\bx)' \bOmg(\bx;h_2) \bE(\bx) \right)} \Bigg\}
\\
+&\sum\limits_{
		\substack{
    		|\bgam|=\rho_2+1
			\\
    		\min \{\gamma_2,\gamma_3 \}=0
    		}
    }
    \sum\limits_{ j=1 }^{K}  \Bigg\{
    	\frac{1}{\bgam!}   (\bc_j-\bx)^{\bgam}
    	\sum\limits_{ | \boeta | =1 }
    		\nabla^{| \bgam + \boeta |} \beta\bigg(
    			\delta_{j}(\bc_{j}- \bx) + \bx
      		\bigg) (\bc_j-\bx)^{\boeta}
      \notag
\\
& \hspace{3.5cm} 
	\frac{det \left( \bE(\bx)' \bOmg(\bx;h_2) \bE_{\bo{0} \leftarrow e_j}(\bx) \right)}
   		{det \left( \bE(\bx)' \bOmg(\bx;h_2) \bE(\bx) \right)} \Bigg\}.
\end{align}
      
The integral over set $\m{C}$ is
\begin{align}
\int_{\m{C}} \widehat \beta(\bx) - \beta(\bx) ~d \bx& =
\m{B}_K
\\
+& \bigintss_{\m{C}} 
		\sum\limits_{
		\substack{
    		|\bgam|=\rho_2+1
			\\
    		\min \{\gamma_2,\gamma_3 \}=0
    		\\
    		| \boeta | =1
    		}
    }
    \sum\limits_{ j=1 }^{K}  \Bigg\{
    	\frac{1}{\bgam!}   (\bc_j-\bx)^{\bgam + \boeta}
   		\nabla^{| \bgam + \boeta |} \beta\bigg(
   			\delta_{j}(\bc_{j}- \bx) + \bx
       	\bigg) 
      \notag
\\
& \hspace{3.5cm} 
	\frac{det \left( \bE(\bx)' \bOmg(\bx;h_2) \bE_{\bo{0} \leftarrow e_j}(\bx) \right)}
   		{det \left( \bE(\bx)' \bOmg(\bx;h_2) \bE(\bx) \right)} \Bigg\} 
   		~d \bx.
\label{lemma_rate_int:2order}
\end{align}

The absolute value of the expression inside the integral in Equation \ref{lemma_rate_int:2order} is bounded by
\begin{align}
& \sum\limits_{
	\substack{
		|\bgam|=\rho_2+1
	\\
	\min \{\gamma_2,\gamma_3 \}=0
	\\
	| \boeta | =1}
}
	\sum\limits_{ j=1 }^{K}  \Bigg\{
		\frac{1}{\bgam!}   \left| \bc_j-\bx \right|^{\bgam + \boeta}
		\left| \nabla^{ | \bgam + \boeta | } 
		\beta\bigg(
			\delta_{j}(\bc_{j}- \bx) + \bx
  		\bigg) \right|
\\
& \hspace{3cm} 
		\left| e_1' \left( \bE(\bx)' \bOmg(\bx;h_2) \bE(\bx) \right)^{-1}
		\Omega_{j}(\bx;h_2) E_{j}(\bx) \right| ~
	\Bigg\}
\\
= & \sum\limits_{
	\substack{
		|\bgam|=\rho_2+1
	\\
	\min \{\gamma_2,\gamma_3 \}=0
	\\
	| \boeta | =1}
}
	\sum\limits_{ j=1 }^{K}  \Bigg\{
		\frac{1}{\bgam!}   \left| \bc_j-\bx \right|^{\bgam + \boeta}
		\left| \nabla^{ | \bgam + \boeta | } 
		\beta\bigg(
			\delta_{j}(\bc_{j}- \bx) + \bx
  		\bigg) \right|
\\
& \hspace{3cm} 
		\left| e_1' \left( \bE(\bx/h_2)' \bOmg(\bx;h_2) \bE(\bx/h_2) \right)^{-1}
		\Omega_{j}(\bx;h_2) E_{j}(\bx/h_2) \right| ~
	\Bigg\}
\\
\leq & 
M
h_2^{\rho_2+2}
\left\|
	Kh_2^3 \left( \bE(\bx/h_2)' \bOmg(\bx;h_2) \bE(\bx/h_2) \right)^{-1} 
\right\|
\frac{1}{K}\sum\limits_{ j=1 }^{K}
	\left\| 
		\frac{1}{h_2^3} \Omega_{j}(\bx;h_2) E_{j}(\bx/h_2)    
	\right\|
\\
\leq &
M  h_2^{\rho_2+2} O(1) O(1) = O\left( h_2^{\rho_2+2} \right).
\end{align}
where it is used that the derivatives of $\beta$ are bounded;
that $| (\bc_j-\bx )^{\bgam+\boeta}| \leq h_2^{\rho_2 +2}$;
that the norm of the inverse of $Kh_2^3 \left( \bE(\bx/h_2)' \bOmg(\bx;h_2) \bE(\bx/h_2) \right)$
is bounded over $x$ and $n$ (Assumption (ii));
and the fact that
$\sum_j \left\| \Omega_{j}(\bx;h_2) E_{j}(\bx/h_2)    \right\| \leq  M K h_2^3$.
It follows that Equation \ref{lemma_rate_int:2order} is $O\left( h_2^{\rho_2+2} \right)$.
A similar argument yields $\m{B}_K = O\left( h_2^{\rho_2+1} \right)$.

\textbf{Part \eqref{lemma_rate_int:3} : }

Assume WLOG the support of the kernel is $[-1,1]$ (Assumption \ref{assu_srd_est_kernel}).
Define:

$F(\bx) = E_j(\bx+\bc_j)$;

$\m{C}_{h_2} = \{\bx \in \mmr^3 : \prod_{i=1}^3 (x_i \pm  h_2) \subseteq \m{C} \}$,
where $\prod$ is used to denote the Cartesian product;

$\Theta = \int_{[-1,1]^3}
  k(u_1)k(u_2) k(u_3) F(\bu) ~ d\bu$, where
$\bu=(u_1,u_2,u_3)$.

\begin{gather*}
0 \leq \max_{j: \bc_j \in \m{C}_{h_2}}
\left\| {h_2^{-3}} \int_{\m{C}}
\Omega_j(\bx;h_2) E_j(\bx/h_2) ~ d\bx - \Theta \right\|
\\
\leq
\sup_{\bc \in \m{C}_{h_2}}
\left\| {h_2^{-3}} \int_{\bc \pm h_2}
\prod_{i=1}^3 k((x_i - c_{i})/h_2) F((\bx - \bc)/h_2) ~ d\bx - \Theta \right\|
\\
=
\sup_{\bc \in \m{C}_{h_2}}
\left\| \int_{[-1,1]^3}
 k(u_1)k(u_2) k(u_3) F(\bu) ~ d\bu - \Theta \right\| = 0
\end{gather*}
where the transformation $\bu=(\bx - \bc)/h_2$ is used.
The result follows from the fact that $\m{C}_{h_2} \uparrow \m{C}$.

\textbf{Part \eqref{lemma_rate_int:4} : }

Using the formula for the correction weights $\Delta_j$
\begin{gather*}
\left| K \Delta_j - e_1' \bQ \Theta \right|
\\
= \left| K \int_{\m{C}} ~ e_1'
 \left( \bE(\bx/h_2)' \bOmg(\bx;h_2) \bE(\bx/h_2) \right)^{-1}
 \Omega_{j}(\bx;h_2) E_{j}(\bx/h_2) ~~d(\bx) - e_1' \bQ \Theta \right|
\\
= \left|  \int_{\m{C}} ~ e_1'
 \left[Kh_2^3 \left( \bE(\bx/h_2)' \bOmg(\bx;h_2) \bE(\bx/h_2) \right)^{-1} \right] h_2^{-3}
 \Omega_{j}(\bx;h_2) E_{j}(\bx/h_2) ~~d(\bx) - e_1' \bQ \Theta \right|
\\
\leq \left|  \int_{\m{C}} ~ e_1'
 \left[Kh_2^3 \left( \bE(\bx/h_2)' \bOmg(\bx;h_2) \bE(\bx/h_2) \right)^{-1} - \bQ \right] h_2^{-3}
 \Omega_{j}(\bx;h_2) E_{j}(\bx/h_2) ~~d(\bx)  \right|
\\
+
\left|  \int_{\m{C}} ~ e_1'\bQ
   h_2^{-3}
 \Omega_{j}(\bx;h_2) E_{j}(\bx/h_2) ~~d(\bx) - e_1' \bQ \Theta \right|
\\
\leq
\int_{\m{C}} ~
\left\| Kh_2^3
\left( \bE(\bx/h_2)' \bOmg(\bx;h_2) \bE(\bx/h_2) \right)^{-1} - \bQ
\right\|
h_2^{-3}
 \left\| \Omega_{j}(\bx;h_2) E_{j}(\bx/h_2) \right\| ~~d(\bx)
\\
+
\left| e_1' \bQ  \left[
\int_{\m{C}} ~
h_2^{-3}
\Omega_{j}(\bx;h_2) E_{j}(\bx/h_2) ~~d(\bx)
- \Theta
\right]
\right|
\\
= o(1) O(1) + o(1) = o(1)
\end{gather*}

Next, 
\begin{gather*}
\left| 
	\int_{\m{C}} 1 ~d\bc - e_1' \bQ \Theta
\right|
\leq
\left| 
	\int_{\m{C}} 1 ~d\bc - \sum_j \Delta_j  
\right|
+
\left| 
	\sum_j \Delta_j - e_1' \bQ \Theta
\right|
\\
\leq o(1) +  \frac{1}{K} \sum_j \left| K\Delta_j - e_1' \bQ \Theta \right|
\\
\leq o(1) + \max_j \left| K\Delta_j - e_1' \bQ \Theta \right|
=o(1)
\end{gather*}
shows that $e_1' \bQ \Theta = \int_{\m{C}} 1 ~d\bc > 0 $.
\end{proof}

\begin{remark}
Lemma \ref{lemma_rate_int} also applies to weighted integrals of the form
\[
 \mu = \int_{\m{C}} \omega(\bx) \beta(\bx) ~~ d(\bx)
\]
where $\omega(\bx)$ is a probability density function that is continuous, bounded and bounded away from zero.
There are three main differences between unweighted integrals (treated above) and weighted integrals (considered in the main text) : (i) the formula for the weights $\Delta_j$ changes to
\begin{gather*}
 \Delta_j = \bigintsss_{\m{C}} ~ \omega(\bx) e_1'
 \left( \bE(\bx)' \bOmg(\bx;h_2) \bE(\bx) \right)^{-1}
 \Omega_{j}(\bx;h_2) E_{j}(\bx) ~~d(\bx)
 \\
 = \bigintsss_{\m{C}} ~ \omega(\bx)
 \frac{det \left( \bE(\bx)' \bOmg(\bx;h_2) \bE_{\bo{0} \leftarrow e_j}(\bx) \right)}
   {det \left( \bE(\bx)' \bOmg(\bx;h_2) \bE(\bx) \right)}
   ~~d(\bx),
\end{gather*}
(ii) the formula for the bias changes to
\begin{align*}
 \m{B}_{K} = & \bigints_{\m{C}}  \omega(\bx)
	\sum\limits_{
		\substack{
			|\bgam|=\rho_2+1\\\min \{\gamma_2,\gamma_3 \}=0}
    	}
    	\sum\limits_{ j=1 }^{K}  \Bigg\{ 
    		\frac{1}{\bgam!}   (\bc_j-\bx)^{\bgam}
    		\nabla^{|\bgam|} \beta(	\bx )
\notag      \\*
   & \hspace{4.5cm} \frac{det \left( \bE(\bx)' \bOmg(\bx;h_2) \bE_{\bo{0} \leftarrow e_j}(\bx) \right)}
   {det \left( \bE(\bx)' \bOmg(\bx;h_2) \bE(\bx) \right)} \Bigg\}
~~d\bx,
\end{align*}
and (iii) conclusion \ref{lemma_rate_int:4} of Lemma \ref{lemma_rate_int} changes to
\[
\max_{1 \leq j \leq K} | K \Delta_j / {\omega(\bc_j)} - e_1' \bQ \Theta | = o\left(1 \right).
\]

\end{remark}

Lemma \ref{lemma_rate_int} states a condition on the asymptotic behavior
of the triangular array of points $\{ \bc_j \}_{j=1}^K$.
For large $K$,
the observations must be uniformly distributed on the domain $\m{C}$ such that
$\bE(\bx)' \bOmg(\bx;h_2) \bE(\bx)$ is invertible and of magnitude
$Kh_2^3$, that is, $K$ times the volume of every $h_2$-neighborhood of $\bx$, for every $\bx$ in $\m{C}$.
These conditions are satisfied in a variety
of examples of triangular arrays of points that cover $\m{C}$ uniformly well for large $K$.

To be clearer, this assumption is illustrated in a simple example.
In the main text, the conditions of Lemma \ref{lemma_rate_int}
are restated in Assumption \ref{assu_srd_kinf_est_cut}(c)
and in the rate conditions of Theorem \ref{theo_srd_kinf_est_int}.
The choice of $h_1,h_2,\rho_2$ needs to satisfy both the conditions in Assumption \ref{assu_srd_kinf_est_cut}
and the rate conditions of Theorem \ref{theo_srd_kinf_est_int}.

Pick the choices given in the example of Figure \ref{figure_rates}, for which,
$h_1=K^{-\lambda_1/\lambda_2}$, $h_2=K^{-3/10}$, and $\rho_2=3$.
Let $\bc \in \mmr^3$, $\bx\in \mmr^3$, $k(u)=.5\mmi\{|u| \leq 1\}$,
and $\m{C}=(0,1)^3$.
Define $N$ points for each $l$th coordinate of $\bc=(c_1,c_2,c_3)$ as
$c_{l,j,N}=j/(N+1)$, $j=1,\ldots,N$, $l=1,2,3$. In this case, $K=N^3$,
and $h_2=1/N^{9/10}$. Assumption \ref{assu_srd_kinf_est_cut}(b)
requires the distance $c_{1,j+1,K} - c_{1,j,K}=1/(N+1) = 1 / (K^{1/3}+1)$,
to be greater than the order of $h_1 = K^{-\lambda_1/\lambda_2}$. This is equivalent to
$\lambda_1 > \lambda_2 / 3$ which is always satisfied for the choices in the feasibility
set depicted in Figure \ref{figure_rates}. Next, condition (c) in Assumption
\ref{assu_srd_kinf_est_cut} is shown.

Define $\tilde K = \sum_{(l_1,l_2,l_3)}^{K} \mmi \left\{ \Omega_{(l_1,l_2,l_3)}(\bx;h_2)>0 \right\}$,
where $(l_1,l_2,l_3)$ indexes point $\bc_{(l_1,l_2,l_3)}=(c_{l_1}, c_{l_2}, c_{l_3}) $.
For each $\bx$, the number of
$\bc_{(l_1,l_2,l_3)}$ in the $h_2$-neighborhood of $\bx$
grows to infinity at $ K^{1/10} = K h_2^3$ rate, so $\tilde K = O(K h_2^3)$.
This rate of growth is uniform over $\bx \in \m{C}$.
The vector of polynomials $E_j(\bx)$ is written as $E_j(\bx) = F\left(\bx-\bc_{(l_1,l_2,l_3)}\right)$ where
\begin{gather*}
F(\bu)=\Big[\bu^{(0,0,0)}, \bu^{(1,0,0)}, \bu^{(0,1,0)}, \ldots, \bu^{(2,0,1)}
\Big]'
\end{gather*}
that is, all polynomials $\bu^{(\gamma_1,\gamma_2,\gamma_3)}$
such that $\gamma_i \in \mmz_+$ $\forall i$, $0\leq \gamma_1+\gamma_2+\gamma_3 \leq 3$,
and $\min\{\gamma_2,\gamma_3 \}=0$. Then, $F(\bu)$ is a $J \times 1$ vector where $J=16.$

Now, fix $\bx \in (0,1)^3$ and a large $K$. Consider an uniform discrete random vector $\tilde \bu$
taking values on $(-1,1)^3$ according to
$\bu_{(l_1,l_2,l_3)}  = \left( \bx - \bc_{(l_1,l_2,l_3)} \right) / h_2 $
for all $\bc_{(l_1,l_2,l_3)} \in (\bx \pm h_2)$.
It turns out that
\begin{gather*}
 \frac{1}{\tilde K} \bE(\bx/h_2)' \bOmg(\bx;h_2) \bE(\bx/h_2)
\\
=
\frac{1}{2 \tilde K} \sum_{(l_1,l_2,l_3)} \mmi \left\{ \Omega_{(l_1,l_2,l_3)}(\bx;h_2)>0 \right\}
F\left( (\bx - \bc_{(l_1,l_2,l_3)})/h_2 \right) F\left( (\bx - \bc_{(l_1,l_2,l_3)})/h_2 \right)'
\\
= \frac{1}{2 } \mme\left[ F(\tilde \bu) F(\tilde \bu)' \right]
\end{gather*}
This is approximately equal to $\int_{\bu \in [-1,1]^3} F(\bu) F(\bu)'~d \bu$
uniformly in $\bx$. Simply call $\bQ$ the inverse of this integral, a positive definite
matrix.
Finally,
\[
\sup\limits_{  \bx \in \m{C}  }
\left\|
 \left[ \frac{1}{K h_2^3} \bE(\bx/h_2)' \bOmg(\bx;h_2) \bE(\bx/h_2) \right]^{-1}
 - \bQ
\right\|
=o( 1 )
\]

\subsection{Consistent Estimation of Standard Errors}\label{sec:supp:app:est:se}

\indent

This section demonstrates that the estimator for the variance of $\ha{\mu}^c$ proposed in Section \ref{sec_case2} is a consistent estimator.
For the nearest-neighbor matching, the distribution of $X_i$ is continuous, so I assume $X_1 < \ldots < X_n$ WLOG.
For a fixed number of neighbors $N \in \mmz_+$, define
\begin{align}
c : & \m{X} \to \{0, 1, \ldots, K \} \text{, where}
\notag
\\
     & \hspace{1cm} c(x) = \max_{0 \leq j \leq K} \left\{ 
	c_j : c_j \leq x
\right\}
\\
\ell :&  \{1, \ldots, n \} \times \mmz_+ \to \{1, \ldots, n \} \text{, where } \ell(i,N) \text{ is such that }
\notag
\\
& \hspace{1cm}  	 \sum_{ \substack{v=1 \\ v\neq i, v\neq \ell(i,N) } }^n 
	\mmi\left\{
		|X_v-X_i| \leq |X_{\ell(i,N)} - X_i |, c(X_v)=c(X_i)
	\right\}
	=N
\\
\ha{\eps}_i^2 = & \frac{N}{N+1} \left(
	Y_i  - \frac{1}{N} \sum_{l = 1 }^N Y_{\ell(i,l)}
\right)^2
\end{align}
The expression for the variance estimator in the continuous case is given in Equation \ref{est:var:c}.
\begin{lemma}\label{lemma:est:var}
Assume the conditions of Theorem \ref{theo_srd_kinf_est_int} and $(K \underline{h}_{1})^{-1}=O(1) $ hold.
Then, $\ha{\m{V}}_n^c / \m{V}_n^c \pto 1$.
\end{lemma}
\begin{proof}
The proof extends the arguments of Theorem A3 by CCT to the case where the number of cutoffs grows to infinity.

Define $\phi_n$ and $\ha{\phi}_n$ by 

\begin{align}
\phi_n(X_i) = &
		\sum_{j=1}^K 
			\frac{\Delta_j}{nh_{1j}}
			k\left( \frac{X_i - c_{j}}{h_{1j}} \right)
			e_1'
			\left(
				v_i^{j +} \mme[ G_n^{j +} ]  
				-
				v_i^{j -} \mme[ G_n^{j -} ]  
			\right)
			\ti{H}_{i}^{j}		
\\
\ha{\phi}_n(X_i) = &
		\sum_{j=1}^K 
			\frac{\Delta_j}{nh_{1j}}
			k\left( \frac{X_i - c_{j}}{h_{1j}} \right)
			e_1'
			\left(
				v_i^{j +}  G_n^{j +}  
				-
				v_i^{j -}  G_n^{j -}  
			\right)
			\ti{H}_{i}^{j}.
\end{align}

Use them to rewrite $\m{V}_n^c$ and $\ha{\m{V}}_n^c$ as
\begin{align}
\m{V}_n^c = &n \mme \left[
	\eps_i^2 \phi_n(X_i)^2
\right]
\\
\ha{\m{V}}_n^c = &\sum_{i=1}^n 
	\ha{\eps}_i^2 \ha{\phi}_n(X_i)^2.
\end{align}

In order to show $\ha{\m{V}}_n^c / {\m{V}}_n^c \pto 1$, it suffices to show that 
$(Kn\overline{h}_1) ( \ha{\m{V}}_n^c - {\m{V}}_n^c ) \pto 0$ because Theorem \ref{theo_srd_kinf_est_int} shows that 
$({\m{V}}_n^c)^{-1} = O(Kn\overline{h}_1) $.

\begin{align}
(Kn\overline{h}_1) ( \ha{\m{V}}_n^c - {\m{V}}_n^c )  = & 
(Kn\overline{h}_1)  \sum_{i=1}^n 
	\ha{\eps}_i^2 \ha{\phi}_n(X_i)^2
-
(Kn\overline{h}_1)  \sum_{i=1}^n 
	\eps_i^2 \phi_n(X_i)^2
\\
= &
(Kn\overline{h}_1)  \sum_{i=1}^n 
	\ha{\eps}_i^2 \left( \ha{\phi}_n(X_i) -\phi_n(X_i) + \phi_n(X_i) \right)^2
-
(Kn\overline{h}_1)  \sum_{i=1}^n 
	\eps_i^2 \phi_n(X_i)^2
\\
= &
(Kn\overline{h}_1)  \sum_{i=1}^n 
	\ha{\eps}_i^2 \left( \ha{\phi}_n(X_i) -\phi_n(X_i) \right)^2
\label{eq:lemma:est:var:c:part1}
\\
&\hspace{1cm} 
+2
(Kn\overline{h}_1)  \sum_{i=1}^n 
	\ha{\eps}_i^2 \left( \ha{\phi}_n(X_i) -\phi_n(X_i) \right) \phi_n(X_i)
\label{eq:lemma:est:var:c:part2}
\\
& \hspace{1cm} 
+
(Kn\overline{h}_1)  \sum_{i=1}^n 
	\left( \ha{\eps}_i^2 - \eps_i^2 \right) \phi_n(X_i)^2
\label{eq:lemma:est:var:c:part3}
\end{align}

The rest of the proof shows that parts \eqref{eq:lemma:est:var:c:part1} - \eqref{eq:lemma:est:var:c:part3}
converge in probability to zero. 

\bigskip

\begin{center}
\textbf{\underline{Part \eqref{eq:lemma:est:var:c:part1} }} 
\end{center}

First, for arbitrary $x \in \m{X}$ 
\begin{align}
\left| 
	\ha{\phi}_n(x) -\phi_n(x)
\right| 
 & 
\notag \\
		& \hspace{-2.5cm} \leq
		\sum_{j=1}^K \Bigg\{
			\frac{| \Delta_j | }{nh_{1j}}
			\left| k\left( \frac{x - c_{j}}{h_{1j}} \right) \right| 
\notag \\
		& \hspace{-.5cm} 
			\left| e_1'
			\left(
				v_{c_j,h_{1j}}^{+}(x) \left( G_n^{j +} - \mme[ G_n^{j +} ]  \right)
				-
				v_{c_j,h_{1j}}^{-}(x) \left( G_n^{j -} - \mme[ G_n^{j -} ]  \right)
			\right)
			H\left( \frac{x-c_j}{h_{1j}} \right)	\right|
		\Bigg\}
\\
	& \hspace{-2.5cm} \leq	
		2 \max_{1 \leq j \leq K} \Bigg\{
			\frac{| \Delta_j | }{nh_{1j}}
			\left| k\left( \frac{x - c_{j}}{h_{1j}} \right) \right| 
\notag \\
	& \hspace{-.5cm} 
			\left| e_1'
			\left(
				v_{c_j,h_{1j}}^{+}(x) \left( G_n^{j +} - \mme[ G_n^{j +} ]  \right)
				-
				v_{c_j,h_{1j}}^{-}(x) \left( G_n^{j -} - \mme[ G_n^{j -} ]  \right)
			\right)
			H\left( \frac{x-c_j}{h_{1j}} \right)	\right|		
		\Bigg\}
\\
	& \hspace{-2.5cm} =	
	O \left( \frac{1}{  K n \underline{h}_1 } \right)
	O_P \left( \sqrt{ \frac{\log n}{n \overline{h}_1 } }  \right)
=
	O_P \left( \frac{1}{  K n \overline{h}_1 } \sqrt{ \frac{\log n}{n \overline{h}_1 } }  \right)	
\label{eq:lemma:est:var:c:phidelta}
\end{align}
where the second inequality uses the fact that at most two elements of the sum over 
$j$ are non-zero for each value of $x$;
the first equality relies on $\max_j | \Delta_j | = O(K^{-1})$ (Lemma \ref{lemma_rate_int}), 
on $h_{1j}^{-1}\leq \underline{h}_1^{-1}$,
on the fact that the kernel is bounded (Assumption \ref{assu_srd_est_kernel}),
that $v_{c_j,h_{1j}}^{\pm}(x)H(h_{1j}^{-1}(x-c_j))$ is bounded,
and that
$\max_j \left\|
G_n^{j \pm} - \mme[ G_n^{j \pm} ]
\right\| = O_P( ( \log n / K \overline{h_1})^{1/2}  )$ (Lemma \ref{lemma_app});
the last equality uses the rate condition  $ \overline{h}_1 / \underline{h}_1 = O(1)$.
The rate in \eqref{eq:lemma:est:var:c:phidelta} is uniform over $x \in \m{X}$.
 
 Then, it follows that
\begin{align}
\left| \eqref{eq:lemma:est:var:c:part1} \right| 
\leq &
(Kn\overline{h}_1) n \frac{1}{n} \sum_{i=1}^n 
	\ha{\eps}_i^2 \max_x \left| \ha{\phi}_n(x) -\phi_n(x) \right|^2
\\
=& (Kn\overline{h}_1)  n O_P(1) O_P \left( \frac{1}{  (K n \overline{h}_1)^2 }  \frac{\log n}{n \overline{h}_1  }  \right)
=\frac{1}{K \overline{h}_1} 	\frac{\log n}{n \overline{h}_1  } O_P(1) = o_P(1)
\end{align} 
 where the first equality used the rate derived in \eqref{eq:lemma:est:var:c:phidelta}, and the fact that $\ha{\eps}_i$ is a.s. bounded  because 
 $\eps_i$ is a.s. bounded 
 (Assumption \ref{assu_srd_kinf_est_int});
 the last equality relied on the rate conditions $( K \overline{h}_1)^{-1}=O(1)$
 and $\log n (n \overline{h}_1)^{-1}=o(1)$.

\bigskip

\begin{center}
\textbf{\underline{Part \eqref{eq:lemma:est:var:c:part2} }} 
\end{center}

First, for arbitrary $x \in \m{X}$ 

\begin{align}
\left| 
	\phi_n(x) 
\right| 
\leq &  
		\sum_{j=1}^K \Bigg\{
			\frac{| \Delta_j | }{nh_{1j}}
			\left| k\left( \frac{x - c_{j}}{h_{1j}} \right) \right| 
\notag \\
		& \hspace{1cm} 
			\left| e_1'
			\left(
				v_{c_j,h_{1j}}^{+}(x)  \mme[ G_n^{j +} ]  
				-
				v_{c_j,h_{1j}}^{-}(x)  \mme[ G_n^{j -} ]  
			\right)
			H\left( \frac{x-c_j}{h_{1j}} \right)	\right|
		\Bigg\}
\\
\leq &  
		2 \max_{1 \leq j \leq K} \Bigg\{
			\frac{| \Delta_j | }{nh_{1j}}
			\left| k\left( \frac{x - c_{j}}{h_{1j}} \right) \right| 
\notag \\
		& \hspace{1.5cm} 
			\left| e_1'
			\left(
				v_{c_j,h_{1j}}^{+}(x)  \mme[ G_n^{j +} ]  
				-
				v_{c_j,h_{1j}}^{-}(x)  \mme[ G_n^{j -} ]  
			\right)
			H\left( \frac{x-c_j}{h_{1j}} \right)	\right|
		\Bigg\}
\\
= & 
	O \left( \frac{1}{  K n \underline{h}_1 } \right)
=
	O \left( \frac{1}{  K n \overline{h}_1 }   \right)	
\label{eq:lemma:est:var:c:phi}
\end{align}
where the second inequality uses the fact that at most two elements of the sum over 
$j$ are non-zero for each value of $x$;
the first equality relies on $\max_j | \Delta_j | = O(K^{-1})$ (Lemma \ref{lemma_rate_int}), 
on $h_{1j}^{-1}\leq \underline{h}_1^{-1}$,
on the fact that the kernel is bounded (Assumption \ref{assu_srd_est_kernel}),
that $v_{c_j,h_{1j}}^{\pm}(x)H(h_{1j}^{-1}(x-c_j))$ is bounded,
and that $\mme[ G_n^{j \pm } ]$ is approximately equal to a positive definite matrix with determinant bounded away from zero (Lemma \ref{lemma_app});
the last equality uses the rate condition  $ \overline{h}_1 / \underline{h}_1 = O(1)$.
The rate in \eqref{eq:lemma:est:var:c:phi} is uniform over $x \in \m{X}$.

 Then, it follows that
\begin{align}
\left| \eqref{eq:lemma:est:var:c:part2} \right| 
\leq &
(Kn\overline{h}_1) n \frac{1}{n} \sum_{i=1}^n 
	\ha{\eps}_i^2 \max_x \left| \ha{\phi}_n(x) -\phi_n(x) \right| \max_x \left| \phi_n(x) \right| 
\\
=& (Kn\overline{h}_1)  n O_P(1) O_P \left( \frac{1}{  K n \overline{h}_1 }  \sqrt{ \frac{ \log n}{n \overline{h}_1  } }  \right)
O \left( \frac{1}{  K n \overline{h}_1 } \right)
=\frac{1}{K \overline{h}_1} 	\sqrt{ \frac{\log n}{n \overline{h}_1  } } O_P(1) = o_P(1)
\end{align} 
 where the first equality used the rate derived in \eqref{eq:lemma:est:var:c:phidelta}
and  \eqref{eq:lemma:est:var:c:phi}, and the fact that $\ha{\eps}_i$ is a.s. bounded because $\eps_i$ is a.s. bounded 
 (Assumption \ref{assu_srd_kinf_est_int});
 the last equality relied on the rate conditions $( K \overline{h}_1)^{-1}=O(1)$
 and $\log n (n \overline{h}_1)^{-1}=o(1)$.

\bigskip

\begin{nopgbreak}
\begin{center}
\textbf{\underline{Part \eqref{eq:lemma:est:var:c:part3} }} 
\end{center}

First,  expand $\ha{\eps}_i^2$ around $\eps_i^2$.
To simplify notation,  abbreviate 
$\mme[Y_i | X_i]=R(X_i,D_i) $ to $R_i$. 
\end{nopgbreak}

\begin{align}
\ha{\eps}_i^2 = & \frac{N}{N+1} \left(
	Y_i  - \frac{1}{N} \sum_{l = 1 }^N Y_{\ell(i,l)}
\right)^2
\\
 = & \frac{N}{N+1} \left(
	R_i + \eps_i   - \frac{1}{N} \sum_{l = 1 }^N (R_{\ell(i,l)} + \eps_{\ell(i,l)})
\right)^2
\\
 = & \frac{N}{N+1} \left(
	 \eps_i - \frac{1}{N} \eps_{\ell(i,l)} 
\right)^2
+ \frac{N}{N+1} \left(
		\frac{1}{N} \sum_{l = 1 }^N (R_i - R_{\ell(i,l)} )
\right)^2
\notag\\
& \hspace{1cm} +\frac{2N}{N+1} \left(
	 \eps_i - \frac{1}{N} \eps_{\ell(i,l)} 
\right)
\left(
		\frac{1}{N} \sum_{l = 1 }^N (R_i - R_{\ell(i,l)} )
\right)
\\
= & \eps_i^2  
- \frac{2 \eps_i}{N+1} \sum_{l=1}^N \eps_{\ell(i,l)} 
+ \frac{1}{N(N+1)} \sum_{l=1}^N ( \eps_{\ell(i,l)}^2 - \eps_i^2 )
\notag\\
& \hspace{1cm} + \frac{2}{N(N+1)} \sum_{l=1}^N \sum_{  v=1,   v > l }^N \eps_{\ell(i,l)} \eps_{\ell(i,v)}
+ \frac{1}{N(N+1)} \left( \sum_{l=1}^N  R_i - R_{\ell(i,l)} \right)^2
\notag\\
& \hspace{1cm} + \frac{2 \eps_i }{N+1} \sum_{l=1}^N \left( R_i - R_{\ell(i,l)} \right) 
- \frac{2 }{N(N+1)} \sum_{l=1}^N \eps_{\ell(i,l)} \sum_{v=1}^N (R_i - R_{\ell(i,v)})
\end{align}

Then, substitute the expression for $\ha{\eps_i}^2 - \eps_i^2$ derived above in part \eqref{eq:lemma:est:var:c:part3}:
\begin{align}
(\ref{eq:lemma:est:var:c:part3}) = 
& (K n \overline{h}_1 ) \sum_{i=1}^n \left\{ \frac{-2 \eps_i}{N+1} \sum_{l=1}^N \eps_{\ell(i,l)} \right\} \phi_n(X_i)^2 
\label{eq:lemma:est:var:c:part3:1}\\
& \hspace{1cm} + (K n \overline{h}_1 ) \sum_{i=1}^n \left\{ 
	\frac{1}{N(N+1)} \sum_{l=1}^N ( \eps_{\ell(i,l)}^2 - \eps_i^2 )
\right\} \phi_n(X_i)^2
\label{eq:lemma:est:var:c:part3:2}
\\
& \hspace{1cm} + (K n \overline{h}_1 ) \sum_{i=1}^n \left\{ 
	\frac{2}{N(N+1)} \sum_{l=1}^N \sum_{  v=1,   v > l }^N \eps_{\ell(i,l)} \eps_{\ell(i,v)}
\right\} \phi_n(X_i)^2
\label{eq:lemma:est:var:c:part3:3}
\\
& \hspace{1cm} + (K n \overline{h}_1 ) \sum_{i=1}^n \left\{
	\frac{1}{N(N+1)} \left( \sum_{l=1}^N  R_i - R_{\ell(i,l)} \right)^2
\right\} \phi_n(X_i)^2
\label{eq:lemma:est:var:c:part3:4}
\\
& \hspace{1cm} + (K n \overline{h}_1 ) \sum_{i=1}^n \left\{ 
	\frac{2 \eps_i }{N+1} \sum_{l=1}^N \left( R_i - R_{\ell(i,l)} \right) 
\right\} \phi_n(X_i)^2
\label{eq:lemma:est:var:c:part3:5}
\\
& \hspace{1cm} + (K n \overline{h}_1 ) \sum_{i=1}^n \left\{
	\frac{-2 }{N(N+1)} \sum_{l=1}^N \eps_{\ell(i,l)} \sum_{v=1}^N (R_i - R_{\ell(i,v)})
\right\} \phi_n(X_i)^2
\label{eq:lemma:est:var:c:part3:6}
\end{align}

The steps below demonstrate that parts \eqref{eq:lemma:est:var:c:part3:1} - \eqref{eq:lemma:est:var:c:part3:6}
converge in probability to zero. 

\textbf{Part \eqref{eq:lemma:est:var:c:part3:1}:} the expected value $\mme[\eqref{eq:lemma:est:var:c:part3:1} | \m{X}_n]=0$.
To compute the variance of \eqref{eq:lemma:est:var:c:part3:1} centered at $\mme[\eqref{eq:lemma:est:var:c:part3:1} | \m{X}_n]=0$, abbreviate
$N^{-1}\sum_{l=1}^N \eps_{\ell(i,l)}$ to $\overline{\eps}_i$,
$\phi_n(X_i)$ to $\phi_{ni}$,
$\mmi\{ c(X_i) = c(X_j) \}$ to $\mmi_{ij}^{=}$,
and 
$\mmi\{ c(X_i) \neq c(X_j) \}$ to $\mmi_{ij}^{\neq}$. Then,
\begin{align}
\mme\left[
	\left( \eqref{eq:lemma:est:var:c:part3:1} -  \mme[\eqref{eq:lemma:est:var:c:part3:1} | \m{X}_n] \right)^2
\right]
= &
M
( K n \overline{h}_1 )^2
\mme\left[
	\sum_{i=1}^n \sum_{j=1}^n (\mmi_{ij}^{=} + \mmi_{ij}^{\neq}) \eps_i \overline{\eps}_i \phi_{ni}^2 \eps_j \overline{\eps}_j \phi_{nj}^2 		
\right]
\\
= & 
M
( K n \overline{h}_1 )^2
\mme\left[
	\sum_{i=1}^n \sum_{j=1}^n \mmi_{ij}^{=}  \eps_i \overline{\eps}_i \phi_{ni}^2 \eps_j \overline{\eps}_j \phi_{nj}^2 		
\right]
\\
= & 
M
( K n \overline{h}_1 )^2
\sum_{i=1}^n \sum_{j=1}^n 
	\mme\left[ \mmi_{ij}^{=} \right]
	O_P\left(( K n \overline{h}_1 )^{-4}  \right)
\\
= & 
M
 ( K n \overline{h}_1 )^{-2} 
n^2
O_P\left( K^{-1}\right)
=  o_P(1)
\end{align}
where 
the second equality uses that 
the expected value of $\mmi_{ij}^{\neq} \eps_i \overline{\eps}_i \phi_{ni}^2 \eps_j \overline{\eps}_j \phi_{nj}^2$
conditional on $\m{X}_n$ is zero because
if $\mmi_{ij}^{\neq} =1$, then $\eps_i \overline{\eps}_i$ is independent of 
$\eps_j \overline{\eps}_j$,
and  
$
\mme[ \eps_i \overline{\eps}_i | \m{X}_n ] = \mme[ \eps_i  | \m{X}_n ] \mme[ \overline{\eps}_i | \m{X}_n ] =0
$;
the third equality relies on the fact that 
$\eps_i$ and $\overline{\eps}_i$ are a.s. bounded  (Assumption \ref{assu_srd_kinf_est_int}),
and that 
$\phi_{ni}$ is $O\left(( K n \overline{h}_1 )^{-1}  \right)$ (Equation \ref{eq:lemma:est:var:c:phi});
the fourth equality uses that 
$\mme[\mmi^{=}_{ij}] = \mme[\mmp(c_j \leq X_j < c_{j+1} | X_i ) ] = O(K^{-1})$
(for some $j$ as function of $X_i$) because the derivative of the pdf of $X_i$ is bounded (Assumption
\ref{assu_srd_est_fx}) and
$\max_j |c_{j+1} - c_{j}| =O(K^{-1})$ (Assumption \ref{assu_srd_kinf_est_cut}).
The Chebyshev's inequality yields that $\eqref{eq:lemma:est:var:c:part3:1} = o_P(1)$.

\bigskip

\textbf{Part \eqref{eq:lemma:est:var:c:part3:2}:} for $x^*_{il}$ between $X_i$ and $X_{\ell(i,l)}$:
\begin{align}
\mme[\eqref{eq:lemma:est:var:c:part3:2} | \m{X}_n] =& M 
(K n \overline{h}_1 ) \sum_{i=1}^n \left\{ 
	\sum_{l=1}^N ( \sigma^2(X_{\ell(i,l)},D_i) - \sigma^2(X_i,D_i) )
\right\} \phi_{ni}^2
\\
= & 
M 
(K n \overline{h}_1 ) \sum_{i=1}^n \left\{ 
	\sum_{l=1}^N  \nabla_x \sigma^2(x^*_{il},D_i) (X_{\ell(i,l)}  - X_i  )
\right\} \phi_{ni}^2
\\
= & (K n \overline{h}_1 ) n O\left( K^{-1} \right) O\left( ( K n \overline{h}_1 )^{-2}  \right)
= O\left( K^{-1} ( K  \overline{h}_1 )^{-1}  \right) = o(1)
\end{align}
which uses that $\nabla_x \sigma^2(x^*_{il}, D_i)$ is bounded (Assumption \ref{assu_srd_est_mx}),
that $| X_{\ell(i,l)}  - X_i  | \leq \max_j |c_{j+1} - c_{j}| =O(K^{-1})$,
and that $\phi_{ni}$ is $O\left(( K n \overline{h}_1 )^{-1}  \right)$ (Equation \ref{eq:lemma:est:var:c:phi}).
To compute the variance, let $\nu_i =  \eps^2_i - \sigma^2(X_i,D_i)$
and use the abbreviations from part \eqref{eq:lemma:est:var:c:part3:1}.
 Then,
\begin{align}
\mme\left[
	\left( \eqref{eq:lemma:est:var:c:part3:2} -  \mme[\eqref{eq:lemma:est:var:c:part3:2} | \m{X}_n] \right)^2
\right] &
\notag\\
 & \hspace{-3.5cm}= 
M \mme\left\{
	(K n \overline{h}_1 ) \sum_{i=1}^n \left[ 
		\sum_{l=1}^N \left( (\eps^2_{\ell(i,l)} -  \sigma^2(X_{\ell(i,l)},D_i)) - ( \eps^2_i - \sigma^2(X_i,D_i) ) \right)
	\right] \phi_{ni}^2
\right\}^2
\\
 & \hspace{-3.5cm}= 
M \mme\left\{
	(K n \overline{h}_1 ) \sum_{i=1}^n \left( 
		\overline{\nu}_i - \nu_i \right)
		\phi_{ni}^2
\right\}^2
\\
 & \hspace{-3.5cm}= 
M ( K n \overline{h}_1 )^2
\mme\left[
	\sum_{i=1}^n \sum_{j=1}^n (\mmi_{ij}^{=} + \mmi_{ij}^{\neq}) \left( \overline{\nu}_i - \nu_i \right) \phi_{ni}^2
		\left( \overline{\nu}_j - \nu_j \right)  \phi_{nj}^2 		
\right]=o_P(1)
\end{align}
where the
expected value of $\mmi_{ij}^{\neq} \left( \overline{\nu}_i - \nu_i \right) \phi_{ni}^2 \left( \overline{\nu}_j - \nu_j \right) \phi_{nj}^2$
conditional on $\m{X}_n$ is zero because
if $\mmi_{ij}^{\neq} =1$, then $\left( \overline{\nu}_i - \nu_i \right) $ is independent of 
$\left( \overline{\nu}_j - \nu_j \right) $,
and  
$
\mme[ \overline{\nu}_i - \nu_i   | \m{X}_n ] = 0$;
and the rest follows arguments similar to the ones used in part \eqref{eq:lemma:est:var:c:part3:1}.
The Chebyshev's inequality yields that $\eqref{eq:lemma:est:var:c:part3:2} = o_P(1)$.

\bigskip

\textbf{Part \eqref{eq:lemma:est:var:c:part3:3}:} the expected value $\mme[\eqref{eq:lemma:est:var:c:part3:3} | \m{X}_n]=0$
because $\mme [ \eps_{\ell(i,l)} \eps_{\ell(i,v)}  | \m{X}_n ]$ $=\mme [ \eps_{\ell(i,l)} | \m{X}_n ]$ $ \mme [ \eps_{\ell(i,v)}  | \m{X}_n ]$ $= 0$.
\begin{align}
\mme\left[
	\left( \eqref{eq:lemma:est:var:c:part3:3} -  \mme[\eqref{eq:lemma:est:var:c:part3:3} | \m{X}_n] \right)^2
\right] &
\notag \\
  & \hspace{-4.5cm} = M \mme
\left\{
	(K n \overline{h}_1 ) \sum_{i=1}^n \left\{ 
		\sum_{v > l}^N \eps_{\ell(i,l)} \eps_{\ell(i,v)}
	\right\} \phi_{ni}^2
\right\}^2
\\
  &\hspace{-4.5cm} = M \mme
\left\{
	(K n \overline{h}_1 ) \sum_{i=1}^n \sum_{j=1}^n (\mmi_{ij}^{=} + \mmi_{ij}^{\neq})
	\left\{ 
		\sum_{v > l}^N \eps_{\ell(i,l)} \eps_{\ell(i,v)}
	\right\} \phi_{ni}^2
	\left\{ 
		\sum_{v > l}^N \eps_{\ell(j,l)} \eps_{\ell(j,v)}
	\right\} \phi_{nj}^2
\right\} = o_P(1)
\end{align}
where the
expected value of 
$\mmi_{ij}^{\neq} \sum_{v > l} \eps_{\ell(i,l)} \eps_{\ell(i,v)} \phi_{ni}^2
										 \sum_{v > l} \eps_{\ell(j,l)} \eps_{\ell(j,v)} \phi_{nj}^2$
conditional on $\m{X}_n$ is zero because
if $\mmi_{ij}^{\neq} =1$, then $\eps_{\ell(i,l)} \eps_{\ell(i,v)} $ is independent of 
$\eps_{\ell(j,l)} \eps_{\ell(j,v)}$,
and  
$
\mme[ \eps_{\ell(i,l)} \eps_{\ell(i,v)}   | \m{X}_n ] = 0$; the rest follows arguments similar to the ones used in part \eqref{eq:lemma:est:var:c:part3:1}.
The Chebyshev's inequality yields that $\eqref{eq:lemma:est:var:c:part3:3} = o_P(1)$.

\bigskip

\textbf{Part \eqref{eq:lemma:est:var:c:part3:4}:} for $x^*_{il}$ between $X_i$ and $X_{\ell(i,l)}$:
\begin{align}
\eqref{eq:lemma:est:var:c:part3:4} = &
M (K n \overline{h}_1 ) \sum_{i=1}^n 
	\left( \sum_{l=1}^N  \nabla_x R(x^*_{il},D_i) (X_i - X_{\ell(i,l)}) \right)^2
\phi_{ni}^2
\\
= & (K n \overline{h}_1 ) n O(K^{-2}) O\left(( K n \overline{h}_1 )^{-2}  \right)
= O(K^{-2}) O\left(( K  \overline{h}_1 )^{-1}\right) = o(1)
\end{align}
which uses the fact that $\nabla_x R(x^*_{il},D_i) $ is bounded (Assumption \ref{assu_srd_est_mx}),
that $(X_i - X_{\ell(i,l)}) = O(K^{-1})$,
that 
$\phi_{ni}$ is $O\left(( K n \overline{h}_1 )^{-1}  \right)$ (Equation \ref{eq:lemma:est:var:c:phi}),
and the rate condition $( K  \overline{h}_1 )^{-1}  = O(1)$.

\bigskip

\textbf{Part \eqref{eq:lemma:est:var:c:part3:5}:} for $x^*_{il}$ between $X_i$ and $X_{\ell(i,l)}$:
\begin{align}
\eqref{eq:lemma:est:var:c:part3:5} = & M (K n \overline{h}_1 ) \sum_{i=1}^n \left\{ 
	\eps_i  \sum_{l=1}^N \nabla_x R(x^*_{il},D_i) (X_i - X_{\ell(i,l)})
\right\} \phi_{ni}^2
\\
= & (K n \overline{h}_1 ) n O_P(K^{-1})  O\left(( K  \overline{h}_1 )^{-1}\right) = o_P(1)
\end{align}
which follows from the same arguments as the ones in part \eqref{eq:lemma:est:var:c:part3:4}
plus the fact that $\eps_i$ is a.s. bounded (Assumption \ref{assu_srd_kinf_est_int}).

\bigskip

\textbf{Part \eqref{eq:lemma:est:var:c:part3:6}:} for $x^*_{iv}$ between $X_i$ and $X_{\ell(i,v)}$:
\begin{align}
\eqref{eq:lemma:est:var:c:part3:6} = & 
M (K n \overline{h}_1 ) \sum_{i=1}^n \left\{
	\sum_{l=1}^N \eps_{\ell(i,l)} \sum_{v=1}^N \nabla_x R(x^*_{iv},D_i) (X_i - X_{\ell(i,v)})
\right\} \phi_{ni}^2
\\
= & (K n \overline{h}_1 ) n O_P(K^{-1})  O\left(( K  \overline{h}_1 )^{-1}\right) = o_P(1)
\end{align}
as seen in part  \eqref{eq:lemma:est:var:c:part3:5}.

\bigskip

Therefore, $\eqref{eq:lemma:est:var:c:part3}=o_P(1)$, which concludes the proof.

\end{proof}

\subsection{Fuzzy RDD with Multiple Cutoffs}

\subsubsection{Example of Compliance Behaviors}\label{sec_fuzzy_table}

\indent

Here is a simple example
with three different treatments and two cutoffs (that is, 3 schools, $K=2$)
to illustrate the different compliance behaviors.
Table \ref{table_compliance} below
lists all possible combinations of treatment eligibility and assignment
produced by ${\m{U}}_i(x)$.
\begin{center}
\begin{table}[H]
    \caption{Different Compliance Behaviors}
    \label{table_compliance}
    \begin{center}
        \begin{multicols}{2}
                        \begin{tabular}{|ccc|c|}
              \hline
                \multicolumn{3}{|c|}{\textbf{Eligibility}} & \multirow{2}{*}{\textbf{Type}}  \\
                $\bm{d_0}$ & $\bm{d_1}$ & $\bm{d_2}$ &  \\
                \hline
                     $d_0$   &   $d_0$   &   $d_1$   &  \multirow{13}{*}{ever-defiers} \\
                      $d_0$   &   $d_1$   &   $d_0$   & \\
                      $d_0$   &   $d_2$   &   $d_0$   & \\
                      $d_0$   &   $d_2$   &   $d_1$   & \\
                      $d_0$   &   $d_2$   &   $d_2$   & \\
                      $d_1$   &   $d_0$   &   $d_0$   & \\
                      $d_1$   &   $d_0$   &   $d_1$   & \\
                      $d_1$   &   $d_0$   &   $d_2$   & \\
                      $d_1$   &   $d_1$   &   $d_0$   & \\
                      $d_1$   &   $d_2$   &   $d_0$   & \\
                      $d_1$   &   $d_2$   &   $d_1$   & \\
                      $d_1$   &   $d_2$   &   $d_2$   & \\
                      $d_2$   &   $d_0$   &   $d_0$   & \\
                      $d_2$   &   $d_0$   &   $d_1$   & \\
                      $d_2$   &   $d_0$   &   $d_2$   & \\
                      $d_2$   &   $d_1$   &   $d_0$   & \\
                      $d_2$   &   $d_2$   &   $d_0$   & \\
                      $d_2$   &   $d_2$   &   $d_1$   & \\
                \hline
                \end{tabular}

                \begin{tabular}{|ccc|c|}
                \hline
                \multicolumn{3}{|c|}{\textbf{Eligibility}} & \multirow{2}{*}{\textbf{Type}}  \\
                $\bm{d_0}$ & $\bm{d_1}$ & $\bm{d_2}$ &  \\
                \hline
                      $d_0$   &   $d_0$   &   $d_0$   & \multirow{3}{*}{never-changers}\\
                      $d_1$   &   $d_1$   &   $d_1$   & \\
                      $d_2$   &   $d_2$   &   $d_2$   & \\
                \hline
                      $d_0$   &   $d_0$   &   $d_2$   & \multirow{6}{*}{ever-compliers}\\
                      $d_0$   &   $d_1$   &   $d_1$   & \\
                      $d_0$   &   $d_1$   &   $d_2$   & \\
                      $d_1$   &   $d_1$   &   $d_2$   & \\
                      $d_2$   &   $d_1$   &   $d_1$   & \\
                      $d_2$   &   $d_1$   &   $d_2$   & \\
                \hline
                \end{tabular}

        \end{multicols}
        \caption*{
            \footnotesize
            Notes: All possible realizations of the random function ${\m{U}}_i(x)$ for values of $x$
            such that $D(x) \in \{ d_0, d_1, d_2 \}$.
        }

    \end{center}
\end{table}
\end{center}

\subsubsection{Estimation and Inference}\label{sec_fuzzy_infer}

\indent

Theorem \ref{theo_frd_id} in the main text suggests a two-step estimation procedure for $\boup{\theta}_0^{ec}$.
In the first step, obtain $\widehat{B}_{j}$ as in Section \ref{sec_case1},
and compute estimates $\widehat{\widetilde{W}}_{j}$ using
LPRs of $\boup{\m{W}}(X_i,D_i)$ on $X_i$ at each side of the cutoff
$c_{j}$.
For each  $j=1, \ldots, K$, and $l$-th coordinate of  the vector $\ti{W}_{j}$, $l=1,\ldots,q$, 
the researcher computes
\begin{align}
\widehat{\widetilde{W}}_{j,l} = & \hat{a}_{j,l}^{+} - \hat{a}_{j,l}^{-}
\label{eq_est_omega_pjl}
\\
(\hat{a}_{j,l}^{+}, \hat{\boup{b}}_{j,l}^{+} )
= & \argmin\limits_{(a, \boup{b} )}
\sum\limits_{i=1}^n  \Bigg\{ k \left( \frac{X_i - c_{j}}{h_{1j} } \right) v_{i}^{j+} 
\nonumber
\\
& \hspace{2.5cm} \bigg[ e_l' \boup{ \m{W}}(X_i,D_i) - a - b_1 (X_i - c_{j}) - \ldots -b_{\rho_1} (X_i - c_{j})^{\rho_1} \bigg]^2
\Bigg\}
\label{eq_est_lprd_r}
\\
(\hat{a}_{j,l}^{-}, \hat{\boup{b}}_{j,l}^{-} )
= & \argmin\limits_{(a, \boup{b} )}
\sum\limits_{i=1}^n \Bigg\{ k \left( \frac{X_i - c_{j}}{h_{1j} } \right) v_{i}^{j-} 
\nonumber
\\
& \hspace{2.5cm}
\bigg[ e_l' \boup{\m{W}}(X_i,D_i) - a - b_1 (X_i - c_{j}) - \ldots -b_{\rho_1} (X_i - c_{j})^{\rho_1} \bigg]^2 \Bigg\}
\label{eq_est_lprd_l}
\end{align}
where $e_l$ is the $q \times 1$ vector of zeros except for 1 in its $l$-th coordinate.
The $q \times 1$ vector $\widehat{\widetilde{W}}_{j}$ is constructed by stacking the $q$ estimates
$\widehat{\widetilde{W}}_{j} = \left[\widehat{\widetilde{W}}_{j,1}, \ldots, \widehat{\widetilde{W}}_{j,q} \right]' $.

In the second step, regress $\ha B_{j}$ on $\widehat{\widetilde{W}}_{j}$ to obtain an estimate for
$\boup{\theta}_0^{ec}$.
More specifically, stack all $q \times 1 $ vectors
$\widehat{\widetilde{W}}_{j}$ into the $K \times q$ matrix $\widehat{\widetilde{\boup{W}}}$, and $\ha B_{j}$ into the
$K \times 1 $ vector $\widehat{\boup{B}}$.
Choose a $K \times K$ symmetric and
positive-definite weighting matrix $\Omega$.
The estimator $\ha{\boup{\theta}}^{ec}$ is the solution to the following weighted least-squares problem:
\begin{align}
\widehat{\boup{\theta}}^{ec} = & \argmin_{\boup{\theta}}
\left( \widehat{\boup{B}} - \widehat{\widetilde{\boup{W}}} \boup{\theta} \right)'
\Omega
\left( \widehat{\boup{B}} - \widehat{\widetilde{\boup{W}}} \boup{\theta} \right).
\label{eq_est_frd_argmin}
\end{align}
The estimator for the ATE on ever-compliers $\mu^{ec}$ is a linear combination of $\widehat{\boup{\theta}}^{ec}$,
\begin{gather}
\widehat{\mu}^{ec} = \boup{Z}(F) \widehat{\boup{\theta}}^{ec}
\label{est:mu:ec}
\end{gather}
where $\boup{Z}(F)$ is defined in Equation \ref{eq_mu3}.

Asymptotic normality of $\ha{\btheta}^{ec}$ relies on smoothness assumptions
on the conditional moments of $Y_i$ and the probabilities of treatment
for different compliance behaviors.
The sample size grows large, while the number of cutoffs remains fixed.
\loadaspt{00I}
\begin{theorem}\label{theo_frd_est}
Suppose Assumptions  \ref{assu_srd_est_kernel}-\ref{assu_srd_est_fx} and \ref{assu_frd_id_nodef}-\ref{assu_frd_est_mx} hold, and 
that the number of cutoffs $K$ is fixed.
Let $\underline{h}_1 = \min_j\{  h_{1j}  \}$ and
$\overline{h}_1 = \max_j\{  h_{1j} \}$.
As $n\to\infty$, assume that $\overline{h}_1 \to 0$,
$\overline{h}_1/\underline{h}_1 = O(1)$,
$n \overline{h}_1 \to \infty$,
and
$(n \overline{h}_1)^{1/2} \overline{h}_1^{\rho_1 + 1} = O(1)$.
Then,
\begin{align}
(\bmV_n^{\btheta^{ec}})^{-1/2}\left( \ha{\btheta}^{ec} - \bmB_n^{\btheta^{ec}} - \btheta_{0}^{ec} \right)
&\dto N(\bzero,\bI)
\\
\frac{ \ha{\mu}^{ec} - \m{B}_n^{\mu^{ec}} - \mu^{ec} }{ (\m{V}_n^{\mu^{ec}})^{1/2} }
&\dto N(0,1),
\end{align}
where $\bzero$ denotes the $q \times 1 $ vector of zeros, 
and $\bI$ is the $q \times q$ identity matrix.

The bias and variance terms are characterized as follows,
\begin{align}
\bmB_n^{\btheta^{ec}} = & \left(  \ha{\bWt}' \Omega \ha{\bWt} \right)^{-1}
 \ha{\bWt}' \Omega  \bmB_n^{ec} 
 \text{, a } q \times 1 \text{ vector;}
\label{def:bias:theta:ec}
\\
\m{B}_n^{\mu^{ec}} = & \bZ (F) \left( \ha{\bWt}' \Omega \ha{\bWt} \right)^{-1}
 \ha{\bWt}' \Omega  \bmB_n^{ec} 
\text{, a scalar; }
\label{def:bias:mu:ec}
\\
\bmV_n^{\btheta^{ec}} =& 
\left( \ha{\bWt}' \Omega \ha{\bWt} \right)^{-1}
\ha{\bWt}' \Omega  \bmV_n^{ec}
\Omega \ha{\bWt} 
\left( \ha{\bWt}' \Omega \ha{\bWt} \right)^{-1} 
\text{, a } q \times q \text{ matrix;}
\label{def:var:theta:ec}
\\
\m{V}_n^{\mu^{ec}} =&
  \bZ (F) \left( \ha{\bWt}' \Omega \ha{\bWt} \right)^{-1}
   \ha{\bWt}' \Omega  \bmV_n^{ec}
   \Omega \ha{\bWt} \left( \ha{\bWt}' \Omega \ha{\bWt} \right)^{-1} \bZ(F)'
   \text{, a scalar.}
\label{def:var:mu:ec}
\end{align}
These terms depend on $\bmB_n^{ec}$ ($K \times 1$ vector) 
and $\bmV_n^{ec}$ ($K \times K$ matrix) that are defined below
\begin{align}
\bmB_n^{ec} = & \left[ \m{B}_{n1}^{ec},\ldots, \m{B}_{nK}^{ec} \right]' \text{, where for each }j
\\
& \m{B}_{nj}^{ec} =  
			\frac{ h_{1j}^{\rho_1+1}  f(c_j)} { (\rho_1 + 1 )! } 
			\left[1 ~~ -{\btheta_0^{ec}}^{'} \right]
\notag\\*
& \hspace{3cm}
			\be_1'  \Bigg\{ 
				   \bG_n^{j +}  \bgam^* \nabla^{\rho_1 + 1 }_x 
				   \left[
					\begin{array}{c}
						R(c_j,d_{j})
					\\
					\bmW(c_j,d_{j})
					\end{array}
					\right]
				- 
				\bG_n^{j -}  \bgam^* \nabla^{\rho_1 + 1 }_x
				\left[
				\begin{array}{c}
					R(c_j,d_{j-1})
				\\
					\bmW(c_j,d_{j-1})
				\end{array}
				\right]
			\Bigg\};
\label{def:bias:ec:j}	
\\
\bmV_n^{ec} = &  \left[
	\begin{array}{ccc}
		\m{V}_{n11}^{ec} & \ldots & \m{V}_{n1K}^{ec}
		\\
		\vdots & \ddots & \vdots
		\\
		\m{V}_{nK1}^{ec} & \ldots & \m{V}_{nKK}^{ec}
	\end{array}
\right]
\text{, where } \m{V}_{njl}^{ec}=0 \text{ if } |j-l|> 1, otherwise
\\
\notag\\
& 
\m{V}_{njl}^{ec} = 
n
\mme \Bigg\{
	\frac{1}{n h_{1j}}
	k\left( \frac{X_i - c_{j}}{h_{1j}} \right)
	\frac{1}{n h_{1l}}
	k\left( \frac{X_i - c_{l}}{h_{1l}} \right)
\notag \\
& \hspace{2cm}	
	\left[1 ~~ -{\btheta_0^{ec}}^{'} \right]  \be_1'
	\left(
		v_i^{j +} \mme[ \bG_n^{j +} ]  
		-
		v_i^{j -} \mme[ \bG_n^{j -} ]  
	\right)
	\bHt_{i}^{j}
	\beps_i
\notag \\
& \hspace{2cm}
	\beps_i'
	\bHt_{i}^{l'}
	\left(
		v_i^{l +} \mme[ \bG_n^{l +} ]'  
		-
		v_i^{l -} \mme[ \bG_n^{l -} ]'  
	\right)
	\be_1 \left[1 ~~ -{\btheta_0^{ec}}^{'} \right]'   
\Bigg\},
\label{def:var:ec:jl}	
\end{align}
where
$\beps_i= \left[ Y_i ~~ \bmW(X_i,D_i)' \right]'
-
\mme \left\{ \left[ Y_i ~~ \bmW(X_i,D_i)' \right]' ~|~X_i \right\}
$, a $(q+1) \times 1$ vector;
$\be_1=  \bI_q \otimes e_1$ where
$\bI_q$  is the $q \times q$ identity matrix,
$ \otimes$  denotes the Kronecker product,
and $e_1$ is a $(\rho_1+1)\times 1$ vector of zeros except for the first coordinate that equals
$1$;
$\bgam^* = \bI_q \otimes \gamma^*$ for  $\gamma^*$ defined in Theorem \ref{theo_srd_est_avg};
$\bHt_i^{j} =  \bI_q \otimes H( ( X_i - c_{j} ) / h_{1j}   )$ for $H(u)$
defined in Theorem \ref{theo_srd_est_avg};
and
$\bG_n^{j \pm} = [ (n h_{1j} )^{-1} \sum_{i=1}^n k\left( \frac{X_i - c_{j}}{ h_{1j} } \right)
v_i^{j \pm} \bHt_{i}^{j }  \bHt_{i}^{j'}  ]^{-1}$
for  $v_i^{j \pm}$ defined in Equation \ref{def_vij}.

Furthermore, 
$(\bmV_n^{\btheta^{ec}})^{-1/2} =  O_P \left( \left( n \overline{h}_1 \right)^{1/2}  \right)$, and
$(\bmV_n^{\btheta^{ec}})^{-1/2} \bmB_n^{\btheta^{ec}}  =  O_P\left( \left( n \overline{h}_1\right)^{1/2} \overline{h}_1^{\rho_1 + 1}  \right)$,
where $A^{-1/2}$ denotes the inverse of the square root of a positive-definite matrix $A$.
Similarly,
$(\m{V}_n^{\mu^{ec}})^{-1/2} =  O_P \left( \left( n \overline{h}_1 \right)^{1/2}  \right)$, and
$(\m{V}_n^{\mu^{ec}})^{-1/2} \m{B}_n^{\mu^{ec}}  =  O_P \left( \left( n \overline{h}_1\right)^{1/2} \overline{h}_1^{\rho_1 + 1}  \right)$.

The approximate MSE of either $ \ha{\btheta }^{ec}$
or
$\ha{\mu}^{ec}$
is minimized by setting $\Omega$ $= \Big( \bmB_n^{{ec}} \bmB_n^{{ec'} } + \bmV_n^{{ec}} \Big)^{-1}$.
\end{theorem}

The proof of Theorem \ref{theo_frd_est} is in Section \ref{sec_fuzzy_proof} of the supplemental appendix. 
The variance terms in Equations \ref{def:var:theta:ec} and \ref{def:var:mu:ec}
contain the matrix $\bmV^{ec}_{n}$ that needs to be estimated.
The elements  $\m{V}_{njl}^{ec}$ of such matrix are consistently estimated by
\begin{align}
\ha{\m{V}}_{njl}^{ec} =& 
\sum_{i=1}^n \Bigg\{
	\frac{1}{n h_{1j}}
	k\left( \frac{X_i - c_{j}}{h_{1j}} \right)
	\frac{1}{n h_{1l}}
	k\left( \frac{X_i - c_{l}}{h_{1l}} \right)
\notag \\
& \hspace{2cm}	
	\left[1 ~~ -{\ha{\btheta}^{ec'}} \right]  \be_1'
	\left(
		v_i^{j +} \bG_n^{j +}  
		-
		v_i^{j -} \bG_n^{j -}  
	\right)
	\bHt_{i}^{j}
	\ha{\beps}_i
\notag \\
& \hspace{2cm}
	\ha{\beps}_i'
	\bHt_{i}^{l'}
	\left(
		v_i^{l +} \bG_n^{l +'}  
		-
		v_i^{l -} \bG_n^{l -'}  
	\right)
	\be_1 \left[1 ~~ -{\ha{\btheta}^{ec'}} \right]'   
\Bigg\}.
\label{est:var:ec:jl}
\end{align} 
where $\ha{\btheta}^{ec}$ is a consistent estimator of $\btheta^{ec}_0$,
and the vector of residuals is estimated by a nearest-neighbor matching estimator 
analogously to Section \ref{sec_case1}'s Equation \ref{est:resid} : 
\begin{align}
\ha{\beps}_i \ha{\beps}_i'  = & \frac{3}{4} 
\left(
	\bY_i  - \frac{1}{3} \sum_{l =1}^3 \bY_{\ell(i,l)}
\right)
\left(
	\bY_i  - \frac{1}{3} \sum_{l =1}^3 \bY_{\ell(i,l)}
\right)'.
\label{est:resid:fuzzy}
\end{align}
If the bandwidth choices are such that the standardized bias term $\left(\bmV^{\btheta^{ec}}_{n} \right)^{-1/2} \bmB^{\theta^{ec}}_{n}$ differs from  zero asymptotically, then inference must be done using a bias-corrected estimator. 
A practical way of doing bias correction  is to increase the order of the polynomial to $\rho_1+1$ and compute 
$\ha{\btheta}^{ec'}$
and
$\ha{\bmV}^{ec'}_n$
using the same bandwidth choices.
It follows that 
$(\ha{\bmV}_n^{\btheta^{ec'}})^{-1/2}\left( \ha{\btheta}^{ec'} - \btheta_{0}^{ec} \right)
\dto N(\bzero,\bI)$.
Similar to Theorems \ref{theo_srd_est_avg} and \ref{theo_srd_kinf_est_int},
Theorem \ref{theo_frd_est} allows for bandwidth choices that produce overlapping estimation windows across cutoffs.
The variance estimator in \eqref{est:var:ec:jl}
takes account of overlap by allowing $\m{V}_{njl}$ to be non-zero for $j \neq l$.

The following steps are a practical recommendation to implement 
MSE-optimal and bias-corrected estimates.
The source of MSE in estimation is 
$\m{B}_{nj}^{ec}$ and $\m{V}_{njl}^{ec}$, which come
from the regression of
$Y_i - \bmW(X_i,D_i)' {\btheta}^{ec} $ on $X_i$ 
at each cutoff $c_j$; thus, it makes sense to choose MSE-optimal bandwidths for these regressions.
\begin{enumerate}
   \item[0.] Take initial values $\ha{\btheta}^{ec(0)}$ and $\Omega^{(0)}$;
	\item Compute first-step IK bandwidths $h_{1j}^{(0)}$ 
	for sharp RD of 
	$Y_i - \bmW(X_i,D_i)' {\btheta}^{ec(0)} $
	on $X_i$ at each cutoff $c_j$.
	Use local-linear regression ($\rho_1=1$) and edge kernel;
	\item Obtain bias-corrected estimates $\ha{B}_j^{(0)}$ 
	for each cutoff $j$ using sharp RD of $Y_i$ on $X_i$
	using local-quadratic regression, edge kernel, and bandwidth $h_{1j}^{(0)}$;
	do the same for each coordinate of $\bmW(X_i,D_i)$ to compute
	$\ha{\ti{W}}_j^{(0)}$ for each $j$;
	stack estimates into $\ha{\bB}^{(0)}$ and  $\ha{\bWt}^{(0)}$;
	\item Update $\ha{\btheta}^{ec}$:
	compute $\ha{\btheta}^{ec(1)}$ using Equation \ref{eq_est_frd_argmin}
	with $\Omega^{(0)}$, $\ha{\bB}^{(0)}$, and  $\ha{\bWt}^{(0)}$;
	\item  Estimate the variance of 
	 $\ha{\bB}^{(0)} - \ha{\bWt}^{(0)} \theta^{ec}_0$  
	 using Equation \ref{est:var:ec:jl}
	 with  $\rho_1=2$, $\ha{\btheta}^{ec(1)}$, and
	 bandwidths $h_{1j}^{(0)}$;
	 call the estimated variance  $\ha{\bmV}^{ec(1)}_n$.

\item Update $\Omega$:
compute $\Omega^{(1)} = \Big( \ha{\bmV}^{ec(1)}_n \Big)^{-1}$;

\item Update $\ha{\btheta}^{ec}$:
	 compute $\ha{\btheta}^{ec(2)}$
	  using Equation \ref{eq_est_frd_argmin} 
	    with $\Omega^{(1)}$, $\ha{\bB}^{(0)}$,  and $\ha{\bWt}^{(0)}$;	

\item Repeat Steps 4-6 starting with $\ha{\btheta}^{ec(2)}$ in the place of
$\ha{\btheta}^{ec(1)}$. Iterate these three steps until convergence of $\ha{\btheta}^{ec}$.
Call $\ha{\btheta}^{ec(3)}$  and $\Omega^{(3)}$,
the iterated values of, respectively, 
$\ha{\btheta}^{ec}$ and $\Omega$;

\item Repeat Steps 1-7 starting with $\ha{\btheta}^{ec(3)}$ in the place of
$\ha{\btheta}^{ec(0)}$, and with $\Omega^{(3)}$ in the place of $\Omega^{(0)}$.
Iterate these 7 steps until the difference between the $\btheta$s of Step 3 and Step 7
converges to zero.
Call $\ha{\btheta}^{ec(4)}$,  
 $\Omega^{(4)}$,
 and  $\ha{\bWt}^{(4)}$
the iterated values of, respectively, 
$\ha{\btheta}^{ec}$, 
 $\Omega^{}$,
 and  $\ha{\bWt}$;

\item Estimate the variance of $\ha{\btheta}^{ec}$
using 
$\ha{\bmV}_n^{\btheta^{ec}} = 
\left( \ha{\bWt}^{(4)'} \Omega^{(4)} \ha{\bWt}^{(4)} \right)^{-1}$;
compute $\widehat{\mu}^{ec}$ using Equation \ref{est:mu:ec}
with $\widehat{\boup{\theta}}^{ec(4)}$
and $\bZ (F)$ given by the counterfactual policy of interest;
estimate the variance of $\widehat{\mu}^{ec}$
using 
$\ha{\bmV}_n^{\mu^{ec}} =\bZ (F) \ha{\bmV}_n^{\btheta^{ec}} \bZ (F)' $.

\end{enumerate}

\subsubsection{Proof of Theorem \ref{theo_frd_est}}\label{sec_fuzzy_proof}

\indent 

This proof relies heavily on Lemma \ref{lemma_porter}, which is a CLT for the LPR estimator of the
difference in side-limits of a conditional mean function of the
vector $\bY_i$ given $X_i$ at $X_i=c_j$.
Such lemma is applied to $\bY_i = [ Y_i ~~ \boup{\m{W}}(X_i,D_i)']'$ to arrive at 
\begin{align}
(\bmV_{nj} )^{-1/2} \left(\ha{\bJ}_{j} - \bmB_{nj} - \bJ_{j} \right) \dto N\left(\bzero ; \bI  \right)
\label{eq:theo_frd_est:cltj}
\end{align}
for each $j$.
 
The assumptions of Theorem \ref{theo_frd_est} satisfy the assumptions of Lemma \ref{lemma_porter}.
In fact, the conditions on the  rates, on the distribution of $X_i$, and on the kernel density in Lemma \ref{lemma_porter}
are simply restated in the conditions of Theorem \ref{theo_frd_est}.
It remains to verify the other two sufficient conditions of Lemma \ref{lemma_porter}:
(a) $\bmm(x)$ has continuous derivatives
wrt $x$ of order $\rho_1+1$ in a compact interval centered at $c_{j}$ but excluding $c_{j}$, and existence of side limits
at $c_{j}$; and
(b) continuity of $\bzeta(x)$
wrt $x$ in a compact interval centered at $c_{j}$ but excluding $c_{j}$, existence of side limits
at $c_{j}$, and boundedness of the third moment conditional on $X_i$.

For (a), note that, in the fuzzy case, the mean of $Y_i$ and $\bmW(X_i,D_i)$ conditional on $X_i$
is a sum of the means of potential outcomes $Y_i(d)$ and
$\bmW(c_{j},d)$ for various dosages $d$ conditional on sets of the form $\{\m{U}_i(c_j)=d \} $
weighted by conditional probabilities of the same sets  (see proof of Theorem \ref{theo_frd_id}).
Assumption \ref{assu_frd_est_mx}
implies that such conditional means and conditional probabilities are
smooth functions of $x$ and side-limits exist at $x=c_j$.
Similarly, for (b), the conditional covariance of $[ Y_i,~~ \bmW(X_i,D_i)' ]'$
is a function of sums of the first and second moments
of potential outcomes $Y_i(d)$ for various dosages
conditional on sets of the form $\{\m{U}_i(c_j)=d \} $
weighted by conditional probabilities of the same sets.
Assumption \ref{assu_frd_est_mx}
ensures continuity of $\bzeta(x)$ wrt $x$ and existence of side-limits at $c_{j}$.
A similar argument bounds the third centered moment of  $Y_i(d)^{3}$,
and Lemma \ref{lemma_porter} applies.

Next, note that
\begin{align} 
\lim_{e \downarrow 0} & \Bigg\{
	\mme \left[\boup{\m{W}}(c_{j},D_i) | X_i=c_{j}+e \right]
	-
	\mme \left[\boup{\m{W}}(c_{j},D_i) | X_i=c_{j}-e \right]
\bigg\}
\\
=& \sum_{l=0,l \neq j}^{K} \Big\{
	\boup{\m{W}}(c_{j},d_{j})
	-
	\boup{\m{W}}(c_{j},d_{l})
\Big\}
\omega_{j,l}
 = \ti{W}_{j}
\end{align}
which means that $\bJ_{j} = [B_{j} ~~ \ti{W}_{j}']'$.
Call $\balpha = [1 ~~ -{\btheta_0^{ec}}']'$, a $( (q+1)\times 1)$ vector. 
Then, \eqref{eq:theo_frd_est:cltj} implies
\begin{align}
&\left( \balpha' \bmV_{nj} \balpha \right)^{-1/2}
\left( \balpha' \ha{\bJ}_{j} - \balpha' \bmB_{nj} - \balpha' \bJ_{j} \right)
\dto N(0,1)
\\
&
\left( \m{V}_{njj}^{ec} \right)^{-1/2}
\left(  \ha{B}_{j} - {\btheta_0^{ec}}' \ha{\ti{W}}_{j} - \m{B}_{nj}^{ec} \right)
\dto N(0,1)
\end{align}
where $\balpha' \bJ_j=0$ by Assumption \ref{assu_srd_id_param},
and the definitions \eqref{def:bias:ec:j} and \eqref{def:var:ec:jl} are used.
Stacking across cutoffs gives
\begin{align}
\left( \diag\{ \m{V}_{njj}^{ec} \}_j\right)^{-1/2}
\left(  \ha{\bB} -  \ha{\bWt} \btheta_0^{ec} - \bmB_{n}^{ec} \right)
\dto N(\bzero,\bI)
\\
\left(  \bmV_{n}^{ec} \right)^{-1/2}
\left(  \ha{\bB} -  \ha{\bWt} \btheta_0^{ec} - \bmB_{n}^{ec} \right)
\dto N(\bzero,\bI)
\end{align}
where $\left(  \bmV_{n}^{ec} \right)^{1/2}\left( \diag\{ \m{V}_{njj}^{ec} \}_j\right)^{-1/2} \to \bI$
because the covariances (off-diagonal terms) converge to zero since the estimation windows do not overlap in the limit.
Define $\bGam = \left( {\bWt}' \Omega {\bWt} \right)^{-1} {\bWt}' \Omega$. 
Then,
\begin{align}
\left( \bGam \bmV_{n}^{ec} \bGam' \right)^{-1/2}
\left(  \bGam \ha{\bB} -  \bGam \ha{\bWt} \btheta_0^{ec} - \bGam \bmB_{n}^{ec} \right)
\dto N(\bzero,\bI)
\label{eq:theo_frd_est:cltj:vec}
\end{align}

Define $\ha{\bGam} = \left( \ha{\bWt}' \Omega \ha{\bWt} \right)^{-1} \ha{\bWt}' \Omega$,
and write
\begin{align}
(\bmV_n^{\btheta^{ec}})^{-1/2}
\left( \ha{\btheta}^{ec} - \bmB_n^{\btheta^{ec}} - \btheta_{0}^{ec} \right)
&= 
\left( \ha{\bGam} \bmV_{n}^{ec} \ha{\bGam}' \right)^{-1/2}
\left(  \ha{\bGam} \ha{\bB} -  \ha{\bGam} \ha{\bWt} \btheta_0^{ec} - \ha{\bGam} \bmB_{n}^{ec} \right) 
\\
&\hspace{-2cm}= 
\left( \ha{\bGam} \bmV_{n}^{ec} \ha{\bGam}' \right)^{-1/2}
 \left(  \bGam \ha{\bB} -  \bGam \ha{\bWt} \btheta_0^{ec}  - {\bGam} \bmB_{n}^{ec} \right)
\\
& \hspace{-2cm} +  \left( \ha{\bGam} \bmV_{n}^{ec} \ha{\bGam}' \right)^{-1/2}
\left( \ha{\bGam} - \bGam \right) \left( \ha{\bB} - \ha{\bWt} \btheta_0^{ec} -  \bmB_{n}^{ec} \right)
\\
&\hspace{-2cm}= 
\left( \ha{\bGam} \bmV_{n}^{ec} \ha{\bGam}' \right)^{-1/2}
\left( \bGam \bmV_{n}^{ec} \bGam' \right)^{1/2}
\left( \bGam \bmV_{n}^{ec} \bGam' \right)^{-1/2}
 \left(  \bGam \ha{\bB} -  \bGam \ha{\bWt} \btheta_0^{ec}  - \bGam \bmB_{n}^{{ec} } \right)
\\
& \hspace{-2cm} + O_P\left( \left( n h_1 \right)^{1/2} \right) 
o_P(1)
O_P\left( \left( n h_1 \right)^{-1/2} \right) 
\end{align}
which converges in distribution to $N\left(\bzero, \bI \right)$
because of \eqref{eq:theo_frd_est:cltj:vec},
the fact that
$\left( \ha{\bGam} \bmV_{n}^{ec} \ha{\bGam}' \right)^{-1/2}
\left( \bGam \bmV_{n}^{ec} \bGam' \right)^{1/2} \pto \bI$,
and that $\ha{\bGam} \pto {\bGam}$.

$\square$

\subsection{Estimation of Counterfactual Distributions}

\indent 

This section considers applications where the counterfactual distribution is estimated as opposed to being known by the researcher.
For brevity, I focus on the setting of Section \ref{sec_case1}, that is, sharp RD with discrete counterfactual and fixed $K$.
The analysis for the other settings of the paper follows similar arguments. 
In what follows, I derive the limiting distribution for $\ha{\mu}^d$ and propose a consistent variance estimator.

In the first step, the researcher estimates the counterfactual probability mass function $\omega^d(\bc)$ for every $\bc \in \m{C}_K$
using iid observations $Z_i=(Y_i,X_i)$, $i=1,\ldots, n$.
There is a variety of ways to obtain estimates  for $\ha{\omega}^d(\bc)$. 
For example, one may estimate the distribution of $X_i$ non-parametrically, obtain $\ha{f}_X(c_j)$ for every $j$, and construct
$\ha{\omega}^d_j = \ha{f}_X(c_j) / \sum_{l=1}^K \ha{f}_X(c_l)$.
Another way is to specify a parametric distribution and estimate its parameters. To keep the analysis general, assume 
\begin{gather}
\ha{\omega}^d_j - \omega^d_j = \sum_{i=1}^n \eta_{nj}(Z_i) + o_P(r_n^{-1/2})
\label{eq:estw:formw}
\end{gather}
 for every $j$, where  $\eta_{nj}(Z_i)$ has zero mean and finite variance for each $n$ and $j$, and $r_n$ is a sequence that converges to infinity.
The exact forms of the function $\eta_{nj}(Z_i)$ and $r_n$ depend on the type of estimator used to obtain $\ha{\omega}^d_j$.
The sequence $r_n$ represents the rate at which the inverse of the variance of $\ha{\omega}^d_j$ grows.
Namely, $1/VAR[\sum_{i=1}^n \eta_{nj}(Z_i)] = O(r_n)$.
For example, if $\omega^d(\bc)$ is estimated parametrically by maximum likelihood, then $\eta_{nj}(Z_i)$ will be a function of the Hessian matrix times the score function,
and $r_n=n$;
if $\ha{\omega}^d_j$ is based of a kernel estimator for the density of $X_i$ with bandwidth $h_{\omega}$,
then $\eta_{nj}(Z_i)= (nh_{\omega})^{-1} k((X_i-c_j)/h_{\omega} )/ \sum_{l=1}^K f_X(c_l)$ and $r_n=n h_{\omega}$.

The second step consists of  estimating  $\mu^d$,
\begin{gather}
\ha{\mu}^d  = \sum_{j=1}^K \ha{\omega}^d_j  \ha{B}_j.
\end{gather}

Rewrite $\ha{\mu}^d$ as
\begin{align}
\ha{\mu}^d  - {\mu}^d =  \sum_{j=1}^K \omega^d_j ( \ha{B}_j - B_j )
+  \sum_{j=1}^K B_j ( \ha{\omega}^d_j  - \omega^d_j)
+  \sum_{j=1}^K ( \ha{B}_j - B_j) ( \ha{\omega}^d_j  - \omega^d_j).
\label{eq:estw:eq1}
\end{align}

Suppose $\ha{B}_j$ has no first-order asymptotic bias (i.e. bias-corrected).
The proofs of  Lemma \ref{lemma_porter} and Theorem \ref{theo_srd_kinf_est_int}  imply that
\begin{gather}
\sum_{j=1}^K \omega^d_j ( \ha{B}_j - B_j ) = \sum_{i=1}^n \varphi_{n}(Z_i) + o_P\left( (n \overline{h}_1 )^{-1/2} \right),
\label{eq:estw:eq2}
\end{gather} 
where $1/VAR[\sum_{i=1}^n \varphi_{n}(Z_i) ] = O \left( n \overline{h}_1  \right) $, and $ n \overline{h}_1 \to \infty$.
Similarly, the sum across $j$ of   \eqref{eq:estw:formw}  times $B_j$ gives
\begin{align}
\sum_{j=1}^K B_j ( \ha{\omega}^d_j  - \omega^d_j) = &  \sum_{i=1}^n \underset{\equiv  \eta_{n}(Z_i) }{\underbrace{ \sum_{j=1}^K B_j \eta_{nj}(Z_i) }}+ o_P(r_n^{-1/2})
\\
= & \sum_{i=1}^n  \eta_{n}(Z_i) + o_P(r_n^{-1/2}),
\label{eq:estw:formw2}
\end{align}
where $1/VAR[\sum_{i=1}^n \eta_{n}(Z_i)] = O(r_n)$.

Next, substitute \eqref{eq:estw:eq2} and \eqref{eq:estw:formw2} into Equation \ref{eq:estw:eq1}, 
\begin{align}
\ha{\mu}^d  - {\mu}^d = &  \sum_{i=1}^n \varphi_{n}(Z_i) 
 + \sum_{i=1}^n \eta_{n}(Z_i) 
 + o_P(r_n^{-1/2})
+ o_P\left( (n \overline{h}_1 )^{-1/2} \right)
 + O_P\left( (n \overline{h}_1 )^{-1/2} r_n^{-1/2} \right)
 \\
 = & \sum_{i=1}^n  \left\{ \varphi_{n}(Z_i)  + \eta_{n}(Z_i)  \right\} + o_P(r_n^{-1/2})
+ o_P\left( (n \overline{h}_1 )^{-1/2} \right)
\label{eq:estw_eq3}
\end{align}
where the second equality relies on $\ha{B}_j - B_j = O_P\left( (n \overline{h}_1 )^{-1/2} \right)$, $\ha{\omega}^d_j - \omega^d_j  = O_P\left( r_n^{-1/2} \right)$,
and on the fact that 
$(n \overline{h}_1 )^{-1/2} r_n^{-1/2} $ converges to zero faster than each of $(n \overline{h}_1 )^{-1/2}$ and $r_n^{-1/2}$.

Define $\m{V}^{\omega}_{n} = VAR[\sum_{i=1}^n \varphi_{n}(Z_i) + \eta_{n}(Z_i) ] $,
and note that
$\left( \m{V}^{\omega}_{n} \right)^{-1/2}  = O \left( \max\{ r_n^{-1}, (n \overline{h}_1 )^{-1} \}^{-1/2} \right) = 
O \left( \min\{ r_n^{1/2}, (n \overline{h}_1 )^{1/2} \} \right)$.

Then,
\begin{align}
\left( \m{V}^{\omega}_{n} \right)^{-1/2} \left( \ha{\mu}^d  - {\mu}^d \right)= &  
\left( \m{V}^{\omega}_{n} \right)^{-1/2} \sum_{i=1}^n  \left\{ \varphi_{n}(Z_i) +  \eta_{n}(Z_i)  \right\}  
\\
+ & O \left( \min\{ r_n^{1/2}, (n \overline{h}_1 )^{1/2} \} \right) o_P\left( r_n^{-1/2} \right) 
\\
+ & O \left( \min\{ r_n^{1/2}, (n \overline{h}_1 )^{1/2} \} \right) o_P\left(  (n \overline{h}_1 )^{-1/2} \right)
\\
= & \left( \m{V}^{\omega}_{n} \right)^{-1/2} \sum_{i=1}^n  \left\{ \varphi_{n}(Z_i) +  \eta_{n}(Z_i)  \right\}   + o_P(1)
\\
\dto & N(0,1).
\end{align}

A consistent estimator for the variance is:
\begin{align}
\ha{\m{V}}^{\omega}_{n} = \sum_{i=1}^n  \left\{ \ha{\varphi}_{n}(Z_i) +  \ha{\eta}_{n}(Z_i) \right\}^2,
\end{align}
with $\ha{\varphi}_{n}(Z_i)$ constructed as in  Equation \eqref{est:var:d}, Section \ref{sec_case1},
\begin{align}
\ha{\varphi}_{n}(Z_i) = &
	\ha{\eps}_i 
		\sum_{j=1}^K 
			\frac{\ha{\omega}_j^d}{n h_{1j}}
			k\left( \frac{X_i - c_{j}}{h_{1j}} \right)
			e_1'
			\left(
				v_i^{j +}  G_n^{j +}   
				-
				v_i^{j -} G_n^{j -}  
			\right)
			\ti{H}_{i}^{j},				
\end{align}
and the formula for $\ha{\eta}_{n}(Z_i)$ depends on the form of the estimator $\omega^d$.
In the kernel density example,
\begin{align}
\ha{\eta}_{n}(Z_i) = \frac{ \sum_{j=1}^K \ha{B_j} ( nh_{\omega} )^{-1} k( (X_i-c_j)/h_{\omega} ) }
{ \sum_{l=1}^K ( nh_{\omega} )^{-1} \sum_{m=1}^n k( (X_m-c_l)/h_{\omega} ) }.
\end{align}

An interesting particular case occurs when 
$ r_n$ grows faster than $ n \overline{h}_1 $.
This is the case if ${\omega}^d(\bc)$ is assumed to be in a parametric class; 
or 
if $\ha{\omega}^d(\bc)$ is based of a kernel density estimator with a bandwidth that converges to zero more slowly than $\overline{h}_1$.
Let $\m{V}^{d}_{n} = VAR[\sum_{i=1}^n \varphi_{n}(Z_i)  ] $ as defined in Theorem \ref{theo_srd_est_avg}.
It follows that,
\begin{align}
 \left( \m{V}^{d}_{n} \right)^{-1/2} \left( \ha{\mu}^d  - {\mu}^d \right)= &  
\left( \m{V}^{d}_{n} \right)^{-1/2} \sum_{i=1}^n  \left\{ \varphi_{n}(Z_i) +  \eta_{n}(Z_i)  \right\}  
\\
+ & o_P\left((n \overline{h}_1)^{1/2} r_n^{-1/2} \right) + o_P\left((n \overline{h}_1)^{1/2} (n \overline{h}_1 )^{-1/2} \right)
\\
= & \left( \m{V}^{d}_{n} \right)^{-1/2} \sum_{i=1}^n  \left\{ \varphi_{n}(Z_i) +  \eta_{n}(Z_i)  \right\}   + o_P(1)
\\
\dto & N(0,1),
\end{align}
where $\left( \m{V}^{d}_{n} \right)^{-1/2} = O\left( (n \overline{h}_1)^{1/2} \right)$ from Theorem \ref{theo_srd_est_avg},
and 
  $\left( \m{V}^{d}_{n} \right)^{-1/2} VAR[\sum_{i=1}^n \varphi_{n}(Z_i) + \eta_{n}(Z_i) ] \to 1$.
Therefore, when $\omega^d(\bc)$ is estimated at a faster rate than
$\beta(\bc)$, 
the asymptotic distribution and variance estimator provided in Section \ref{sec_case1} remain valid.

\subsection{Monte Carlo Simulations with Data-driven Bandwidths}
\label{sec:supp:datadrivenh}

\indent 

This section revisits the simulations in Section \ref{sec_simul} with data-driven bandwidth choices.
Both first and second-step bandwidths follow the rules for practical implementation suggested in Section \ref{sec_case2} 
(refer to page \pageref{parag:bandw:choice}, paragraph starting with 
``\textit{A simple recommendation to implement Theorem \ref{theo_srd_kinf_est_int}} '').
The rest of the simulation design remains the same as that of Section \ref{sec_simul}.

Table \ref{tab:datadrivenh:precision} compares the estimation precision of $\ha\mu$ and $\ha\mu^{bc}$
across five sample sizes $n$, with respective numbers of cutoffs $K$.
Table  \ref{tab:datadrivenh:intervals} analyzes coverage of confidence intervals.
Overall, the finite sample properties are consistent with those of Section \ref{sec_simul}, when bandwidths are non-random.
The randomness of bandwidths increases the variance and bias, but they decrease with $n$ at approximately the same rate as before.
Bias correction eliminates most of the bias, and produces confidence intervals with correct finite sample coverage.

\begin{table}[H]
\begin{center}
\caption{Precision of Estimators}
\label{tab:datadrivenh:precision}
\begin{tabular}{cc|cc|cc|cc}
\hline \hline 
& & \multicolumn{2}{|c}{Bias} & \multicolumn{2}{|c}{Variance}  & \multicolumn{2}{|c}{MSE} \\ 
$n$ & $K$ &  $\ha{\mu}$ &  $\ha{\mu}^{bc}$ &  
                     $\ha{\mu}$  &  $\ha{\mu}^{bc}$ &  
                     $\ha{\mu}$  &  $\ha{\mu}^{bc}$   \\ 
1789 & 20 
& 0.1056 & 0.0042  
& 0.0266 & 0.0551  
& 0.0377 & 0.0551  
\\ 
10120 & 40 
& 0.0364 & 0.0003  
& 0.0048 & 0.0084  
& 0.0061 & 0.0084  
\\ 
27886 & 60 
& 0.0190 & -0.0008  
& 0.0019 & 0.0032  
& 0.0023 & 0.0032  
\\ 
57244 & 80 
& 0.0132 & 0.0001  
& 0.0010 & 0.0017  
& 0.0012 & 0.0017  
\\ 
100000 & 100 
& 0.0092 & -0.0004  
& 0.0006 & 0.0010  
& 0.0007 & 0.0010  
\\ 
\hline \hline 
\end{tabular}

\caption*{\small
Notes: The table reports simulated bias, variance, and mean squared error (MSE) for two estimators
$(\ha{\mu}, \ha{\mu}^{bc})$, and five sample sizes $n$, with respective numbers of cutoffs $K$.
Following Section \ref{sec_case2}, the first-step bandwidths are picked by the IK algorithm, and adjusted to be of order $1/K$.
The second-step bandwidth is chosen on the grid $h_2 \in \{3/(K+1), \ldots, 12/(K+1)\}$
to minimize the estimated MSE of $\ha{\mu}$.
The number of simulations is 10,000.
}
\end{center}
\end{table}

\begin{table}[H]
\begin{center}
\caption{Coverage of 95\% Confidence Intervals}
\label{tab:datadrivenh:intervals}
\begin{tabular}{cc|cc|cc} 
\hline \hline 
& & \multicolumn{2}{|c}{\% Coverage} & \multicolumn{2}{|c}{ Avg. Length}   
\\ 
$n$ & $K$  
        & $\ha{\mu}$ &  $\ha{\mu}^{bc}$ 
        & $\ha{\mu}$ &  $\ha{\mu}^{bc}$ 
\\ 
1789 & 20 
& 0.8542 & 0.9576  
& 0.5528 & 0.9453   
\\ 
10120 & 40 
& 0.8750 & 0.9544  
& 0.2383 & 0.3697   
\\ 
27886 & 60 
& 0.8834 & 0.9528  
& 0.1499 & 0.2272   
\\ 
57244 & 80 
& 0.8821 & 0.9488  
& 0.1084 & 0.1625   
\\ 
100000 & 100 
& 0.8906 & 0.9543  
& 0.0845 & 0.1259   
\\ 
\hline \hline 
\end{tabular}

\caption*{\small
Notes: The table reports simulated percentage of correct coverage, and average length of 95\% confidence intervals.
Confidence intervals are constructed 
using two estimators $(\ha{\mu}, \ha{\mu}^{bc})$.
They equal an estimator plus or minus its estimated standard deviation multiplied by $1.96$.
Coverage and average length are computed for five sample sizes $n$ and respective numbers of cutoffs $K$.
Following Section \ref{sec_case2}, the first-step bandwidths are picked by the IK algorithm, and adjusted to be of order $1/K$.
The second-step bandwidth is chosen on the grid $h_2 \in \{3/(K+1), \ldots, 12/(K+1)\}$
to minimize the estimated MSE of $\ha{\mu}$.
The number of simulations is 10,000.
}
\end{center}
\end{table}


\end{document}